\documentclass{lmcs}
\pdfoutput=1
\usepackage[utf8]{inputenc}

\usepackage{lastpage}
\lmcsdoi{19}{4}{14}
\lmcsheading{}{\pageref{LastPage}}{}{}%
{Dec.~05,~2022}{Nov.~27,~2023}{}


\usepackage{xspace}
\usepackage{stmaryrd, amsmath, amssymb}
\usepackage{nicefrac}
\usepackage{paralist}
\usepackage{multirow}
\usepackage{xcolor}
\usepackage[colorlinks]{hyperref}
\usepackage{tikz}
	\usetikzlibrary{trees,decorations,arrows,automata,positioning,plotmarks,
	shapes,backgrounds}

\usepackage{FTMPST}

\begin{document}

%%%%%%%%%%%%%%%%%%%%%%%
%  Title Information  %
%%%%%%%%%%%%%%%%%%%%%%%

\title{FTMPST: Fault-Tolerant Multiparty Session Types}

\author[K.~Peters]{Kirstin Peters\lmcsorcid{0000-0002-4281-0074}}[a]
\author[U.~Nestmann]{Uwe Nestmann\lmcsorcid{0000−0002−8520−5448}}[b]
\author[C.~Wagner]{Christoph Wagner}[c]

\address{Augsburg University, Germany}
\email{kirstin.peters@uni-a.de}

\address{TU Berlin, Germany}
\email{uwe.nestmann@tu-berlin.de}

\address{Paessler AG, Germany}

%%%%%%%%%%%%%%
%  Abstract  %
%%%%%%%%%%%%%%

\begin{abstract}
	Multiparty session types are designed to abstractly capture the structure of communication protocols and verify behavioural properties. One important such property is progress, \ie the absence of deadlock. Distributed algorithms often resemble multiparty communication protocols. But proving their properties, in particular termination that is closely related to progress, can be elaborate. Since distributed algorithms are often designed to cope with faults, a first step towards using session types to verify distributed algorithms is to integrate fault-tolerance.

	We extend multiparty session types to cope with system failures such as unreliable communication and process crashes. Moreover, we augment the semantics of processes by failure patterns that can be used to represent system requirements (as, \eg, failure detectors). To illustrate our approach we analyse a variant of the well-known rotating coordinator algorithm by Chandra and Toueg.
\end{abstract}

\maketitle

%%%%%%%%%%%%%%%%%%
%  Introduction  %
%%%%%%%%%%%%%%%%%%

\section{Introduction}
\label{sec:introduction}

Multi-Party Session Types (\MPST) are used to statically ensure correctly coordinated behaviour in systems without global control (\cite{hondaYoshidaCarbone16,gentleintro}).
One important such property is progress, \ie the absence of deadlock. Like with every other static typing approach, the main advantage is their efficiency, \ie they avoid the problem of state space explosion.
\MPST are designed to abstractly capture the structure of communication protocols.
They describe global behaviours as \emph{sessions}, \ie units of
conversations \cite{hondaYoshidaCarbone16,BettiniAtall08,BocciAtall10}. The participants of such sessions are called \emph{roles}.
\emph{Global types} specify protocols from a global point of view.
These types are used to reason about processes formulated
in a \emph{session calculus}.
%Most of the existing session calculi are variants of the $ \pi $-calculus \cite{milnerParrowWalker92}.

Distributed algorithms (DA) very much resemble multiparty communication protocols.
An essential behavioural property of DA is termination \cite{Tel94,Lynch96}, despite failures, but it is often elaborate to prove.
It turns out that progress (as provided by MPST) and termination (as required by DA) are closely related.

Many DA were designed in a fault-tolerant way, in order to work in environments, where they have to cope with system failures---be it links dropping messages or processes crashing.
G\"artner \cite{DBLP:journals/csur/Gartner99} suggested four different forms of fault-tolerance, depending on whether the safety and liveness requirements are met, or not. An algorithm is called \emph{masking} in the (best) case that both properties hold while tolerating faults transparently, \ie without further intervention by the programmer. It is called \emph{non-masking}, however, if faults are dealt with explicitly in order to cope with unsafe states, while still guaranteeing liveness. The \emph{fail-safe} case then captures algorithms that remain safe, but not live. (The fourth form is just there for completeness; here neither safety nor liveness is guaranteed.)
We focus on masking fault-tolerant algorithms.

While the detection of conceptual design errors is a standard property of type systems, proving correctness of algorithms despite the occurrence of system failures is not.
Likewise, traditional \MPST do not cover fault tolerance or failure handling.
There are several approaches to integrate explicit failure handling in \MPST (\eg \cite{CarboneHondaYoshida08,capecchi2016,chen.viering.bejleri.ziarek.eugster,viering18,dem15,adameitPetersNestmann17}).
These approaches are sometimes enhanced with recovery mechanisms such as \cite{castellaniEtAl} or even provide algorithms to help find safe states to recover from as in \cite{neykova2017let}.
Many of these approaches introduce nested \textsc{try}-and-\textsc{catch}-blocks and a challenge is to ensure that all participants are consistently informed about concurrent \textsc{throws} of exceptions.
Therefore, exceptions are propagated within the system.
Though explicit failure handling makes sense for high-level applications, the required message overhead is too inefficient for many low-level algorithms.
Instead these low-level algorithms are often designed to tolerate a certain amount of failures.
Since we focus on the communication structure of systems, additional messages as reaction to faults (\eg to propagate faults) are considered \emph{non-masking} failure handling.
In contrast, we expect masking fault-tolerant algorithms to cope without messages triggered by faults.
We study how much unhandled failures a well-typed system can tolerate, while maintaining the typical properties of \MPST.

We propose a variant of \MPST with unreliable interactions and augment the semantics to also represent failures such as message loss and crashing processes, as well as more abstract concepts of fault-tolerant algorithms such as the possibility to suspect a process to be faulty.
To guide the behaviour of unreliable communication, the semantics of processes uses failure patterns that are not defined but \emph{could} be instantiated by an application. This allows us to cover requirements on the system---as, \eg, a bound on the number of faulty processes---as well as more abstract concepts like failure detectors.
It is beyond the scope of this paper to discuss \emph{how} failure patterns could be implemented.

\subsection{Related Work}

Type systems are usually designed for failure-free scenarios.
An exception is \cite{KouzapasGutkovasGay14} that introduces unreliable broadcast, where a transmission can be received by multiple receivers but not necessarily all available receivers. In the latter case, the receiver is deadlocked. In contrast, we consider fault-tolerant interactions, where in the case of a failure the receiver is \emph{not} deadlocked.

The already mentioned systems in \cite{CarboneHondaYoshida08,capecchi2016,chen.viering.bejleri.ziarek.eugster,viering18,dem15} extend session types with exceptions thrown by processes within \textsc{try}-and-\textsc{catch}-blocks, interrupts, or similar syntax.
They structurally and semantically encapsulate an unreliable part of a protocol and provide some means to 'detect' a failure and 'react' to it.
For example \cite{viering18} proposes a variant of \MPST with the explicit handling of crash failures.
Therefore they coordinate asynchronous messages for run-time crash notifications using a coordinator.
Processes in \cite{viering18} have access to local failure detectors which eventually detect all failed peers and do not falsely suspect peers.
In contrast we augment the semantics of the session calculus with failure patterns that \eg allow to implement failure detectors but may also be used to implement system requirements.
Exceptions may also describe, why a failure occurred.
Here we deliberately do not model causes of failures or how to 'detect' a failure.
Different system architectures might provide different mechanisms to do so, for example, by means of time-outs.
As is standard for the analysis of DA, our approach allows us to port the verified algorithms on different system architectures that satisfy the necessary system requirements.

Another essential difference is how systems react to faults.
In~\cite{capecchi2016}, \textsc{throw}-messages are propagated among nested \textsc{try}-and-\textsc{catch}-blocks to ensure that all participants are \emph{consistently} informed about concurrent \textsc{throws} of exceptions.
Fault-tolerant DA, however, have to deal with the problem of inconsistency that some part of a system may consider a process/link as crashed, while at the same time the same process/link is regarded as correct by another part.
(This is one of the most challenging problems in the design and verification of \emph{fault-tolerant} DA.)
The reason is that \emph{distributed} processes usually cannot reliably
observe an error on another system part, unless they are informed by some
system ``device'' (like the ``coordinator'' of~\cite{viering18} or the ``oracle'' of~\cite{capecchi2016}).
Therefore, abstractions like unreliable failure detectors are used to model this restricted observability which can, for example, be implemented by time-outs.
Failure detectors are often considered to be local (see previous paragraph), but they cannot ensure global consistency.
Various degrees of consistency, or [un]reliability, of failure detectors are often defined via constraints that are expressed as global temporal properties~\cite{ChandraToueg96} (see also Section~\ref{sec:example}).
Abstract properties, like the communication predicates in the Heard-Of model~\cite{DBLP:journals/dc/Charron-BostS09}, can also be used to specify minimum requirements on system behaviours at a global level of abstraction in order to be able to guarantee correctness properties.
%In a sense, our use of failure patterns follows this tradition of identifying levels of abstraction that best fit an application or verification scenario in order to focus on other details, such as the session type machinery presented in this paper.

In previous work, we used the above-mentioned failure detector abstractions in the context of (untyped) process calculi~\cite{DBLP:conf/asian/NestmannF03} to verify properties of several algorithms for Distributed Consensus~\cite{DBLP:conf/concur/NestmannFM03,KuhnrichNestmann09}.
Key for the respective proofs was the intricate reconstruction of global state information from process calculus terms, as we later on formalized in~\cite{wagnerNestmann14}.
We conjecture that MPST could provide proof support in this context, for example, for the methods that apply to these global states.
The work by Francalanza and Hennessy~\cite{DBLP:conf/esop/FrancalanzaH07} also uses a process calculus for the analysis of DA, but employs bisimulation proof techniques.
In order to do so, however, the intended properties need to be formulated via some global wrapper code, which provides a level of indirection to proofs.
%Double-Blind:And like our own work~\cite{KuhnrichNestmann09}, where we adapted their approach, it suffers from the absence of clear global state information.
This approach suffers from the absence of clear global state information.
In contrast, MPST supply useful global (session type) information from scratch.

\subsection{Summary}

The present paper is an extended version of \cite{petersNestmannWagner22} that additionally contains the proofs of the presented results as well as some additional explanations (also see the technical report \cite{petersNestmannWagnerTR22}).
In Section~\ref{sec:faultolerance} we give an impression of the forms of fault-tolerant interactions that we consider.
Section~\ref{sec:syntax} introduces the syntax of our version of multiparty session types.
The semantics of the session calculus is given in Section~\ref{sec:failurePatterns}.
In Section~\ref{sec:typing} we provide the typing rules and show that the standard properties that are usually required for multiparty session types versions are valid in our case.
Section~\ref{sec:example} provides an example of using fault-tolerant multiparty session types by analysing an implementation of a well-known Consensus algorithm.

%%%%%%%%%%%%%%%%%%%%%%%%%%%%%%%%%%%%%%%%%%%%%%%%
%  Fault-Tolerance and Distributed Algorithms  %
%%%%%%%%%%%%%%%%%%%%%%%%%%%%%%%%%%%%%%%%%%%%%%%%

\section{Fault-Tolerance in Distributed Algorithms}
\label{sec:faultolerance}

We consider three sources of failure in an \unrel communication (Figure~\ref{fig:unrelComWeakRBran}(a)):
(1) the sender may crash before it releases the message,
(2) the receiver may crash before it can consume the message, or
(3) the communication medium may lose the message.
The design of a DA may allow it to handle some kinds of failures better than others.
Failures are unpredictable events that occur at runtime.
Since types consider only static and predictable information, we do not distinguish between different kinds of failure or model their source in types.
Instead we allow types, \ie the specifications of systems, to distinguish between potentially faulty and reliable interactions.

\begin{figure}[t]
	\centering
	\scalebox{0.9}{
	\begin{tikzpicture}[auto]
		% a
		\node at (-0.5, 1.5) {(a)};
		\node[state, minimum size=5mm] at (2.5, 1.5) (r11) {$ \Role_1 $};
		\node[state, minimum size=5mm] at (4, 1.5) (r12) {$ \Role_2 $};
		\node[state, minimum size=5mm] at (0, 0) (r111) {$ \Role_1 $};
		\node[state, minimum size=5mm] at (1.5, 0) (r112) {$ \Role_2 $};
		\node[state, minimum size=5mm] at (3, 0) (r121) {$ \Role_1 $};
		\node[state, minimum size=5mm] at (4.5, 0) (r122) {$ \Role_2 $};
		\node[state, minimum size=5mm] at (6, 0) (r131) {$ \Role_1 $};
		\node[state, minimum size=5mm] at (7.5, 0) (r132) {$ \Role_2 $};
		\path[->] (r11) edge node{$ \Label{\left<\Expr[v]\right>} $} (r12);
		\draw[-, thick, color=red] (-0.4, -0.4) -- (0.4, 0.4);
		\draw[-, thick, color=red] (-0.4, 0.4) -- (0.4, -0.4);
		\path[-] (r121) edge node{$ \Label{\left<\Expr[v]\right>} $} (4, 0);
		\draw[-, thick, color=red] (4.1, -0.4) -- (4.9, 0.4);
		\draw[-, thick, color=red] (4.1, 0.4) -- (4.9, -0.4);
		\path[-] (r131) edge node{$ \Label{\left<\Expr[v]\right>} $} (7, 0);
		\draw[-, thick, color=red] (6.45, 0.1) -- (6.85, 0.5);
		\draw[-, thick, color=red] (6.45, 0.5) -- (6.85, 0.1);
		\path[->, dotted, color = blue] (3.25, 1.2) edge node[above, near end]{(1)} (1, 0.4);
		\path[->, dotted, color = blue] (3.25, 1.2) edge node[left] {(2)} (3.25, 0.4);
		\path[->, dotted, color = blue] (3.25, 1.2) edge node[above, near end] {(3)} (5.5, 0.4);
		\draw[-, dashed, color=blue] (2.25, -0.5) -- (2.25, 0.5);
		\draw[-, dashed, color=blue] (5.25, -0.5) -- (5.25, 0.5);
		% b
		\node at (9, 1.5) {(b)};
		\node[state, minimum size=5mm] at (11, 1.5) (r21) {$ \Role_1 $};
		\node[state, minimum size=5mm] at (12.5, 1.5) (r22) {$ \Role_2 $};
		\node[state, minimum size=5mm] at (9.5, 0) (r211) {$ \Role_1 $};
		\node[state, minimum size=5mm] at (11, 0) (r212) {$ \Role_2 $};
		\node[state, minimum size=5mm] at (12.5, 0) (r221) {$ \Role_1 $};
		\node[state, minimum size=5mm] at (14, 0) (r222) {$ \Role_2 $};
		\path[->] (r21) edge node{$ \Label $} (r22);
		\draw[-, thick, color=red] (9.1, -0.4) -- (9.9, 0.4);
		\draw[-, thick, color=red] (9.1, 0.4) -- (9.9, -0.4);
		\path[-] (r221) edge node{$ \Label $} (13.5, 0);
		\draw[-, thick, color=red] (13.6, -0.4) -- (14.4, 0.4);
		\draw[-, thick, color=red] (13.6, 0.4) -- (14.4, -0.4);
		\path[->, dotted, color = blue] (11.75, 1.2) edge node[left]{(1)} (10.5, 0.4);
		\path[->, dotted, color = blue] (11.75, 1.2) edge node[right] {(2)} (12.75, 0.4);
		\draw[-, dashed, color=blue] (11.75, -0.5) -- (11.75, 0.5);
	\end{tikzpicture}}
	\caption{\Unrel Communication (a) and \WeakR Branching (b).}
	\label{fig:unrelComWeakRBran}
\end{figure}

A fault-tolerant algorithm has to solve its task despite such failures.
Remember that \MPST analyse the communication structure.
Accordingly, we need a mechanism to tolerate faults in the communication structure.
We want our type system to ensure that a faulty interaction neither blocks the overall protocol nor influences the communication structure of the system after this fault.
We consider an \unrel communication as fault-tolerant if a failure does not influence the guarantees for the overall communication structure except for this particular communication.
Moreover, if a potentially unreliable communication is executed successfully, then our type system ensures the same guarantees as for reliable communication such as \eg the absence of communication mismatches.

To ensure that a failure does not block the algorithm, both the receiver and the sender need to be allowed to proceed without their \unrel communication partner. Therefore, the receiver of an \unrel communication is required to specify a default value that, in the case of failure, is used instead of the value the process was supposed to receive.
The type system ensures the existence of such default values and checks their sort.

Moreover, we augment unreliable communication with labels that help us to avoid communication mismatches.
Assume for instance two subsequent \unrel communications in that values of different sorts, a natural number and a boolean, are transmitted.
If the first message with its natural number is lost but the second message containing a Boolean value is transmitted, the receiver could wrongly receive a Boolean value although it still waits for a natural number.
To avoid this mismatch, we add a label to \unrel communication, ensure (by the typing rules) that the same label is never associated with different types, and let the semantics inspect the label of a message before reception.
Note that this problem, \ie how to ensure the absence of communication mismatches in the case of \unrel communication, is one of the main challenges in structuring fault-tolerant communication.

Branching in the context of failures is more difficult, because a branch marks a decision point in a specification, \ie the participants of the session are supposed to behave differently \wrt this decision.
In an \unrel setting it is difficult to ensure that all participants are informed consistently about such a decision.% and adapt their behaviour accordingly.

Consider a reliable branching that is decided by a process $ \Role_1 $ and transmitted to $ \Role_2 $.
If we try to execute such a branching despite failures, we observe that there are again three ways in that this branching can go wrong (Figure~\ref{fig:unrelComWeakRBran}(b)):
(1) The sender may crash before it releases its decision.
This will block $ \Role_2 $, because it is missing the information about the branch it should move to.
(2) The receiver might crash.
(3) The message of $ \Role_1 $ is lost.
Then again $ \Role_2 $ is blocked.

Case~(2) can be dealt with similar to \unrel communication, \ie by marking the branching as potentially faulty and by ensuring that a crash of $ \Role_2 $ will not block another process.
To deal with Case~(1), we declare one of the offered branches as default. If $ \Role_1 $ has crashed, $ \Role_2 $ moves to the default branch.
Then $ \Role_2 $ will not necessarily move to the branch that $ \Role_1 $ had in mind before it crashed, but to a valid/specified branch and, since $ \Role_1 $ is crashed, no two processes move to different branches.
The main problem is in Case~(3).
Let $ \Role_1 $ move to a non-default branch and transmit its decision to $ \Role_2 $, this message gets lost, and $ \Role_2 $ moves to the default branch.
Now both processes did move to branches that are described by their types; but they are in different branches.
This case violates the specification in the type and we want to reject it.
More precisely, we consider three levels of failures in interactions:
\begin{description}
	\item[\StrongR ($ \iR $)] Neither the sender nor the receiver can crash as long as they are involved in this interaction. The message cannot be lost by the communication medium. This form corresponds to reliable communication as it was described in \cite{aguileraChenToueg97} in the context of distributed algorithms.
		This is the standard, failure-free case.
	\item[\WeakR ($ \iW $)] Both the sender and the receiver might crash at every possible point during this interaction. But the communication medium cannot lose the message.
	\item[\Unrel ($ \iU $)] Both the sender and the receiver might crash at every possible point during this interaction and the communication medium might lose the message. There are no guarantees that this interaction---or any part of it---takes place.
		Here, it is difficult to ensure interesting properties in branching.
\end{description}
We use the subscripts or superscripts $ \iR $, $ \iW $, or $ \iU $ to indicate actions of the respective kind.

%%%%%%%%%%%%%%%%%%%%%%%%%%%%%%%%%%%%%%%%
%  Fault-Tolerant Types and Processes  %
%%%%%%%%%%%%%%%%%%%%%%%%%%%%%%%%%%%%%%%%

\section{Fault-Tolerant Types and Processes}
\label{sec:syntax}

For clarity, we often distinguish names into \emph{values}, \ie the payload of messages, \emph{shared channels}, or \emph{session channels} according to their usage; there is, however, no need to formally distinguish between different kinds of names.

We assume that the sets $ \names $ of names $ \Chan[a], \Chan, \Args \ldots $; $ \roles $ of roles $ \Role[n], \Role, \ldots $; $ \labels $ of labels $ \Label, \LabelD, \ldots $; $ \typeVars $ of type variables $ \TypeV $; and $ \procVars $ of process variables $ \ProcV $ are pairwise distinct.
To simplify the reduction semantics of our session calculus, we use natural numbers as roles (compare to \cite{hondaYoshidaCarbone16}).
Sorts $ \Sort $ range over $ \mathbb{B}, \mathbb{N}, \ldots $.
The set $ \expressions $ of expressions $ \Expr, \Expr[v], \Expr[b], \ldots $ is constructed from the standard Boolean operations, natural numbers, names, and (in)equalities.

Global types specify the desired communication structure from a global point of view.
In local types this global view is projected to the specification of a single role/participant.
We use standard \MPST (\cite{hondaYoshidaCarbone08,hondaYoshidaCarbone16}) extended by \unrel communication and \weakR branching (highlighted in blue) in Figure~\ref{fig:syntax}.

\begin{figure}[t]
	\centering
	\renewcommand{\tabcolsep}{1pt}
	\renewcommand{\arraystretch}{1.1}
	\begin{tabular}{|llclr|llclr|llclr|}
		\hline
		\multicolumn{5}{|c|}{Global Types} & \multicolumn{5}{c|}{Local Types} & \multicolumn{5}{c|}{Processes}\\
		&&&&& &&&&& & $ P $ & $ \deffTerms $ & $ \PReq{\Chan[a]}{\Role[n]}{\Chan}{P} $ &\\
		&&&&& &&&&& & & $ \sepTerms $ & $ \PAcc{\Chan[a]}{\Role}{\Chan}{P} $ &\\
		& \multirow{2}{*}{$ \GT $} & \multirow{2}{*}{$ \deffTerms $} & \multirow{2}{*}{$ \GComR{\Role_1}{\Role_2}{\Sort}{\GT} $} & & & $ \LT $ & $ \deffTerms $ & $ \LSendR{\Role_2}{\Sort}{\LT} $ & & & & $ \sepTerms $ & $ \PSendR{\Chan}{\Role_1}{\Role_2}{\Expr}{P} $ &\\
		&&&&& & & $ \sepTerms $ & $ \LGetR{\Role_1}{\Sort}{\LT} $ & & & & $ \sepTerms $ & $ \PGetR{\Chan}{\Role_2}{\Role_1}{\Args}{\PT} $ &\\
		& & \multirow{2}{*}{$ \sepTerms $} & \multirow{2}{*}{$ \textcolor{blue}{\GComU{\Role_1}{\Role_2}{\Label}{\Sort}{\GT}} $} & & & & $ \sepTerms $ & $ \textcolor{blue}{\LSendU{\Role_2}{\Label}{\Sort}{\LT}} $ & & & & $ \sepTerms $ & $ \textcolor{blue}{\PSendU{\Chan}{\Role_1}{\Role_2}{\Label}{\Expr}{P}} $ &\\
		&&&&& & & $ \sepTerms $ & $ \textcolor{blue}{\LGetU{\Role_1}{\Label}{\Sort}{\LT}} $ & & & & $ \sepTerms $ & $ \textcolor{blue}{\PGetU{\Chan}{\Role_2}{\Role_1}{\Label}{\Expr[v]}{\Args}{P}} $ &\\
		& & \multirow{2}{*}{$ \sepTerms $} & \multirow{2}{*}{$ \GBranR{\Role_1}{\Role_2}{\Set{ \Label_i.\GT_i}_{i \in \indexSet}} $} & & & & $ \sepTerms $ & $ \LSelR{\Role_2}{\Set{ \Label_i.\LT_i }_{i \in \indexSet}} $ & & & & $ \sepTerms $ & $ \PSelR{\Chan}{\Role_1}{\Role_2}{\Label}{P} $ &\\
		&&&&& & & $ \sepTerms $ & $ \LBranR{\Role_1}{\Set{ \Label_i.\LT_i }_{i \in \indexSet}} $ & & & & $ \sepTerms $ & $ \PBranR{\Chan}{\Role_2}{\Role_1}{\Set{ \Label_i.P_i }_{i \in \indexSet}} $ &\\
		& & \multirow{2}{*}{$ \sepTerms $} & \multirow{2}{*}{$ \textcolor{blue}{\GBranW{\Role}{\Role[R]}{\Set{ \Label_i.\GT_i }_{i \in \indexSet, \LabelD}}} $} & & & & $ \sepTerms $ & $ \textcolor{blue}{\LSelW{\Role[R]}{\Set{ \Label_i.\LT_i }_{i \in \indexSet}}} $ & & & & $ \sepTerms $ & $ \textcolor{blue}{\PSelW{\Chan}{\Role}{\Role[R]}{\Label}{P}} $ &\\
		&&&&& & & $ \sepTerms $ & $ \textcolor{blue}{\LBranW{\Role}{\Set{ \Label_i.\LT_i }_{i \in \indexSet, \LabelD}}} $ & & & & $ \sepTerms $ & $ \textcolor{blue}{\PBranW{\Chan}{\Role_j}{\Role}{\Set{ \Label_i.P_i }_{i \in \indexSet, \LabelD}}} $ &\\
		& & $ \sepTerms $ & $ \GPar{\GT_1}{\GT_2} $ & & &&&&& & & $ \sepTerms $ & $ P_1 \mid P_2 $ &\\
		& & $ \sepTerms $ & $ \GRep{\TypeV}{\GT} \sepTerms \TypeV \sepTerms \GEnd $ & & & & $ \sepTerms $ & $ \LRep{\TypeV}{\LT} \sepTerms \TypeV \sepTerms \LEnd $ & & & & $ \sepTerms $ & $ \PRep{\ProcV}{P} \sepTerms \ProcV \sepTerms \PEnd $ &\\
		&&&&& &&&&& & & $ \sepTerms $ & $ \PITE{\Expr[b]}{P_1}{P_2} $ &\\
		&&&&& &&&&& & & $ \sepTerms $ & $ \PRes{\Args}{P} \sepTerms \PCrash $ &\\
		& & \multirow{2}{*}{$ \sepTerms $} & \multirow{2}{*}{$ \GDel{\Role_1}{\Role_2}{\Chan'}{\Role}{\LT}{\GT} $} & & & & $ \sepTerms $ & $ \LDelA{\Role_2}{\Chan'}{\Role}{\LT}{\LT'} $ & & & & $ \sepTerms $ & $ \PDelA{\Chan}{\Role_1}{\Role_2}{\AT{\Chan'}{\Role}}{\PT} $ &\\
		&&&&& & & $ \sepTerms $ & $ \LDelB{\Role_1}{\Chan'}{\Role}{\LT}{\LT'} $ && & & $ \sepTerms $ & $ \PDelB{\Chan}{\Role_2}{\Role_1}{\AT{\Chan'}{\Role}}{\PT} $ &\\
		&&&&& &&&&& & & $ \sepTerms $ & $ \MQ{\Chan}{\Role_1}{\Role_2}{\Queue} $ &\\
		\hline
		\multicolumn{10}{|c|}{Message Types} & \multicolumn{5}{c|}{Messages}\\
		\multicolumn{10}{|c|}{\multirow{2}{*}{$ \mathsf{mt} \deffTerms \MessR{\Sort} \sepTerms \textcolor{blue}{\MessU{\Label}{\Sort}} \sepTerms \MessBR{\Label} \sepTerms \textcolor{blue}{\MessBW{\Label}} \sepTerms \AT{\Chan}{\Role} $}} & \multicolumn{5}{c|}{$ \; \mathsf{m} \deffTerms \MessR{\Expr[v]} \sepTerms \textcolor{blue}{\MessU{\Label}{\Expr[v]}} \sepTerms \MessBR{\Label} $}\\
		\multicolumn{10}{|c|}{} &&& $ \sepTerms $ & $ \textcolor{blue}{\MessBW{\Label}} \sepTerms \AT{\Chan}{\Role} $ &\\
		\hline
	\end{tabular}
	\caption{Syntax of Fault-Tolerant \MPST{}}
	\label{fig:syntax}
\end{figure}

A new session $ \Chan $ with $ \Role[n] $ roles is initialised with $ \PReq{\Chan[a]}{\Role[n]}{\Chan}{P} $ and $ \PAcc{\Chan[a]}{\Role}{\Chan}{P} $ via the shared channel $ \Chan[a] $. We identify sessions with their unique session channel.

The type $ \GComR{\Role_1}{\Role_2}{\Sort}{\GT} $ specifies a \strongR communication from role $ \Role_1 $ to role $ \Role_2 $ to transmit a value of the sort $ \Sort $ and then continues with $ \GT $.
A system with this type will be guaranteed to perform a corresponding action.
In a session $ \Chan $ this communication is implemented by the sender $ \PSendR{\Chan}{\Role_1}{\Role_2}{\Expr}{\PT_1} $ (specified as $ \LSendR{\Role_2}{\Sort}{\LT_1} $) and the receiver $ \PGetR{\Chan}{\Role_2}{\Role_1}{\Args}{\PT_2} $ (specified as $ \LGetR{\Role_1}{\Sort}{\LT_2} $).
As result, the receiver instantiates $ \Args $ in its continuation $ \PT_2 $ with the received value.

The type $ \GComU{\Role_1}{\Role_2}{\Label}{\Sort}{\GT} $ specifies an \unrel communication from $ \Role_1 $ to $ \Role_2 $ transmitting (if successful) a label $ \Label $ and a value of type $ \Sort $ and then continues (regardless of the success of this communication) with $ \GT $.
The \unrel counterparts of senders and receivers are $ \PSendU{\Chan}{\Role_1}{\Role_2}{\Label}{\Expr}{\PT_1} $ (specified as $ \LSendU{\Role_2}{\Label}{\Sort}{\LT_1} $) and $ \PGetU{\Chan}{\Role_2}{\Role_1}{\Label}{\Args[v]}{\Args}{\PT_2} $ (specified as $ \LGetU{\Role_1}{\Label}{\Sort}{\LT_2} $).
The receiver $ \PGetU{\Chan}{\Role_2}{\Role_1}{\Label}{\Args[v]}{\Args}{\PT_2} $ declares a default value $ \Args[v] $ that is used instead of a received value to instantiate $ \Args $ after a failure.
Moreover, a label is communicated that helps us to ensure that a faulty \unrel communication has no influence on later actions.

The \strongR branching $ \GBranR{\Role_1}{\Role_2}{\Set{ \Label_i.\GT_i}_{i \in \indexSet}} $ allows $ \Role_1 $ to pick one of the branches offered by $ \Role_2 $.
We identify the branches with their respective label.
Selection of a branch is by $ \PSelR{\Chan}{\Role_1}{\Role_2}{\Label}{P} $ (specified as $ \LSelR{\Role_2}{\Set{ \Label_i.\LT_i }_{i \in \indexSet}} $).
Upon receiving $ \Label_j $, $ \PBranR{\Chan}{\Role_2}{\Role_1}{\Set{ \Label_i.P_i }_{i \in \indexSet}} $ (specified as $ \LBranR{\Role_1}{\Set{ \Label_i.\LT_i }_{i \in \indexSet}} $) continues with $ \PT_j $.

As discussed in the end of Section~\ref{sec:introduction}, the counterpart of branching is \weakR and not \unrel.
It is implemented by $ \GBranW{\Role}{\Role[R]}{\Set{\Label_i.\GT_i}_{i \in \indexSet, \LabelD}} $, where $ \Role[R] \subseteq \roles $ and $ \LabelD $ with $ \default \in \indexSet $ is the default branch.
We use a broadcast from $ \Role $ to all roles in $ \Role[R] $ to ensure that the sender can influence several participants consistently.
Splitting this action to inform the roles in $ \Role[R] $ separately does not work, because we cannot ensure consistency if the sender crashes while performing these subsequent actions.
The type system will ensure that no message is lost.
Because of that, all processes that are not crashed will move to the same branch.
We often abbreviate branching \wrt to a small set of branches by omitting the set brackets and instead separating the branches by $ \oplus $, where the last branch is always the default branch.
In contrast to the \strongR cases, $ \PSelW{\Chan}{\Role}{\Role[R]}{\Label}{\PT} $ (specified as $ \LSelW{\Role[R]}{\Set{ \Label_i.\LT_i }_{i \in \indexSet}} $) allows to broadcast its decision to $ \Role[R] $ and $ \PBranW{\Chan}{\Role_j}{\Role}{\Set{ \Label_i.\PT_i }_{i \in \indexSet, \LabelD}} $ (specified as $ \LBranW{\Role}{\Set{ \Label_i.\LT_i }_{i \in \indexSet, \LabelD}} $) defines a default label $ \LabelD $.

The $ \PCrash $ denotes a process that crashed.
Similar to \cite{hondaYoshidaCarbone16}, we use message queues to implement asynchrony in sessions.
Therefore, session initialisation introduces a directed and initially empty message queue $ \MQ{\Chan[s]}{\Role_1}{\Role_2}{\emptyList} $ for each pair of roles $ \Role_1 \neq \Role_2 $ of the session $ \Chan[s] $.
The separate message queues ensure that messages with different sources or destinations are not ordered, but each message queue is FIFO.
Since the different forms of interaction might be implemented differently (\eg by TCP or UDP), it make sense to further split the message queues into three message queues for each pair $ \Role_1 \neq \Role_2 $ such that different kinds of messages do not need to be ordered.
To simplify the presentation of examples in this paper and not to blow up the number of message queues, we stick to a single message queue for each pair $ \Role_1 \neq \Role_2 $, but the correctness of our type system does not depend on this decision.
We have five kinds of messages $ \mathsf{m} $ and corresponding message types $ \mathsf{mt} $ in Figure~\ref{fig:syntax}---one for each kind of interaction.
In \strongR communication a value $ v $ (of sort $ \Sort $) is transmitted in a message $ \MessR{v} $ of type $ \MessR{\Sort} $.
In \unrel communication the message $ \MessU{\Label}{v} $ (of type $ \MessU{\Label}{\Sort} $) additionally carries a label $ \Label $.
For branching only the picked label $ \Label $ is transmitted and we add the kind of branching as superscript, \ie message/type $ \MessBR{\Label} $ is for \strongR branching and message/type $ \MessBW{\Label} $ for \weakR branching.
Finally, message/type $ \AT{\Chan}{\Role} $ is for session delegation.
A message queue $ \Queue $ is a list of messages $ \mathsf{m} $ and $ \MT $ is a list of message types $ \mathsf{mt} $.

The remaining operators for independence $ \GPar{\GT}{\GT'} $; parallel composition $ \PPar{\PT}{\PT'} $; recursion $ \GRep{\TypeV}{\GT} $, $ \PRep{\ProcV}{\PT} $; inaction $ \GEnd $, $ \PEnd $; conditionals $ \PITE{\Expr[b]}{\PT_1}{\PT_2} $; session delegation $ \GDel{\Role_1}{\Role_2}{\Chan'}{\Role}{\LT}{\GT} $, $ \PDelA{\Chan}{\Role_1}{\Role_2}{\AT{\Chan'}{\Role}}{\PT} $, $ \PDelB{\Chan}{\Role_2}{\Role_1}{\AT{\Chan'}{\Role}}{\PT} $; and restriction $ \PRes{\Args}{\PT} $ are all standard.
As usual, we assume that recursion variables are guarded and do not occur free in types or processes.

In types $ \GRep{\TypeV}{\GT} $ and $ \LRep{\TypeV}{\LT} $ the type variable $ \TypeV $ is \emph{bound}.
In processes $ \PRep{\ProcV}{\PT} $ the process variable $ \ProcV $ is bound.
Similarly, all names in round brackets are bound in the remainder of the respective process, \eg $ \Chan $ is bound in $ \PT $ by $ \PReq{\Chan[a]}{\Role[n]}{\Chan}{\PT} $ and $ \Args $ is bound in $ \PT $ by $ \PGetR{\Chan}{\Role_1}{\Role_2}{\Args}{\PT} $.
A variable or name is \emph{free} if it is not bound. Let $ \FreeNames{\PT} $ return the free names of $ \PT $.

Let \emph{subterm} denote a (type or process) expression that syntactically occurs within another (type or process) term.
We use '$ . $' (as \eg in $ \PReq{\Chan[a]}{\Role}{\Chan}{\PT} $) to denote sequential composition. In all operators the \emph{prefix} before '$ . $' guards the \emph{continuation} after the '$ . $'.
Let $ \prod_{1 \leq i \leq n} \PT_i $ abbreviate $ \PPar{\PT_1}{\PPar{\ldots}{\PT_n}} $.

Let $ \Roles{\GT} $ return all roles that occur in $ \GT $.
We write  $ \Unreliable{\GT} $, $ \Unreliable{\LT} $, and $ \Unreliable{\PT} $, if none of the prefixes in $ \GT $, $ \LT $, and $ \PT $ is \strongR or for delegation and if $ \PT $ does not contain message queues.

\begin{defi}[Well-Formedness, Global Type]
	A global type is \emph{well-formed} if
	\begin{compactenum}[(1)]
		\item it neither contains free nor unguarded type variables,
		\item $ \Roles{G} = \Set{ 1, \ldots, \Length{\Roles{G}} } $,
		\item for all its subterms of the form $ \GComR{\Role_1}{\Role_2}{\Sort}{\GT} $ or $ \GComU{\Role_1}{\Role_2}{\Label}{\Sort}{\GT} $, we have $ \Role_1 \neq \Role_2 $,
		\item for all its subterms of the form $ \GBranR{\Role_1}{\Role_2}{\Set{ \Label_i.\GT_i }_{i \in \indexSet}} $ or $ \GBranW{\Role}{\Role[R]}{\Set{ \Label_i.\GT_i }_{i \in \indexSet, \LabelD}} $, we have $ \Role_1 \neq \Role_2 $, $ \Role \notin \Role[R] $, $ \default \in \indexSet $, and the labels $ \Label_i $ are pairwise distinct, and
		\item for all its subterms of the form $ \GPar{\GT_1}{\GT_2} $, we have $ \Roles{\GT_1} \cap \Roles{\GT_2} = \emptyset $.
	\end{compactenum}
\end{defi}
We restrict our attention to well-formed global types.

\begin{defi}[Well-Formedness, Global Type]
	A local type is \emph{well-formed} if
	\begin{compactenum}[(1)]
		\item it neither contains free nor unguarded type variables and
		\item for all its subterms of the form $ \LSelR{\Role}{\Set{ \Label_i.\LT_i }_{i \in \indexSet}} $, $ \LBranR{\Role}{\Set{ \Label_i.\LT_i }_{i \in \indexSet}} $, $ \LSelW{\Role[R]}{\Set{ \Label_i.\LT_i }_{i \in \indexSet, \LabelD}} $, or of the form $ \LBranW{\Role[R]}{\Set{ \Label_i.\LT_i }_{i \in \indexSet, \LabelD}} $, we have $ \default \in \indexSet $ and the labels $ \Label_i $ are pairwise distinct.
	\end{compactenum}
\end{defi}
We restrict our attention to well-formed local types.

A session channel and a role together uniquely identify a participant of a session, called an \emph{actor}. A process has an actor $ \AT{\Chan[s]}{\Role} $ if it has an action prefix on $ \Chan $ that mentions $ \Role $ as its first role.
Let $ \Actors{\PT} $ be the set of actors of $ \PT $.

\subsection{Examples}

Consider the specification $ \GDice $ of a simple dice game in a bar
\begin{align*}
	\GRep{\TypeV}{\GComR{\Role[3]}{\Role[1]}{\nat}{\GComR{\Role[3]}{\Role[2]}{\nat}{\GBranR{\Role[3]}{\Role[1]}{{\Set{\Label[roll].\GBranR{\Role[3]}{\Role[2]}{\Label[roll].\TypeV}, \Label[exit].\GBranR{\Role[3]}{\Role[2]}{\Label[exit].\GEnd}}}}}}} & \tag{1}\label{eq:diceG}
\end{align*}
where the dealer Role~$ \Role[3] $ continues to $ \Label[roll] $ a dice and tell its value to player $ \Role[1] $ and then to $ \Label[roll] $ another time for player $ \Role[2] $ until the dealer decides to $ \Label[exit] $ the game.

We can combine \strongR communication/branching and \unrel communication, \eg by ordering a drink before each round in $ \GDice $.
\begin{align*}
	& \GRep{\TypeV}{\GComU{\Role[3]}{\Role[4]}{\Label[drink]}{\nat}{\GComR{\Role[3]}{\Role[1]}{\nat}{\GComR{\Role[3]}{\Role[2]}{\nat}{}}}}\\
	& \GBranR{\Role[3]}{\Role[1]}{{\Set{\Label[roll].\GBranR{\Role[3]}{\Role[2]}{\Label[roll].\TypeV}, \quad \Label[exit].\GBranR{\Role[3]}{\Role[2]}{\Label[exit].\GEnd}}}}
\end{align*}
where role~$ \Role[4] $ represents the bar tender and the noise of the bar may swallow these orders.
Moreover, we can remove the branching and specify a variant of the dice game in that $ \Role[3] $ keeps on rolling the dice forever, but, \eg due to a bar fight, one of our three players might get knocked out at some point or the noise of this fight might swallow the announcements of role~$ \Role[3] $:
\begin{align*}
	\GDiceUC &= \GRep{\TypeV}{\GComU{\Role[3]}{\Role[1]}{\Label[roll]}{\nat}{\GComU{\Role[3]}{\Role[2]}{\Label[roll]}{\nat}{\TypeV}}} & \tag{2}\label{eq:diceGU}
\end{align*}

To restore the branching despite the bar fight that causes failures, we need the \weakR branching mechanism.
\begin{align*}
	\GDiceW = \GRep{\TypeV}{}\GBranW{\Role[3]}{\Set{\Role[1], \Role[2]}}{}
		\begin{array}[t]{l}
			\Label[play].\GComU{\Role[3]}{\Role[1]}{\Label[roll]}{\nat}{\GComU{\Role[3]}{\Role[2]}{\Label[roll]}{\nat}{\TypeV}},\\
			\oplus \; \Label[end].\GComU{\Role[3]}{\Role[1]}{\Label[win]}{\bool}{\GComU{\Role[3]}{\Role[2]}{\Label[win]}{\bool}{\GEnd}}
		\end{array} & \tag{3}\label{eq:diceGW}
\end{align*}
If $ \Role[3] $ is knocked out by the fight, \ie crashes, the game cannot continue.
Then $ \Role[1] $ and $ \Role[2] $ move to the default branch $ \Label[end] $, have to skip the respective \unrel communications, and terminate.
But the game can continue as long as $ \Role[3] $ and at least one of the players $ \Role[1], \Role[2] $ participate.

An implementation of $ \GDiceW $ is $ \PDice = \PPar{\PDiceD}{\PPar{\PDiceR}{\PDiceQ}} $, where for $ \Role[i] \in \Set{ \Role[1], \Role[2] } $:
\begin{align*}
	\PDiceD ={}& \PReq{\Chan[a]}{\Role[3]}{\Chan}{\PRep{\ProcV}{\myif \; \Args_{\Role[1]} \leq 21 \wedge \Args_{\Role[2]} \leq 21}}\\
		& \begin{array}{l}
			\mythen \; \PSelW{\Chan}{\Role[3]}{\Set{\Role[1], \Role[2]}}{\Label[play]}{\PSendU{\Chan}{\Role[3]}{\Role[1]}{\Label[roll]}{\textsf{roll}(\Args_{\Role[1]})}{\PSendU{\Chan}{\Role[3]}{\Role[2]}{\Label[roll]}{\textsf{roll}(\Args_{\Role[2]})}{\ProcV}}}\\
			\myelse \; \PSelW{\Chan}{\Role[3]}{\Set{\Role[1], \Role[2]}}{\Label[end]}{\PSendU{\Chan}{\Role[3]}{\Role[1]}{\Label[win]}{\Args_{\Role[1]} \leq 21}{\PSendU{\Chan}{\Role[3]}{\Role[2]}{\Label[win]}{\Args_{\Role[2]} \leq 21}{\PEnd}}}
		\end{array}\\
	\PDiceP ={}& \PAcc{\Chan[a]}{\Role[i]}{\Chan}{\PRep{\ProcV}{\PBranW{\Chan}{\Role[i]}{\Role[3]}{\Label[play].\PGetU{\Chan}{\Role[i]}{\Role[3]}{\Label[roll]}{\Args}{\Args}{\ProcV} \oplus \Label[end].\PGetU{\Chan}{\Role[i]}{\Role[3]}{\Label[win]}{\false}{\Args[w]}{\PEnd}}}}
\end{align*}
Role~$ \Role[3] $ stores the sums of former dice rolls for the two players in its local variables $ \Args_{\Role[1]} $ and $ \Args_{\Role[2]} $, and $ \textsf{roll}(\Args_{\Role[i]}) $ rolls a dice and adds its value to the respective $ \Args_{\Role[i]} $.
Role~$ \Role[3] $ keeps rolling dice until the sum $ \Args_{\Role[i]} $ for one of the players exceeds $ 21 $.
If both sums $ \Args_{\Role[1]} $ and $ \Args_{\Role[2]} $ exceed $ 21 $ in the same round, then $ \Role[3] $ wins, \ie both players receive $ \false $; else, the player that stayed below $ 21 $ wins and receives $ \true $.
The players $ \Role[1] $ and $ \Role[2] $ use their respective last known sum that is stored in $ \Args $ as default value for the \unrel communication in the branch $ \Label[play] $ and $ \false $ as default value in the branch $ \Label[end] $.
The last branch, \ie $ \Label[end] $, is the default branch.

\subsection{Projection}

Our type system verifies processes, \ie implementations, against a specification that is a global type.
Since processes implement local views, local types are used as a mediator between the global specification and the respective local end points.
To ensure that the local types correspond to the global type, they are derived by \emph{projection}.
Instead of the projection function described in \cite{hondaYoshidaCarbone16} we use a more relaxed variant of projection as introduced in \cite{yoshidaDanielouBejleriHu10,CastellaniAtAll20,RobAtAll21}.

Projection maps global types onto the respective local type for a given role $ \Role[p] $.
The projections of the new global types are obtained straightforwardly from the projection of their respective \strongR counterparts:
\begin{align*}
	\Proj{\left( \Role_1 \to_{\textcolor{blue}{\diamond}} \Role_2{:}{\textcolor{blue}{\mathfrak{S}}}.\GT \right)}{\Role[p]} \deff
	\begin{cases}
		{\left[ \Role_2 \right]}\mathsf{!}_{\textcolor{blue}{\diamond}}{\textcolor{blue}{\mathfrak{S}}}.{\Proj{\GT}{\Role[p]}} & \text{if } \Role[p] = \Role_1\\
		{\left[ \Role_1 \right]}\mathsf{?}_{\textcolor{blue}{\diamond}}{\textcolor{blue}{\mathfrak{S}}}.{\Proj{\GT}{\Role[p]}} & \text{if } \Role[p] = \Role_2\\
		\Proj{\GT}{\Role[p]} & \text{otherwise}
	\end{cases}
\end{align*}
where either $ \diamond = \iR $, $ \mathfrak{S} = {\left< \Sort \right>} $ or $ \textcolor{blue}{\diamond = \iU} $, $ \textcolor{blue}{\mathfrak{S} = \Label{\left< \Sort \right>}} $ and
\begin{align*}
	\Proj{\left( \Role_1 \to_{\textcolor{blue}{\diamond}} \textcolor{blue}{\mathfrak{R}}{:}{\Set{ \Label_i.\GT_i }_{i \in \indexSet \textcolor{blue}{\mathfrak{D}}}} \right)}{\Role[p]} \deff
	\begin{cases}
		{\left[ \textcolor{blue}{\mathfrak{R}} \right]}\mathsf{!}_{\textcolor{blue}{\diamond}}{\Set{ \Label_i.\Proj{\GT_i}{\Role[p]} }_{i \in \indexSet}} & \text{if } \Role[p] = \Role_1\\
		{\left[ \Role_1 \right]}\mathsf{?}_{\textcolor{blue}{\diamond}}{\Set{ \Label_i.\Proj{\GT_i}{\Role[p]} }_{i \in \indexSet \textcolor{blue}{\mathfrak{D}}}} & \text{if } \textcolor{blue}{\mathfrak{B}}\\
		\bigsqcup_{i \in \indexSet} \left( \Proj{\GT_i}{\Role[p]} \right) & \text{otherwise}
	\end{cases}
\end{align*}
where either $ \diamond = \iR $, $ \mathfrak{R} = \Role_2 $, $ \mathfrak{B} $ is $ \Role[p] = \Role_2 $, $ \mathfrak{D} $ is empty or $ \textcolor{blue}{\diamond = \iW} $, $ \textcolor{blue}{\mathfrak{R} = \Role[R]} $, $ \textcolor{blue}{\mathfrak{B}} $ is $ \textcolor{blue}{\Role[p] \in \Role[R]} $, $ \textcolor{blue}{\mathfrak{D}} $ is $ \textcolor{blue}{, \LabelD} $.
In the last case of \strongR or \weakR branching---when projecting onto a role that does not participate in this branching---we map to $ \bigsqcup_{i \in \Set{1, \ldots, n}} \left( \Proj{\GT_i}{\Role[p]} \right) = \left( \Proj{\GT_1}{\Role[p]} \right) \sqcup \ldots \sqcup \left( \Proj{\GT_n}{\Role[p]} \right) $.
The $ \sqcup $ allows to unify the projections $ \Proj{\GT_i}{\Role[p]} $ if all of them return the same kind of branching input $ {\left[ \Role[p] \right]}\mathsf{?}_{\diamond}\ldots $ were the respective sets of branches my differ as long as the same label is always followed by the same local type.
The operation $ \sqcup $ is (similar to \cite{yoshidaDanielouBejleriHu10}) inductively defined as:
\begin{align*}
	\LT \sqcup \LT &= \LT\\
	\left( \LBranR{\Role}{\indexSet_1} \right) \sqcup \left( \LBranR{\Role}{\indexSet_2} \right) &= \LBranR{\Role}{\left( \indexSet_1 \sqcup \indexSet_2 \right)}\\
	\left( \LBranW{\Role}{\indexSet_1} \right) \sqcup \left( \LBranW{\Role}{\indexSet_2} \right) &= \LBranW{\Role}{\left( \indexSet_1 \sqcup \indexSet_2 \right)} \quad \text{if } \indexSet_1 \text{ and } \indexSet_2 \text{ have the same default branch}\\
	\indexSet \sqcup \emptyset &= \indexSet\\
	\indexSet \sqcup \left( \Set{ \Label.\LT } \cup \indexSet[J] \right) &=
		\begin{cases}
			\Set{ \Label.\left( \LT' \sqcup \LT \right) } \cup \left( \left( \indexSet \setminus \Set{ \Label.\LT' } \right) \sqcup \indexSet[J] \right) & \text{if } \Label.\LT' \in \indexSet\\
			\Set{ \Label.\LT } \cup \left( \indexSet \sqcup \indexSet[J] \right) & \text{if } \Label \notin \indexSet
		\end{cases}
\end{align*}
where $ \LT, \LT' \in \localTypes $ are local types, $ \indexSet, \indexSet_1, \indexSet_2, \indexSet[J] $ are sets of branches of local types of the form $ \Label.T $, $ \Label \notin \indexSet $ is short hand for $ \nexists \LT'\logdot \Label.\LT' \in \indexSet $, and is undefined in all other cases.
By the first line, identical types can be merged.
By the second and third line, local types for the reception of a branching request can be merged if they have the same prefix and the respective sets of branches can be merged.
The third line, for the \weakR case, additionally requires that the two sets of branches have the same default branch.
The sets of branches, that need to be merged according to the second and third line, contain elements of the form $ \Label.T $, where $ \Label $ is a label and $ T $ a local type.
The last two lines above inductively define how to merge such sets, \ie here we overload the operator $ \sqcup $ on local types to an operator on sets of branches of local types.
The case distinction in the last line ensures that elements $ \Label.T $ with a label that occurs in only one of the two sets can be kept, but if both sets contain an element with the same label then the respective local types have to be merged for the resulting set.

The mergeability relation $ \sqcup $ states that two types are identical up to their branching types, where only branches with distinct labels are allowed to be different.
This ensures that if the sender~$ \Role_1 $ in $ \GBranR{\Role_1}{\Role_2}{\Set{ \Label_i.\GT_i }_{i \in \indexSet}} $ decides to branch then only processes that are informed about this decision can adapt their behaviour accordingly; else projection is \textbf{not} defined.

The remaining global types are projected as follows:
\[\begin{array}{c}
	\Proj{\left( \GDel{\Role_1}{\Role_2}{\Chan}{\Role}{\LT}{\GT} \right)}{\Role[p]} \deff
		\begin{cases}
			\LDelA{\Role_2}{\Chan}{\Role}{\LT}{\Proj{\GT}{\Role[p]}} & \text{if } \Role[p] = \Role_1\\
			\LDelB{\Role_1}{\Chan}{\Role}{\LT}{\Proj{\GT}{\Role[p]}} & \text{if } \Role[p] = \Role_2\\
			\Proj{\GT}{\Role[p]} & \text{otherwise}
		\end{cases} \vspace{0.5em}\\
	\Proj{\left( \GPar{\GT_1}{\GT_2} \right)}{\Role[p]} \deff
		\begin{cases}
			\Proj{\GT_1}{\Role[p]} & \text{if } \Role[p] \notin \Roles{\GT_2}\\
			\Proj{\GT_2}{\Role[p]} & \text{if } \Role[p] \notin \Roles{\GT_1}
		\end{cases} \vspace{0.5em}\\
	\Proj{\left( \GRep{\TypeV}{\GT} \right)}{\Role[p]} \deff
		\begin{cases}
			\Proj{\GT}{\Role[p]} & \text{if } t \text{ does not occur in } \GT\\
			\LRep{\TypeV}{\Proj{\GT}{\Role[p]}} & \text{else if } \Role[p] \in \Roles{\GT}\\
			\LEnd & \text{otherwise}
		\end{cases}
	\hspace{1em}
	\Proj{\TypeV}{\Role[p]} \deff \TypeV
	\hspace{1em}
	\Proj{\GEnd}{\Role[p]} \deff \LEnd
\end{array}\]

The projection of delegation is similar to communication.
The projection of $ \GPar{\GT_1}{\GT_2} $ on $ \Role[p] $ is \textbf{not} defined if $ \Role[p] $ occurs on both sides of this parallel composition; it is $ \Proj{\GT_i}{\Role[p]} $ if $ \Role[p] $ occurs in exactly one side $ i \in \Set{ 1, 2 } $; or it is $ \Proj{\left( \GPar{\GT_1}{\GT_2} \right)}{\Role[p]} = \Proj{\GT_1}{\Role[p]} = \Proj{\GT_2}{\Role[p]} = \LEnd $ if $ \Role[p] $ does not occur at all.
Recursive types without their recursion variable are mapped to the projection of their recursion body (similar to \cite{CastellaniAtAll20}), else if $ \Role[p] $ occurs in the recursion body we map to a recursive local type, or else to successful termination.
Type variables and successful termination are mapped onto themselves.
We denote a global type $ \GT $ as \emph{projectable} if for all $ \Role \in \Roles{\GT} $ the projection $ \Proj{\GT}{\Role} $ is defined.
We restrict our attention to projectable global types.
Projection maps well-formed global types onto the respective local type for a given role $ \Role[p] $, where the results of projection---if defined---are again well-formed.

Projecting the global type $ \GDice $ in (\ref{eq:diceG}) results in the local types
\begin{align*}
	\LDice{\Role[3]} &= \LRep{\TypeV}{\LSendR{\Role[1]}{\nat}{\LSendR{\Role[2]}{\nat}{\LSelR{\Role[1]}{{\Set{\Label[roll].\LSelR{\Role[2]}{\Label[roll].\TypeV}, \quad \Label[exit].\LSelR{\Role[2]}{\Label[exit].\LEnd}}}}}}}\\
	\LDice{\Role[i]} &= \LRep{\TypeV}{\LGetR{\Role[3]}{\nat}{\LBranR{\Role[3]}{\Set{\Label[roll].\TypeV, \quad \Label[exit].\LEnd}}}}
\end{align*}
where the types of the two players $ \LDice{\Role[1]} = \LDice{\Role[2]} = \LDice{\Role[i]} $ are identical.
The projection of the outer branching in $ \GDice $ on $ \Role[2] $ results in $ \LBranR{\Role[3]}{\Label[roll].\TypeV} $ for the first branch and $ \LBranR{\Role[3]}{\Label[exit].\LEnd} $ for the second branch.
These two $ \LBranR{\Role[3]}{} $ types are unified by $ \sqcup $ into a single $ \LBranR{\Role[3]}{} $ type with two branches.

Projection maps $ \GDiceW $ in (\ref{eq:diceGW}) to:
\begin{align*}
	\LDiceW{\Role[3]} &= \LRep{\TypeV}{\LSelW{\Set{ \Role[1], \Role[2] }}{}
		\begin{array}[t]{l}
			\Label[play].\LSendU{\Role[1]}{\Label[roll]}{\nat}{\LSendU{\Role[2]}{\Label[roll]}{\nat}{\TypeV}}\\
			\oplus \; \Label[end].\LSendU{\Role[1]}{\Label[win]}{\bool}{\LSendU{\Role[2]}{\Label[win]}{\bool}{\LEnd}}
		\end{array}}\\
	\LDiceW{\Role[i]} &= \LRep{\TypeV}{\LBranW{\Role[3]}{\left( \Label[play].\LGetU{\Role[3]}{\Label[roll]}{\nat}{\TypeV} \oplus \Label[end].\LGetU{\Role[3]}{\Label[win]}{\bool}{\LEnd} \right)}}
\end{align*}
where $ \Role[i] \in \Set{ \Role[1], \Role[2] } $ and both $ \LDiceW{\Role[i]} $ are obtained by the second case of projection.
The type system will ensure that either $ \Role[3] $ transmits the request to branch to both players $ \Role[1], \Role[2] $ simultaneously and, since these messages cannot be lost, all players that are not crashed move to same branch or $ \Role[3] $ crashes and all remaining players move to the default branch.

Assume instead that $ \Role[3] $ can only inform one of the players $ \Role[1], \Role[2] $ at once.
\begin{align*}
	\GRep{\TypeV}{\GBranW{\Role[3]}{\Set{ \Role[1] }}{}
		\begin{array}[t]{l}
			\Label[play].\GComU{\Role[3]}{\Role[1]}{\Label[roll]}{\nat}{\GComU{\Role[3]}{\Role[2]}{\Label[roll]}{\nat}{\TypeV}}\\
			\oplus \; \Label[end].\GComU{\Role[3]}{\Role[1]}{\Label[win]}{\bool}{\GComU{\Role[3]}{\Role[2]}{\Label[win]}{\bool}{\GEnd}}
		\end{array}}
\end{align*}
is not projectable, because $ \sqcup $ does not allow to unify the projections $ \LGetU{\Role[3]}{\Label[roll]}{\nat}{\TypeV} $ and $ \LGetU{\Role[3]}{\Label[win]}{\bool}{\LEnd} $ of the two branches of $ \Role[2] $.
Replacing them by \strongR communications implies that neither $ \Role[3] $ nor $ \Role[2] $ fail.
The type
\begin{align*}
	& \GRep{\TypeV}{\GBranW{\Role[3]}{\Set{ \Role[1] }}{}
	 \begin{array}[t]{l}
			\Label[play].\GComU{\Role[3]}{\Role[1]}{\Label[roll]}{\nat}{\GBranW{\Role[3]}{\Set{ \Role[2] }}{\left( \Label[play].\GComU{\Role[3]}{\Role[2]}{\Label[roll]}{\nat}{\TypeV} \oplus \Label[end].\GEnd \right)}}\\
			\oplus \; \Label[end].\GComU{\Role[3]}{\Role[1]}{\Label[win]}{\bool}{\GBranW{\Role[3]}{\Set{ \Role[2] }}{\left( \Label[play].\GEnd \oplus \Label[end].\GComU{\Role[3]}{\Role[2]}{\Label[win]}{\bool}{\GEnd} \right)}}
		\end{array}}
\end{align*}
	where $ \Role[3] $ informs its two players subsequently about the chosen branch is projectable.
	But it introduces the two additional branches $ \Label[end].\GEnd $ and $ \Label[play].\GEnd $, \ie $ \Role[3] $ is allowed to choose the branches for the players $ \Role[1], \Role[2] $ separately and differently, whereas in (\ref{eq:diceG}) as well as in (\ref{eq:diceGW}) the players $ \Role[1], \Role[2] $ are always in the same branch.
Because of that, we allow for broadcast in \weakR branching such that $ \Role[3] $ can inform both players consistently without introducing additional and not-intended branches.

\subsection{Labels}
\label{sec:labels}

We use labels for two purposes:
they allow us to distinguish between different branches, as usual in \MPST-frameworks, and we assume that they may carry additional runtime information such as time\-stamps.
We do that because we think of labels not only as identifiers for branching, but also as some kind of meta data of messages as they can be often found in communication media or as they are assumed by many distributed algorithms.
A prominent example is the use of timestamps in message headers, that allow a receiver to identify outdated messages and to discard them.
Thereby, these additional runtime information placed in the label by the sender help the receiver to implement one of the already mentioned failure patterns; namely the one that allows a receiver to skip a message and continue with a default value instead.
We will introduce failure patterns with the semantics in the next section.
Although it is beyond the scope of this paper to discuss the implementation of failure patterns, we have to provide the technical means to do so.

Allowing for runtime information in labels requires a subtle difference in the way labels are used.
A timestamp may be added by the sender to capture the transmission time, but for the receiver it is hard to have this information already present in its label before or during reception.
Similarly, types in our static type system should not depend on any runtime information.
Hence, in contrast to standard \MPST, we do not expect the labels of senders and receivers as well as the labels of processes and types to match exactly.
Instead we assume a predicate $ \compL[] $ that compares two labels and is satisfied if the parts of the labels that do not refer to runtime information correspond.
If labels do not contain runtime information, $ \compL[] $ can be instantiated with equality.
We require that $ \compL[] $ is unambiguous on labels used in types, \ie given two labels of processes $ \Label_{\PT}, \Label_{\PT}' $ and two labels of types $ \Label_{\LT}, \Label_{\LT}' $ then $ \Label_{\PT} \compL \Label_{\PT}' \wedge \Label_{\PT} \compL \Label_{\LT} \Rightarrow \Label_{\PT}' \compL \Label_{\LT} $ and $ \Label_{\PT} \compL \Label_{\LT} \wedge \Label_{\LT} \nCompL \Label_{\LT}' \Rightarrow \Label_{\PT} \nCompL \Label_{\LT}' $.

Of course, the presented type system remains valid if we use labels without additional runtime information.
Indeed all presented examples carry in their labels statically available information only.
Interestingly, also the static information in labels, that have to coincide for senders and receivers and their types, can be exploited to guide communication.
In contrast to standard \MPST and to support \unrel communication, our \MPST variant will ensure that all occurrences of the same label are associated with the same sort.
This helps us in the case of failures to ensure the absence of communication mismatches, \ie the type of a transmitted value has to be the type that the receiver expects.
The global type $ G_{\nat\bool} = \GComU{\Role[1]}{\Role[2]}{\Label_1}{\nat}{\GComU{\Role[1]}{\Role[2]}{\Label_2}{\bool}{\GEnd}} $ specifies two subsequent \unrel communications in that values of different sorts are transmitted as discussed in Section~\ref{sec:faultolerance}.
If the first message with its natural number is lost but the second message containing a Boolean value is transmitted, the receiver $ \Role[2] $ should not wrongly receive a Boolean value although it still waits for a natural number.
To avoid this mismatch, we add a label to \unrel communication and ensure (by the typing rules) that the same label is never associated with different types.
In the case of $ G_{\nat\bool} $, the type system associates $ \Label_1 $ with sort $ \nat $ and $ \Label_2 $ with sort $ \bool $ and ensures that $ \Label_1 \nCompL \Label_2 $.
Sine $ \Label_1 \nCompL \Label_2 $, the reduction rules do not allow the receiver to use the value received in a message with label $ \Label_2 $ for its first communication action, \ie forces the receiver to first skip its first communication and use a default value before its is allowed to receive the message with label $ \Label_2 $.
Here we interpret labels again as some kind of meta data of messages that allow a receiver to use static information in labels to guide reception.
In particular, our \unrel communication mechanism exploits such meta data to guarantee strong properties about the communication structure including the described absence of communication mismatches.
Since the labels of the sender and the receiver are associated with a unique sort, the type system can then ensure that received values have the expected sort.
Similarly, labels are used in \cite{CairesVieira2010} to avoid communication errors.

%%%%%%%%%%%%%%%%%%%%%%
%  Failure Patterns  %
%%%%%%%%%%%%%%%%%%%%%%

\section{A Semantics with Failure Patterns}
\label{sec:failurePatterns}

Before we describe the semantics, we introduce substitution and structural congruence as auxiliary concepts.
The application of a substitution $ \Subst{\Args[y]}{\Args} $ on a term $ A $, denoted as $ A\Subst{\Args[y]}{\Args} $, is defined as the result of replacing all free occurrences of $ \Args $ in $ A $ by $ \Args[y] $, possibly applying alpha-conversion to avoid capture or name clashes. For all names $ n \in \names \setminus \Set{ \Args } $ the substitution behaves as the identity mapping. We use substitution on types as well as processes and naturally extend substitution to the substitution of variables by terms (to unfold recursions) and names by expressions (to instantiate a bound name with a received value).
We assume an evaluation function $ \Eval{\cdot} $ that evaluates expressions to values.

We use structural congruence to abstract from syntactically different processes with the same meaning, where $ \equiv $ is the least congruence that satisfies alpha conversion and the rules:
\[ \begin{array}{c}
	\PT \mid \PEnd \equiv \PT
	\hspace{2em}
	\PT_1 \mid \PT_2 \equiv \PT_2 \mid \PT_1
	\hspace{2em}
	\PT_1 \mid \left( \PT_2 \mid \PT_3 \right) \equiv \left( \PT_1 \mid \PT_2 \right) \mid \PT_3\
	\hspace{2em}
	\PRep{\ProcV}{\PEnd} \equiv \PEnd\\
	\PRes{\Args}{\PEnd} \equiv \PEnd
	\hspace{2em}
	\PRes{\Args}{\PRes{\Args[y]}{\PT}} \equiv \PRes{\Args[y]}{\PRes{\Args}{\PT}}
	\hspace{2em}
	\PRes{\Args}{\left( \PT_1 \mid \PT_2 \right)} \equiv \PT_1 \mid \PRes{\Args}{\PT_2} \quad \text{if } \Args \notin \FreeNames{\PT_1}
\end{array} \]

\begin{figure}[tp]
	\centering
	\renewcommand{\tabcolsep}{2pt}
	\renewcommand{\arraystretch}{1.2}
	\begin{tabular}{ll}
		(\textsf{Init}) & $ \PReq{\Chan[a]}{\Role[n]}{\Chan}{\PT_{\Role[n]}} \mid \prod_{1 \leq \Role[i] \leq \Role[n] - 1} \PAcc{\Chan[a]}{\Role[i]}{\Chan}{\PT_{\Role[i]}} \step \PRes{\Chan}{\left( \prod_{1 \leq \Role[i] \leq \Role[n]} \PT_{\Role[i]} \mid \prod_{1 \leq \Role[i], \Role[j] \leq \Role[n], \Role[i] \neq \Role[j]} \MQ{\Chan}{\Role[i]}{\Role[j]}{\emptyList} \right)} $\\
		& \hfill if $ \Chan[a] \neq \Chan $\\
		(\textsf{RSend}) & $ \PSendR{\Chan}{\Role_1}{\Role_2}{\Expr[y]}{\PT} \mid \MQ{\Chan}{\Role_1}{\Role_2}{\Queue} \step \PT \mid \MQ{\Chan}{\Role_1}{\Role_2}{\Queue\#\MessR{\Expr[v]}} $ \hfill if $ \Eval{\Expr[y]} = \Expr[v] $\\
		(\textsf{RGet}) & $ \PGetR{\Chan}{\Role_1}{\Role_2}{\Args}{\PT} \mid \MQ{\Chan}{\Role_2}{\Role_1}{\MessR{\Expr[v]}\#\Queue} \step \PT\Subst{\Expr[v]}{\Args} \mid \MQ{\Chan}{\Role_2}{\Role_1}{\Queue} $\\
		\textcolor{blue}{(\textsf{USend})} & \textcolor{blue}{$ \PSendU{\Chan}{\Role_1}{\Role_2}{\Label}{\Expr[y]}{\PT} \mid \MQ{\Chan}{\Role_1}{\Role_2}{\Queue} \step \PT \mid \MQ{\Chan}{\Role_1}{\Role_2}{\Queue\#\MessU{\Label}{\Expr[v]}} $ \hfill if $ \Eval{\Expr[y]} = \Expr[v] $}\\
		\textcolor{blue}{(\textsf{UGet})} & \textcolor{blue}{$ \PGetU{\Chan}{\Role_1}{\Role_2}{\Label}{\Expr[dv]}{\Args}{\PT} \mid \MQ{\Chan}{\Role_2}{\Role_1}{\MessU{\Label'}{\Expr[v]}\#\Queue} \step \PT\Subst{\Expr[v]}{\Args} \mid \MQ{\Chan}{\Role_2}{\Role_1}{\Queue} $}\\
		& \hfill \textcolor{blue}{if $ \Label \compL \Label' $, $ \fpUGet(\Chan, \Role_1, \Role_2, \Label', \ldots) $}\\
		\textcolor{blue}{(\textsf{USkip})} & \textcolor{blue}{$ \PGetU{\Chan}{\Role_1}{\Role_2}{\Label}{\Expr[dv]}{\Args}{\PT} \step \PT\Subst{\Expr[dv]}{\Args} $ \hfill if $ \fpUSkip(\Chan, \Role_1, \Role_2, \Label, \ldots) $}\\
		\textcolor{blue}{(\textsf{ML})} & \textcolor{blue}{$ \MQ{\Chan}{\Role_1}{\Role_2}{\MessU{\Label}{\Expr[v]}\#\Queue} \step \MQ{\Chan}{\Role_1}{\Role_2}{\Queue} $ \hfill if $ \fpML(\Chan, \Role_1, \Role_2, \Label, \ldots) $}\\
		(\textsf{RSel}) & $ \PSelR{\Chan}{\Role_1}{\Role_2}{\Label}{\PT} \mid \MQ{\Chan}{\Role_1}{\Role_2}{\Queue} \step \PT \mid \MQ{\Chan}{\Role_1}{\Role_2}{\Queue\#\MessBR{\Label}} $\\
		(\textsf{RBran}) & $ \PBranR{\Chan}{\Role_1}{\Role_2}{\Set{ \Label_i.\PT_i }_{i \in \indexSet}} \mid \MQ{\Chan}{\Role_2}{\Role_1}{\MessBR{\Label}\#\Queue} \step \PT_j \mid \MQ{\Chan}{\Role_2}{\Role_1}{\Queue} $ \hfill if $ \Label \compL \Label_j $, $ j \in \indexSet $\\
		\textcolor{blue}{(\textsf{WSel})} & \textcolor{blue}{$ \PSelW{\Chan}{\Role}{\Role[R]}{\Label}{\PT} \mid \prod_{\Role_i \in \Role[R]} \MQ{\Chan}{\Role}{\Role_i}{\Queue_i} \step \PT \mid \prod_{\Role_i \in \Role[R]} \MQ{\Chan}{\Role}{\Role_i}{\Queue_i\#\MessBW{\Label}} $}\\
		\textcolor{blue}{(\textsf{WBran})} & \textcolor{blue}{$ \PBranW{\Chan}{\Role_1}{\Role_2}{\Set{ \Label_i.\PT_i }_{i \in \indexSet, \LabelD}} \mid \MQ{\Chan}{\Role_2}{\Role_1}{\MessBW{\Label}\#\Queue} \step \PT_j \mid \MQ{\Chan}{\Role_2}{\Role_1}{\Queue} $ \hfill if $ \Label \compL \Label_j $, $ j \in \indexSet $}\\
		\textcolor{blue}{(\textsf{WSkip})} & \textcolor{blue}{$ \PBranW{\Chan}{\Role_1}{\Role_2}{\Set{ \Label_i.\PT_i }_{i \in \indexSet, \LabelD}} \step \PTD $ \hfill if $ \fpWSkip(\Chan, \Role_1, \Role_2, \ldots) $}\\
		\textcolor{blue}{(\textsf{Crash})} & \textcolor{blue}{$ \PT \step \PCrash $ \hfill if $ \fpCrash(\PT, \ldots) $}
	\end{tabular}
	\caption{Reduction Rules ($ \step $) of Fault-Tolerant Processes (Part I).}
	\label{fig:semanticsPartA}
\end{figure}

The reduction semantics of the session calculus is defined in the Figures~\ref{fig:semanticsPartA} and \ref{fig:semanticsPartB}, where we follow~\cite{hondaYoshidaCarbone16}: session initialisation is synchronous and communication within a session is asynchronous using message queues.
The rules are standard except for the five failure pattern and two rules for system failures: (\textsf{Crash}) for \emph{crash failures} and (\textsf{ML}) for \emph{message loss}.
\emph{Failure patterns} are predicates that we deliberately choose not to define here (see below).
They allow us to provide information about the underlying communication medium and the reliability of processes.

Rule~(\textsf{Init}) initialises a session with $ \Role[n] $ roles.
Session initialisation introduces a fresh session channel and unguards the participants of the session.
Finally, the message queues of this session are initialised with the empty list under the restriction of the session channel.

Rule~(\textsf{RSend}) implements an asynchronous \strongR message transmission.
As a result the value $ \Eval{\Expr[y]} $ is wrapped in a message and added to the end of the corresponding message queue and the continuation of the sender is unguarded.
Rule~(\textsf{USend}) is the counterpart of (\textsf{RSend}) for \unrel senders.
(\textsf{RGet}) consumes a message that is marked as \strongR with the index $ \operatorname{r} $ from the head of the respective message queue and replaces in the unguarded continuation of the receiver the bound variable $ \Args $ by the received value $ \Args[y] $.

There are two rules for the reception of a message in an \unrel communication that are guided by failure patterns.
Rule~(\textsf{UGet}) is similar to Rule~(\textsf{RGet}), but specifies a failure pattern $ \fpUGet $ to decide whether this step is allowed.
This failure pattern could, \eg, be used to reject messages that are too old.
Moreover, $ \Label \compL \Label' $ is required to enforce that the static information in the transmitted label matches the expectation specified in the label of the receiver.
As explained in Section~\ref{sec:labels}, this allows to avoid communication mismatches.
The Rule~(\textsf{USkip}) allows to skip the reception of a message in an \unrel communication using a failure pattern $ \fpUSkip $ and instead substitutes the bound variable $ \Args $ in the continuation with the default value $ \Args[dv] $.
The failure pattern $ \fpUSkip $ tells us whether a reception can be skipped (\eg via failure detector).

Rule~(\textsf{RSel}) puts the label $ \Label $ selected by $ \Role_1 $ at the end of the message queue towards $ \Role_2 $.
Its \weakR counterpart (\textsf{WSel}) is similar, but puts the label at the end of all relevant message queues.
With (\textsf{RBran}) a label is consumed from the top of a message queue and the receiver moves to the indicated branch.
There are again two \weakR counterparts of (\textsf{RBran}).
Rule~(\textsf{WBran}) is similar to (\textsf{RBran}), whereas (\textsf{WSkip}) allows $ \Role_1 $ to skip the message and to move to its default branch if the failure pattern $ \fpWSkip $ holds.
The requirement $ \Label \compL \Label_j $ in \textsf{RBran} and \textsf{WBran} ensures as usual that indeed the branch specified by the message at the queue is picked by the receiver.
Note that this branch has to be identified by the statically available information in the respective labels.

The Rules~(\textsf{Crash}) for \emph{crash failures} and (\textsf{ML}) for \emph{message loss}, describe failures of a system.
With Rule~(\textsf{Crash}) $ \PT $ can crash if $ \fpCrash $, where $ \fpCrash $ can \eg model immortal processes or global bounds on the number of crashes.
(\textsf{ML}) allows to drop an \unrel message if the failure pattern $ \fpML $ is valid.
$ \fpML $ allows, \eg, to implement safe channels that never lose messages or a global bound on the number of lost messages.

\begin{figure}[t]
	\centering
	\renewcommand{\tabcolsep}{2pt}
	\renewcommand{\arraystretch}{1.2}
	\begin{tabular}{ll}
		(\textsf{If-T}) & $ \PITE{\Expr}{\PT}{\PT'} \step \PT $ \hfill if $ \Expr $ is true\\
		(\textsf{If-F}) & $ \PITE{\Expr}{\PT}{\PT'} \step \PT' $ \hfill if $ \Expr $ is false\\
		(\textsf{Deleg}) & $ \PDelA{\Chan}{\Role_1}{\Role_2}{\AT{\Chan'}{\Role}}{\PT} \mid \MQ{\Chan}{\Role_1}{\Role_2}{\Queue} \step \PT \mid \MQ{\Chan}{\Role_1}{\Role_2}{\Queue\#\AT{\Chan'}{\Role}} $\\
		(\textsf{SRecv}) & $ \PDelB{\Chan}{\Role_1}{\Role_2}{\AT{\Chan'}{\Role}}{\PT} \mid \MQ{\Chan}{\Role_2}{\Role_1}{\AT{\Chan''}{\Role'}\#\Queue} \step \PT\Subst{\Chan''}{\Chan'}\Subst{\Role'}{\Role} \mid \MQ{\Chan}{\Role_1}{\Role_2}{\Queue} $\\
		(\textsf{Par}) & $ \PT_1 \mid \PT_2 \step \PT_1' \mid \PT_2 $ \hfill if $ \PT_1 \step \PT_1' $\\
		(\textsf{Res}) & $ \PRes{\Args}{\PT} \step \PRes{\Args}{\PT'} $ \hfill if $ \PT \step \PT' $\\
		(\textsf{Rec}) & $ \PRep{\ProcV}{\PT} \step \PT\Subst{\PRep{\ProcV}{\PT}}{\ProcV} $\\
		(\textsf{Struc}) & $ \PT_1 \step \PT_1' $ \hspace{15em} if $ \PT_1 \equiv \PT_2 $, $ \PT_2 \step \PT_2' $, $ \PT_2' \equiv \PT_1' $
	\end{tabular}
	\caption{Reduction Rules ($ \step $) of Fault-Tolerant Processes (Part II).}
	\label{fig:semanticsPartB}
\end{figure}

Figure~\ref{fig:semanticsPartB} provides the remaining reduction rules for conditionals, delegation, parallel composition, restriction, recursion, and structural congruence.
They are standard.

Consider the implementation of $ \GDiceUC $ in (\ref{eq:diceGU}), \ie an infinite variant of the dice game, where the players $ \Role[1] $ and $ \Role[2] $ use their respective last known sum $ \Args_{\Role[i]} $ of former dice rolls as default value:
\allowdisplaybreaks
\begin{align*}
	\PDiceUC ={}& \PPar{\PDiceDUC}{\PPar{\PDicePUC}{\PDiceQUC}}\\
	\PDiceDUC ={}& \PReq{\Chan[a]}{\Role[3]}{\Chan}{\PRep{\ProcV}{\PSendU{\Chan}{\Role[3]}{\Role[1]}{\Label[roll]}{\textsf{roll}(\Args_{\Role[1]})}{\PSendU{\Chan}{\Role[3]}{\Role[2]}{\Label[roll]}{\textsf{roll}(\Args_{\Role[2]})}{\ProcV}}}}\\
	\PDicePUC ={}& \PAcc{\Chan[a]}{\Role[i]}{\Chan}{\PRep{\ProcV}{\PGetU{\Chan}{\Role[i]}{\Role[3]}{\Label[roll]}{\Args_{\Role[i]}}{\Args_{\Role[i]}}{\ProcV}}}
%		\PDiceQUC ={}& \PInp{\Chan}{2}{\Chan[s]}{\PRec{\ProcV}{\Chan[s]}{\Role[2]}{\left( \PGetU{\Chan[s]}{\Role[2]}{\Role[3]}{\Label[roll]}{\Args_{\Role[2]}}{\Args_{\Role[2]}}{\PVar{\Chan[s]}{\Role[2]}} \right)}}
\end{align*}
An \unrel communication in a global type specifies a communication that, due to system failures, may or may not happen.
Moreover, regardless of the successful completion of this \unrel communication, the future behaviour of a well-typed system will follow its specification in the global type.
Since the players $ \Role[1] $ and $ \Role[2] $ repeat the same kind of \unrel action, they may lose track of the current round.
If they successfully receive a new sum of dice rolls from $ \Role[3] $ they cannot be sure on how often $ \Role[3] $ actually did roll the dice.
Because of lost messages, they may have missed some former announcements of $ \Role[3] $ and, because of their ability to skip the reception of messages, they may have proceeded to the next round before $ \Role[3] $ rolled a dice.
Because the information about the current round is irrelevant for the communication structure in this case, there is no need to enforce round information.

We deliberately do not specify failure pattern, although we usually assume that the failure patterns $ \fpUGet $, $ \fpUSkip $, and $ \fpWSkip $ use only local information, whereas $ \fpML $ and $ \fpCrash $ may use global information of the system in the current run.
We provide these predicates to allow for the implementation of system requirements or abstractions like failure detectors that are typical for distributed algorithms.
Directly including them in the semantics has the advantage that all traces satisfy the corresponding requirements, \ie all traces are valid \wrt the assumed system requirements.
An example for the instantiation of these patterns is given implicitly via the Conditions~\ref{cond:all}.\ref{cond:crash}--\ref{cond:all}.\ref{cond:fpWskip} in Section~\ref{sec:typing} and explicitly in Section~\ref{sec:example}.
If we instantiate the patterns $ \fpUGet $ with true and the patterns $ \fpUSkip $, $ \fpWSkip $, $ \fpCrash $, $ \fpML $ with false, then we obtain a system without failures.
In contrast, the instantiation of all five patterns with true results in a system where failures can happen completely non-deterministically at any time.

Note that we keep the failure patterns abstract and do not model how to check them in producing runs.
Indeed system requirements such as bounds on the number of processes that can crash usually cannot be checked, but result from observations, \ie system designers ensure that a violation of this bound is very unlikely and algorithm designers are willing to ignore these unlikely events.
In particular, $ \fpML $ and $ \fpCrash $ are thus often implemented as oracles for verification, whereas \eg $ \fpUSkip $ and $ \fpWSkip $ are often implemented by system specific time-outs.
Note that we are talking about implementing these failure patterns and not formalising them.
Failure patterns are abstractions of real world system requirements or software.
We implement them by conditions providing the necessary guarantees that we need in general (\ie for subject reduction and progress) or for the verification of concrete algorithms.
In practice, we expect that the systems on which the verified algorithms are running satisfy the respective conditions.
Accordingly, the session channels, roles, labels, and processes mentioned in Figure~\ref{fig:semanticsPartA} are not parameters of the failure patterns, but just a vehicle to more formally specify the conditions on failure patterns in Section~\ref{sec:typing}.
An implementation may or may not use these information to implement these patterns but may also use other information such as runtime information about time or the number of processes, as indicated by the \ldots in failure patterns in Figure~\ref{fig:semanticsPartA} such as $ \fpCrash(\PT, \ldots) $.

Similarly, \strongR and \weakR interactions in potentially faulty systems are abstractions.
They are usually implemented by handshakes and redundancy; replicated servers against crash failures and retransmission of late messages against message loss.
Algorithm designers have to be aware of the additional costs of these interactions.

%%%%%%%%%%%%%%%%%%%%%%%%%%%%%%%%%%%%%
%  Typing Fault-Tolerant Processes  %
%%%%%%%%%%%%%%%%%%%%%%%%%%%%%%%%%%%%%

\section{Typing Fault-Tolerant Processes}
\label{sec:typing}

The type of processes is checked using typing rules that define the derivation of type judgments.
Within type judgements, the type information are stored in type environments.

\begin{defi}[Type Environments]
	\label{def:typeEnvironments}
	The \emph{global} and \emph{session environments} are given by
	\begin{align*}
		\Gamma & \deffTerms
			\emptyset
			\sepTerms \Gamma \compS \Typed{\Args}{\Sort}
			\sepTerms \Gamma \compS \Typed{\Chan[a]}{\GT}
			\sepTerms \textcolor{blue}{\Gamma \compS \Typed{\Label}{\Sort}}
			\sepTerms \Gamma \compS \Typed{\ProcV}{\TypeV}\\
		\Delta & \deffTerms
			\emptyset
			\sepTerms \Delta \compS \Typed{\AT{\Chan}{\Role}}{\LT}
			\sepTerms \Delta \compS \MQ{\Chan}{\Role_1}{\Role_2}{\MT^*}
	\end{align*}
\end{defi}

Assignments $ \Typed{\Args}{\Sort} $ of values to sorts are used to check whether transmitted values are well-sorted, \ie sender and receiver expect the same sort.
Assignments $ \Typed{\Chan[a]}{\GT} $ capture the global type of a session for session initialisation via the shared channel $ \Chan[a] $.
Assignments $ \Typed{\Label}{\Sort} $ link labels to sorts.
Assignments $ \Typed{\ProcV}{\TypeV} $ of process variables to type variables are used to check the type of recursive processes.

Assignments $ \Typed{\AT{\Chan[s]}{\Role}}{\LT} $ of actors to local types are used to compare the behaviour of a process that implements this actor with its local specification $ \LT $.
Assignments $ \MQ{\Chan[s]}{\Role_1}{\Role_2}{\MT^*} $ allow to check the current content of a message queue $ \MQS{\Chan[s]}{\Role_1}{\Role_2} $ against a list of message types $ \MT $.

We write $ \Args \sharp \Gamma $ and $ \Args \sharp \Delta $ if the name $ \Args $ does not occur in $ \Gamma $ and $ \Delta $, respectively.
We use $ \compS $ to add an assignment provided that the new assignment is not in conflict with the type environment.
More precisely, $ \Gamma \compS \Typed{\Args}{\Sort} $ implies $ \Args \sharp \Gamma $, $ \Gamma \compS \Typed{\Label}{\Sort} $ implies $ \Label \sharp \Gamma $, and $ \Gamma \compS \Typed{\ProcV}{\Typed{\AT{\Chan}{\Role}}{\TypeV}} $ implies $ \ProcV, \TypeV \sharp \Gamma $.
Moreover, $ \Delta \compS \Typed{\AT{\Chan[s]}{\Role}}{\LT} $ implies $ \left( \nexists \LT' \logdot \Typed{\AT{\Chan[s]}{\Role}}{\LT'} \in \Delta \right) $ and $ \Delta \compS \MQ{\Chan[s]}{\Role_1}{\Role_2}{\mathcal{M}} $ implies $ \left( \nexists \mathcal{M}' \logdot \MQ{\Chan[s]}{\Role_1}{\Role_2}{\mathcal{M}'} \in \Delta \right) $.
We naturally extend this operator towards sets, \ie $ \Gamma \compS \Gamma' $ implies $ \left( \forall A \in \Gamma' \logdot \Gamma \compS A \right) $ and $ \Delta \compS \Delta' $ implies $ \left( \forall A \in \Delta' \logdot \Delta \compS A \right) $.
The conditions described for the operator $ \compS $ for global and session environments are referred to as \emph{linearity}.
Accordingly, we denote type environments that satisfy these properties as \emph{linear} and restrict in the following our attention to linear environments.
We abstract in session environments from assignments towards terminated local types, \ie $ \Delta \compS \Typed{\AT{\Chan}{\Role}}{\LEnd} = \Delta $.

\begin{figure}[tp]
	\centering
	\[ \begin{array}{c}
		\left( \textsf{Req} \right) \dfrac{\Typed{\Chan[a]}{\GT} \in \Gamma \quad \Length{\Roles{\GT}} = \Role[n] \quad \Gamma \vdash \PT \triangleright \Delta \compS \Typed{\AT{\Chan}{\Role[n]}}{\Proj{\GT}{\Role[n]}}}{\Gamma \vdash \PReq{\Chan[a]}{\Role[n]}{\Chan}{\PT} \triangleright \Delta}
		\vspace{0.8em}\\
		\left( \textsf{Acc} \right) \dfrac{\Typed{\Chan[a]}{\GT} \in \Gamma \quad 0 < \Role < \Length{\Roles{\GT}} \quad \Gamma \vdash \PT \triangleright \Delta \compS \Typed{\AT{\Chan}{\Role}}{\Proj{\GT}{\Role}}}{\Gamma \vdash \PAcc{\Chan[a]}{\Role}{\Chan}{\PT} \triangleright \Delta}
		\vspace{0.8em}\\
		\left( \textsf{RSend} \right) \dfrac{\Gamma \Vdash \Typed{\Expr[y]}{\Sort} \quad \Gamma \vdash \PT \triangleright \Delta \compS \Typed{\AT{\Chan}{\Role_1}}{\LT}}{\Gamma \vdash \PSendR{\Chan}{\Role_1}{\Role_2}{\Expr[y]}{\PT} \triangleright \Delta \compS \Typed{\AT{\Chan}{\Role_1}}{\LSendR{\Role_2}{\Sort}{\LT}}}
		\vspace{0.8em}\\
		\left( \textsf{RGet} \right) \dfrac{\Args \sharp \left( \Gamma, \Delta, \Chan \right) \quad \Gamma \compS \Typed{\Args}{\Sort} \vdash \PT \triangleright \Delta \compS \Typed{\AT{\Chan[s]}{\Role_1}}{\LT}}{\Gamma \vdash \PGetR{\Chan}{\Role_1}{\Role_2}{\Args}{\PT} \triangleright \Delta \compS \Typed{\AT{\Chan}{\Role_1}}{\LGetR{\Role_2}{\Sort}{\LT}}}
		\vspace{0.8em}\\
		\textcolor{blue}{\left( \textsf{USend} \right) \dfrac{\Gamma \Vdash \Typed{\Expr[y]}{\Sort} \quad \Label \compL \Label' \quad \Typed{\Label'}{\Sort} \in \Gamma \quad \Gamma \vdash \PT \triangleright \Delta \compS \Typed{\AT{\Chan}{\Role_1}}{\LT}}{\Gamma \vdash \PSendU{\Chan}{\Role_1}{\Role_2}{\Label}{\Expr[y]}{\PT} \triangleright \Delta \compS \Typed{\AT{\Chan}{\Role_1}}{\LSendU{\Role_2}{\Label'}{\Sort}{\LT}}}}
		\vspace{0.8em}\\
		\textcolor{blue}{\left( \textsf{UGet} \right) \dfrac{\Args \sharp \left( \Gamma, \Delta, \Chan \right) \quad \Gamma \Vdash \Typed{\Expr[v]}{\Sort} \quad \Label \compL \Label' \quad \Typed{\Label'}{\Sort} \in \Gamma \quad \Gamma \compS \Typed{\Args}{\Sort} \vdash \PT \triangleright \Delta \compS \Typed{\AT{\Chan[s]}{\Role_1}}{\LT}}{\Gamma \vdash \PGetU{\Chan}{\Role_1}{\Role_2}{\Label}{\Expr[v]}{\Args}{\PT} \triangleright \Delta \compS \Typed{\AT{\Chan}{\Role_1}}{\LGetU{\Role_2}{\Label'}{\Sort}{\LT}}}} \vspace{0.8em}\\
		\left( \textsf{RSel} \right) \dfrac{j \in \indexSet \quad \Label \compL \Label_j \quad \Gamma \vdash \PT \triangleright \Delta \compS \Typed{\AT{\Chan}{\Role_1}}{\LT_j}}{\Gamma \vdash \PSelR{\Chan}{\Role_1}{\Role_2}{\Label}{\PT} \triangleright \Delta \compS \Typed{\AT{\Chan}{\Role_1}}{\LSelR{\Role_2}{\Set{ \Label_i.\LT_i }_{i \in \indexSet}}}}
		\vspace{0.8em}\\
		\left( \textsf{RBran} \right) \dfrac{\forall j \in \indexSet_2\logdot \exists i \in \indexSet_1\logdot \Label_i \compL \Label_j' \wedge \Gamma \vdash \PT_i \triangleright \Delta \compS \Typed{\AT{\Chan}{\Role_1}}{\LT_j}}{\Gamma \vdash \PBranR{\Chan}{\Role_1}{\Role_2}{\Set{ \Label_i.\PT_i }_{i \in \indexSet_1}} \triangleright \Delta \compS \Typed{\AT{\Chan}{\Role_1}}{\LBranR{\Role_2}{\Set{ \Label_i'.\LT_i }_{i \in \indexSet_2}}}}
		\vspace{0.8em}\\
		\textcolor{blue}{\left( \textsf{WSel} \right) \dfrac{j \in \indexSet \quad \Label \compL \Label_j \quad \Gamma \vdash \PT \triangleright \Delta \compS \Typed{\AT{\Chan}{\Role}}{\LT_j}}{\Gamma \vdash \PSelW{\Chan}{\Role}{\Role[R]}{\Label}{\PT} \triangleright \Delta \compS \Typed{\AT{\Chan}{\Role}}{\LSelW{\Role[R]}{\Set{ \Label_i.\LT_i }_{i \in \indexSet}}}}}
		\vspace{0.8em}\\
		\textcolor{blue}{\left( \textsf{WBran} \right) \dfrac{\LabelD \compL \LabelD' \quad \forall j \in \indexSet_2\logdot \exists i \in \indexSet_1\logdot \Label_i \compL \Label_j' \wedge \Gamma \vdash \PT_i \triangleright \Delta \compS \Typed{\AT{\Chan}{\Role_1}}{\LT_j}}{\Gamma \vdash \PBranW{\Chan}{\Role_1}{\Role_2}{\Set{ \Label_i.\PT_i }_{i \in \indexSet_1, \LabelD}} \triangleright \Delta \compS \Typed{\AT{\Chan}{\Role_1}}{\LBranW{\Role_2}{\Set{ \Label_i'.\LT_i }_{i \in \indexSet_2, \LabelD'}}}}}
		\vspace{0.8em}\\
		\left( \textsf{Deleg} \right) \dfrac{\Gamma \vdash \PT \triangleright \Delta \compS \Typed{\AT{\Chan}{\Role_1}}{\LT}}{\Gamma \vdash \PDelA{\Chan}{\Role_1}{\Role_2}{\AT{\Chan'}{\Role}}{\PT} \triangleright \Delta \compS \Typed{\AT{\Chan}{\Role_1}}{\LDelA{\Role_2}{\Chan'}{\Role}{\LT'}{\LT}} \compS \Typed{\AT{\Chan'}{\Role}}{\LT'}}
		\vspace{0.8em}\\
		\left( \textsf{SRecv} \right) \dfrac{\Gamma \vdash \PT \triangleright \Delta \compS \Typed{\AT{\Chan}{\Role_1}}{\LT} \compS \Typed{\AT{\Chan'}{\Role}}{\LT'}}{\Gamma \vdash \PDelB{\Chan}{\Role_1}{\Role_2}{\AT{\Chan'}{\Role}}{\PT} \triangleright \Delta \compS \Typed{\AT{\Chan}{\Role_1}}{\LDelB{\Role_2}{\Chan'}{\Role}{\LT'}{\LT}}}
		\vspace{0.8em}\\
		\left( \textsf{End} \right) \dfrac{}{\Gamma \vdash \PEnd \triangleright \emptyset}
		\hspace{2em}
		\textcolor{blue}{\left( \textsf{Crash} \right) \dfrac{\Unreliable{\Delta}}{\Gamma \vdash \PCrash \triangleright \Delta}}
		\hspace{2em}
		\left( \textsf{If} \right) \dfrac{\Gamma \Vdash \Typed{\Expr}{\bool} \quad \Gamma \vdash \PT \triangleright \Delta \quad \Gamma \vdash \PT' \triangleright \Delta}{\Gamma \vdash \PITE{\Expr}{\PT}{\PT'} \triangleright \Delta}
		\vspace{0.8em}\\
		\left( \textsf{Par} \right) \dfrac{\Gamma \vdash \PT \triangleright \Delta \quad \Gamma \vdash \PT' \triangleright \Delta'}{\Gamma \vdash \PT \mid \PT' \triangleright \Delta \compS \Delta'}
		\hspace{2em}
		\left( \textsf{Res1} \right) \dfrac{\Args \sharp \left( \Gamma, \Delta \right) \quad \Gamma \compS \Typed{\Args}{\Sort} \vdash \PT \triangleright \Delta}{\Gamma \vdash \PRes{\Args}{\PT} \triangleright \Delta}
		\vspace{0.8em}\\
		\left( \textsf{Rec} \right) \dfrac{\Gamma \compS \Typed{\ProcV}{\TypeV} \vdash \PT \triangleright \Delta \compS \Typed{\AT{\Chan[s]}{\Role}}{\LT}}{\Gamma \vdash \PRep{\ProcV}{\PT} \triangleright \Delta \compS \Typed{\AT{\Chan[s]}{\Role}}{\LRep{\TypeV}{\LT}}}
		\hspace{2em}
		\left( \textsf{Var} \right) \dfrac{}{\Gamma \compS \Typed{\ProcV}{\TypeV} \vdash \ProcV \triangleright \Typed{\AT{\Chan}{\Role}}{\TypeV}}
	\end{array} \]
	\caption{Typing Rules for Fault-Tolerant Systems.}
	\label{fig:typingRules}
\end{figure}

A \emph{type judgement} is of the form $ \Gamma \vdash \PT \triangleright \Delta $, where $ \Gamma $ is a global environment, $ \PT \in \processes $ is a process, and $ \Delta $ is a session environment.
We use \emph{typing rules} to derive type judgements, where we assume that all mentioned global types are well-formed and projectable, all local types are well-formed, and all environments are linear.
A process $ \PT $ is \emph{well-typed} \wrt $ \Gamma $ and $ \Delta $ if $ \Gamma \vdash \PT \triangleright \Delta $ can be derived from the rules in the Figures~\ref{fig:typingRules} and \ref{fig:runtimeTypingRules}.
We write $ \Unreliable{\Delta} $ if none of the prefixes in $ \LT $ is \strongR or for delegation for all local types $ \LT $ in $ \Delta $ and if $ \Delta $ does not contain message queues.
With $ \Gamma \Vdash \Typed{\Expr[y]}{\Sort} $ we check that $ \Expr[y] $ is an expression of the sort $ \Sort $ if all names $ \Args $ in $ \Expr[y] $ are replaced by arbitrary values of sort $ \Sort_{\Args} $ for $ \Typed{\Args}{\Sort_{\Args}} \in \Gamma $.

Let us consider the interaction cases in Figure~\ref{fig:typingRules}. We observe that all new cases are quite similar to their \strongR counterparts.

Rule~(\textsf{RSend}) checks \strongR senders, \ie requires a matching \strongR sending in the local type of the actor and compares the actor with this type.
With $ \Gamma \Vdash \Typed{\Expr[y]}{\Sort} $ we check that $ \Expr[y] $ is an expression of the sort $ \Sort $ if all names $ \Args $ in $ \Expr[y] $ are replaced by arbitrary values of sort $ \Sort_{\Args} $ for $ \Typed{\Args}{\Sort_{\Args}} \in \Gamma $.
Then the continuation of the process is checked against the continuation of the type.
The \unrel case is very similar, but additionally checks that the label is assigned to the sort of the expression in $ \Gamma $.
Rule~(\textsf{RGet}) type \strongR receivers, where again the prefix is checked against a corresponding type prefix and the assumption $ \Typed{\Args}{\Sort} $ is added for the continuation.
Again the \unrel case is very similar, but apart from the label also checks the sort of the default value.

Rule~(\textsf{RSel}) checks the \strongR selection prefix, that the selected label matches one of the specified labels, and that the process continuation is well-typed \wrt the type continuation following the selected label.
The only difference in the \weakR case is the set of roles for the receivers.
For \strongR branching in (\textsf{RBran}) we check the prefix and that for each branch in the type there is a matching branch in the process that is well-typed \wrt the respective branch in the type.
For the \weakR case we have to additionally check that the default labels of the process and the type coincide.

Rule~(\textsf{Crash}) for crashed processes checks that $ \Unreliable{\Delta} $, \ie that for every type $ G $ or $ T $ in $ \Delta $ the predicate $ \Unreliable{G} $ or $ \Unreliable{T} $ holds.

Figure~\ref{fig:runtimeTypingRules} presents the runtime typing rules, \ie the typing rules for processes that may result from steps of a system that implements a global type.
Since it covers only operators that are not part of initial systems, a type checking tool might ignore them.
We need these rules however for the proofs of progress and subject reduction.
Under the assumption that initial systems cannot contain crashed processes, Rule~(\textsf{Crash}) may be moved to the set of runtime typing rules.

\begin{figure}[tp]
	\[ \begin{array}{c}
		\left( \textsf{Res2} \right) \dfrac{\begin{array}{c} \Set{ \Typed{\AT{\Chan}{\Role}}{\Proj{\GT}{\Role}} \mid \Role \in \Roles{\GT} } \compS \Set{ \MQ{\Chan}{\Role}{\Role'}{\emptyList} \mid \Role, \Role' \in \Roles{\GT'} \wedge \Role \neq \Role' } \stackrel{\Chan}{\Mapsto} \Delta'\\ \Chan \sharp \left(\Gamma, \Delta \right) \quad \Typed{\Chan[a]}{G} \in \Gamma \quad \Gamma \vdash \PT \triangleright \Delta \compS \Delta' \end{array}}{\Gamma \vdash \PRes{\Chan[s]}{\PT} \triangleright \Delta}
		\vspace{0.8em}\\
		\left( \textsf{MQComR} \right) \dfrac{\Gamma \Vdash \Typed{\Expr[v]}{\Sort} \quad \Gamma \vdash \MQ{\Chan}{\Role_1}{\Role_2}{\Queue} \triangleright \MQ{\Chan}{\Role_1}{\Role_2}{\MT}}{\Gamma \vdash \MQ{\Chan}{\Role_1}{\Role_2}{\MessR{\Expr[v]}\#\Queue} \triangleright \MQ{\Chan}{\Role_1}{\Role_2}{\MessR{\Sort}\#\MT}}
		\vspace{0.8em}\\
		\textcolor{blue}{\left( \textsf{MQComU} \right) \dfrac{\Gamma \Vdash \Typed{\Expr[v]}{\Sort} \quad \Label \compL \Label' \quad \Typed{\Label'}{\Sort} \in \Gamma \quad \Gamma \vdash \MQ{\Chan}{\Role_1}{\Role_2}{\Queue} \triangleright \MQ{\Chan}{\Role_1}{\Role_2}{\MT}}{\Gamma \vdash \MQ{\Chan}{\Role_1}{\Role_2}{\MessU{\Label}{\Expr[v]}\#\Queue} \triangleright \MQ{\Chan}{\Role_1}{\Role_2}{\MessU{\Label'}{\Sort}\#\MT}}}
		\vspace{0.8em}\\
		\left( \textsf{MQBranR} \right) \dfrac{\Label \compL \Label' \quad \Gamma \vdash \MQ{\Chan}{\Role_1}{\Role_2}{\Queue} \triangleright \MQ{\Chan}{\Role_1}{\Role_2}{\MT}}{\Gamma \vdash \MQ{\Chan}{\Role_1}{\Role_2}{\MessBR{\Label}\#\Queue} \triangleright \MQ{\Chan}{\Role_1}{\Role_2}{\MessBR{\Label'}\#\MT}}
		\vspace{0.8em}\\
		\textcolor{blue}{\left( \textsf{MQBranW} \right) \dfrac{\Label \compL \Label' \quad \Gamma \vdash \MQ{\Chan}{\Role_1}{\Role_2}{\Queue} \triangleright \MQ{\Chan}{\Role_1}{\Role_2}{\MT}}{\Gamma \vdash \MQ{\Chan}{\Role_1}{\Role_2}{\MessBW{\Label}\#\Queue} \triangleright \MQ{\Chan}{\Role_1}{\Role_2}{\MessBW{\Label'}\#\MT}}}
		\vspace{0.8em}\\
		\left( \textsf{MQDeleg} \right) \dfrac{\Gamma \vdash \MQ{\Chan}{\Role_1}{\Role_2}{\Queue} \triangleright \MQ{\Chan}{\Role_1}{\Role_2}{\MT}}{\Gamma \vdash \MQ{\Chan}{\Role_1}{\Role_2}{\AT{\Chan'}{\Role}\#\Queue} \triangleright \MQ{\Chan}{\Role_1}{\Role_2}{\AT{\Chan'}{\Role}\#\MT}}
		\hspace{2em}
		\left( \textsf{MQNil} \right) \dfrac{}{\Gamma \vdash \MQ{\Chan}{\Role_1}{\Role_2}{\emptyList} \triangleright \MQ{\Chan}{\Role_1}{\Role_2}{\emptyList}}
	\end{array} \]
	\caption{Runtime Typing Rules for Fault-Tolerant Systems.}
	\label{fig:runtimeTypingRules}
\end{figure}

Rule (\textsf{Res2}) types sessions that are already initialised and that may have performed already some of the steps described by their global type.
The relation $ \stackrel{\Chan}{\mapsto} $ is given in Figure~\ref{fig:reductionSessionEnv} and describes how a session environment evolves alongside reductions of the system, \ie it emulates the reduction steps of processes.
As an example consider the rule
$ \Delta \compS \Typed{\AT{\Chan}{\Role_1}}{\LSendR{\Role_2}{\Sort}{\LT}} \compS \MQ{\Chan}{\Role_1}{\Role_2}{\MT} \stackrel{\Chan}{\mapsto} \Delta \compS \Typed{\AT{\Chan}{\Role_1}}{\LT} \compS \MQ{\Chan}{\Role_1}{\Role_2}{\MT\#\MessR{\Sort}} $
that emulates (\textsf{RSend}).
Let $ \stackrel{\Chan}{\Mapsto} $ denote the reflexive and transitive closure of $ \stackrel{\Chan}{\mapsto} $.

\begin{figure}[tp]
	\[ \begin{array}{c}
%		\left( \textsf{Init} \right) \dfrac{s \sharp \Delta}{\Delta \stackrel{\Chan}{\mapsto} \Delta \compS \Set{ \Typed{\AT{\Chan}{\Role}}{\Proj{\GT}{\Role}} \mid \Role \in \Roles{\GT} } \compS \Set{ \MQ{\Chan}{\Role}{\Role'}{\emptyList} }}
%		\vspace{0.8em}\\
		\left( \textsf{RSend} \right) \dfrac{}{\Delta \compS \Typed{\AT{\Chan}{\Role_1}}{\LSendR{\Role_2}{\Sort}{\LT}} \compS \MQ{\Chan}{\Role_1}{\Role_2}{\MT} \stackrel{\Chan}{\mapsto} \Delta \compS \Typed{\AT{\Chan}{\Role_1}}{\LT} \compS \MQ{\Chan}{\Role_1}{\Role_2}{\MT\#\MessR{\Sort}}}
		\vspace{0.8em}\\
		\left( \textsf{RGet} \right) \dfrac{}{\Delta \compS \Typed{\AT{\Chan}{\Role_1}}{\LGetR{\Role_2}{\Sort}{\LT}} \compS \MQ{\Chan}{\Role_2}{\Role_1}{\MessR{\Sort}\#\MT} \stackrel{\Chan}{\mapsto} \Delta \compS \Typed{\AT{\Chan}{\Role_1}}{\LT} \compS \MQ{\Chan}{\Role_2}{\Role_1}{\MT}}
		\vspace{0.8em}\\
		\textcolor{blue}{\left( \textsf{USend} \right) \dfrac{}{\Delta \compS \Typed{\AT{\Chan}{\Role_1}}{\LSendU{\Role_2}{\Label}{\Sort}{\LT}} \compS \MQ{\Chan}{\Role_1}{\Role_2}{\MT} \stackrel{\Chan}{\mapsto} \Delta \compS \Typed{\AT{\Chan}{\Role_1}}{\LT} \compS \MQ{\Chan}{\Role_1}{\Role_2}{\MT\#\MessU{\Label}{\Sort}}}}
		\vspace{0.8em}\\
		\textcolor{blue}{\left( \textsf{UGet} \right) \dfrac{}{\Delta \compS \Typed{\AT{\Chan}{\Role_1}}{\LGetU{\Role_2}{\Label}{\Sort}{\LT}} \compS \MQ{\Chan}{\Role_2}{\Role_1}{\MessU{\Label}{\Sort}\#\MT} \stackrel{\Chan}{\mapsto} \Delta \compS \Typed{\AT{\Chan}{\Role_1}}{\LT} \compS \MQ{\Chan}{\Role_2}{\Role_1}{\MT}}}
		\vspace{0.8em}\\
		\textcolor{blue}{\left( \textsf{USkip} \right) \dfrac{}{\Delta \compS \Typed{\AT{\Chan}{\Role_1}}{\LGetU{\Role_2}{\Label}{\Sort}{\LT}} \stackrel{\Chan}{\mapsto} \Delta \compS \Typed{\AT{\Chan}{\Role_1}}{\LT}}}
		\vspace{0.8em}\\
		\textcolor{blue}{\left( \textsf{ML} \right) \dfrac{}{\Delta \compS \MQ{\Chan}{\Role_1}{\Role_2}{\MessU{\Label}{\Sort}\#\MT} \stackrel{\Chan}{\mapsto} \Delta \compS \MQ{\Chan}{\Role_1}{\Role_2}{\MT}}}
		\vspace{0.8em}\\
		\left( \textsf{RSel} \right) \dfrac{j \in \indexSet}{\Delta \compS \Typed{\AT{\Chan}{\Role_1}}{\LSelR{\Role_2}{\Set{ \Label_i.\LT_i }_{i \in \indexSet}}} \compS \MQ{\Chan}{\Role_1}{\Role_2}{\MT} \stackrel{\Chan}{\mapsto} \Delta \compS \Typed{\AT{\Chan}{\Role_1}}{\LT_j} \compS \MQ{\Chan}{\Role_1}{\Role_2}{\MT\#\MessBR{\Label_j}}}
		\vspace{0.8em}\\
		\left( \textsf{RBran} \right) \dfrac{j \in \indexSet}{\Delta \compS \Typed{\AT{\Chan}{\Role_1}}{\LBranR{\Role_2}{\Set{ \Label_i.\LT_i }_{i \in \indexSet}}} \compS \MQ{\Chan}{\Role_2}{\Role_1}{\MessBR{\Label_j}\#\MT} \stackrel{\Chan}{\mapsto} \Delta \compS \Typed{\AT{\Chan}{\Role_1}}{\LT_j} \compS \MQ{\Chan}{\Role_2}{\Role_1}{\MT}}
		\vspace{0.8em}\\
		\textcolor{blue}{\left( \textsf{WSel} \right) \dfrac{j \in \indexSet \quad \Role[R] = \Set{ \Role_1, \ldots, \Role_n}}{\begin{array}{c} \Delta \compS \Typed{\AT{\Chan}{\Role}}{\LSelW{\Role[R]}{\Set{ \Label_i.\LT_i }_{i \in \indexSet, \LabelD}}} \compS \MQ{\Chan}{\Role}{\Role_1}{\MT_1} \compS \ldots \compS \MQ{\Chan}{\Role}{\Role_n}{\MT_n} \stackrel{\Chan}{\mapsto}{}\\ \Delta \compS \Typed{\AT{\Chan}{\Role}}{\LT_j} \compS \MQ{\Chan}{\Role}{\Role_1}{\MT_1\#\MessBW{\Label_j}} \compS \ldots \compS \MQ{\Chan}{\Role}{\Role_n}{\MT_n\#\MessBW{\Label_j}} \end{array}}}
		\vspace{0.8em}\\
		\textcolor{blue}{\left( \textsf{WBran} \right) \dfrac{j \in \indexSet}{\Delta \compS \Typed{\AT{\Chan}{\Role_1}}{\LBranW{\Role_2}{\Set{ \Label_i.\LT_i }_{i \in \indexSet, \LabelD}}} \compS \MQ{\Chan}{\Role_2}{\Role_1}{\MessBW{\Label_j}\#\MT} \stackrel{\Chan}{\mapsto} \Delta \compS \Typed{\AT{\Chan}{\Role_1}}{\LT_j} \compS \MQ{\Chan}{\Role_2}{\Role_1}{\MT}}}
		\vspace{0.8em}\\
		\textcolor{blue}{\left( \textsf{WSkip} \right) \dfrac{}{\Delta \compS \Typed{\AT{\Chan}{\Role_1}}{\LBranW{\Role_2}{\Set{ \Label_i.\LT_i }_{i \in \indexSet, \LabelD}}} \stackrel{\Chan}{\mapsto} \Delta \compS \Typed{\AT{\Chan}{\Role_1}}{\LTD}}}
		\vspace{0.8em}\\
		\left( \textsf{Rec} \right) \dfrac{}{\Delta \compS \Typed{\AT{\Chan}{\Role}}{\LRep{\TypeV}{\LT}} \stackrel{\Chan}{\mapsto} \Delta \compS \Typed{\AT{\Chan}{\Role}}{\TypeV\Subst{\LRep{\TypeV}{\LT}}{\TypeV}}}
		\vspace{0.8em}\\
		\left( \textsf{Deleg} \right) \dfrac{}{\Delta \compS \Typed{\AT{\Chan}{\Role_1}}{\LDelA{\Role_2}{\Chan'}{\Role}{\LT'}{\LT}} \compS \Typed{\AT{\Chan'}{\Role}}{\LT'} \compS \MQ{\Chan}{\Role_1}{\Role_2}{\MT} \stackrel{\Chan}{\mapsto} \Delta \compS \Typed{\AT{\Chan}{\Role_1}}{\LT} \compS \MQ{\Chan}{\Role_1}{\Role_2}{\MT\#\AT{\Chan'}{\Role}}}
		\vspace{0.8em}\\
		\left( \textsf{SRecv} \right) \dfrac{}{\Delta \compS \Typed{\AT{\Chan}{\Role_1}}{\LDelB{\Role_2}{\Chan'}{\Role}{\LT'}{\LT}} \compS \MQ{\Chan}{\Role_2}{\Role_1}{\AT{\Chan'}{\Role}\#\MT} \stackrel{\Chan}{\mapsto} \Delta \compS \Typed{\AT{\Chan}{\Role_1}}{\LT} \compS \Typed{\AT{\Chan'}{\Role}}{\LT'} \compS \MQ{\Chan}{\Role_2}{\Role_1}{\MT}}
	\end{array} \]
	\caption{Reduction Rules for Session Environments.}
	\label{fig:reductionSessionEnv}
\end{figure}

We have to prove that our extended type system satisfies the standard properties of \MPST, \ie subject reduction and progress.
Because of the failure pattern in the reduction semantics in Figure~\ref{fig:semanticsPartA}, subject reduction and progress do not hold in general.
Instead we have to fix conditions on failure patterns that ensure these properties.
Subject reduction needs one condition on crashed processes and progress requires that no part of the system is blocked.
In fact, different instantiations of these failure patterns may allow for progress. We leave it for future work to determine what kind of conditions on failure patterns or requirements on their interactions are necessary. Here, we consider only one such set.

\begin{condition}[Failure Pattern]
	\label{cond:all}
	\hfill
	\begin{compactenum}
		\item If $ \fpCrash(\PT, \ldots) $, then $ \Unreliable{\PT} $. \label{cond:crash}
		\item The failure pattern $ \fpUGet(\Chan, \Role_1, \Role_2, \Label, \ldots) $ is always valid. \label{cond:ugetValid}
		\item The pattern $ \fpML(\Chan, \Role_1, \Role_2, \Label, \ldots) $ is valid iff $ \fpUSkip(\Chan, \Role_2, \Role_1, \Label, \ldots) $ is valid. \label{cond:fpMLifffpUSkip}
		\item If $ \fpCrash(\PT, \ldots) $ and $ \AT{\Chan}{\Role} \in \Actors{\PT} $ is an actor then eventually $ \fpUSkip(\Chan, \Role_2, \Role, \Label, \ldots) $ and $ \fpWSkip(\Chan, \Role_2, \Role, \Label, \ldots) $ for all $ \Role_2, \Label $. \label{cond:fpCrashImpliesSkip}
		\item If $ \fpCrash(\PT, \ldots) $ and $ \AT{\Chan}{\Role} \in \Actors{\PT} $ then eventually $ \fpML(\Chan, \Role_1, \Role, \Label, \ldots) $ for all $ \Role_1, \Label $. \label{cond:fpCrashImpliesML}
		\item If $ \fpWSkip(\Chan, \Role_1, \Role_2, \ldots) $ then $ \AT{\Chan}{\Role_2} $ is crashed, \ie the system does no longer contain an actor $ \AT{\Chan}{\Role_2} $ and the message queue $ \MQS{\Chan}{\Role_2}{\Role_1} $ is empty. \label{cond:fpWskip}
	\end{compactenum}
\end{condition}

The crash of a process should not block \strongR actions, \ie only processes with $ \Unreliable{\PT} $ can crash (Condition~\ref{cond:all}.\ref{cond:crash}).
Condition~\ref{cond:all}.\ref{cond:ugetValid} requires that no process can refuse to consume a message on its queue to prevent deadlocks that may arise from refusing a message that is never dropped.
Condition~\ref{cond:all}.\ref{cond:fpMLifffpUSkip} requires that if a message can be dropped from a message queue then the corresponding receiver has to be able to skip this message and vice versa.
Similarly, processes that wait for messages from a crashed process have to be able to skip (Condition~\ref{cond:all}.\ref{cond:fpCrashImpliesSkip}) and all messages of a queue towards a crashed receiver can be dropped (Condition~\ref{cond:all}.\ref{cond:fpCrashImpliesML}).
Finally, \weakR branching requests should not be lost.
To ensure that the receiver of such a branching request can proceed if the sender is crashed but is not allowed to skip the reception of the branching request before the sender crashed, we require that $ \fpWSkip(\Chan, \Role_1, \Role_2, \ldots) $ is false as long as $ \AT{\Chan}{\Role_2} $ is alive or messages on the respective queue are still in transit (Condition~\ref{cond:all}.\ref{cond:fpWskip}).

The combination of the 6 conditions in Conditions~\ref{cond:all} might appear quite restrictive as \eg the combination of the Condition~\ref{cond:all}.\ref{cond:fpCrashImpliesSkip} and \ref{cond:all}.\ref{cond:fpWskip} ensures the correct behaviour of \weakR branching such that branching messages can be skipped if and only if the respective sender has crashed.
An implementation of such a \weakR interaction in an asynchronous system that is subject to message losses and process crashes, might require something like a perfect failure detector or actually solving consensus\footnote{Note that we consider in Section~\ref{sec:example} a consensus algorithm. So, if the Condition~\ref{cond:all} requires a solution of consensus, an example on top of that solving consensus would be pointless.}.
It is important to remember that these conditions are minimal assumptions on the system requirements and that system requirements are abstractions.
Parts of them may be realised by actual software-code (which then allows to check them), whereas other parts of the system requirements may not be realised at all but rather observed (which then does not allow to verify them).
Weakly reliable branching is a good example of this case.
The easiest way to obtain a \weakR interaction, is by using a handshake communication and time-outs.
If the sender time-outs while waiting for an acknowledgement, it resends the message.
If the sender does not hear from its receiver for a long enough period of time, it assumes that the receiver has crashed and proceeds.
With carefully chosen time-frames for the time-outs, this approach is a compromise between correctness and efficiency.
In a theoretical sense, it is clearly not correct.
There is no time-frame such that the sender can be really sure that the receiver has crashed.
From a practical point of view, this is not so problematic, since in many systems failures are very unlikely.
If failures that are so severe that they are not captured by the time-outs are extremely unlikely, then it is often much more efficient to just accept that the algorithm is not correct in these cases.
Trying to obtain an algorithm that is always correct might be impossible or at least usually results into very inefficient algorithms.
Moreover, verifying this requires to also verify the underlying communication infrastructure and the way in that failures may occur, which is impossible or at least impracticable.
Because of that, it is an established method to verify the correctness of algorithms \wrt given system requirements (\eg in \cite{ChandraToueg96,lamport01,Tanenbaum17}), even if these system requirements are not verified and often do not hold in all (but only nearly all) cases.

Let us have a closer look at the typing rules in the Figures~\ref{fig:typingRules} and \ref{fig:runtimeTypingRules}.
We observe that all typing rules are clearly distinguished by the outermost operator of the process in the conclusion except that there are two typing rules for restriction.
With that, given a type judgement $ \Gamma \vdash \PT \triangleright \Delta $, we can use the structure of $ \PT $---with a case split for restriction---to reason about the structure of the proof tree that was necessary to obtain $ \Gamma \vdash \PT \triangleright \Delta $ and from that derive conditions about the nature of the involved type environments.
If $ \PT $ is \eg a parallel composition $ \PT_1 \mid \PT_2 $ then, since there is only one rule to type parallel compositions (the Rule~(\textsf{Par})), $ \Gamma \vdash \PPar{\PT_1}{\PT_2} \triangleright \Delta $ implies that there are $ \Delta_1, \Delta_2 $ such that $ \Delta = \Delta_1 \compS \Delta_2 $, $ \Gamma \vdash \PT_1 \triangleright \Delta_1 $, and $ \Gamma \vdash \PT_2 \triangleright \Delta_2 $.
In the following, we write 'by Rule~(\textsf{Par})' as short hand for 'by the clear distinction of the typing rules by the process in the conclusion and Rule~(\textsf{Par}) in particular' and similar for the other rules.

In the following we prove some properties of our \MPST variant.
We start with an auxiliary result, proving that structural congruence preserves the validity of type judgements.
The proof is by induction on $ \PT \equiv \PT' $.
In each case we can use the information about the structure of the process that is provided by the considered rule of structural congruence to conclude on the last few typing rules that had to be applied to derive the type judgement in the assumption.
From these partial proof trees we obtain enough information to construct the proof tree for the conclusion.

\begin{lem}[Subject Congruence]
	\label{lem:structuralCongruencePreservesJudgement}
	If $ \Gamma \vdash \PT \triangleright \Delta $ and $ \PT \equiv \PT' $ then $ \Gamma \vdash \PT' \triangleright \Delta $.
\end{lem}

\begin{proof}
	The proof is by induction on $ \PT \equiv \PT' $.
	\begin{description}
		\item[Case $ \PPar{\PT}{\PEnd} \equiv \PT $] Assume $ \Gamma \vdash \PPar{\PT}{\PEnd} \triangleright \Delta $. By the Rule~(\textsf{Par}), then there are $ \Delta_{\PT}, \Delta_{\PEnd} $ such that $ \Delta = \Delta_{\PT} \compS \Delta_{\PEnd} $, $ \Gamma \vdash \PT \triangleright \Delta_{\PT} $, and $ \Delta \vdash \PEnd \triangleright \Delta_{\PEnd} $.
			Moreover, by Rule~(\textsf{End}), $ \Gamma \vdash \PEnd \triangleright \Delta_{\PEnd} $ implies that $ \Delta_{\PEnd} = \emptyset $ and, thus, $ \Delta = \Delta_{\PT} $.
			Then also $ \Gamma \vdash \PT \triangleright \Delta $.\\
			For the opposite direction assume $ \Gamma \vdash \PT \triangleright \Delta $. By Rule~(\textsf{End}), $ \Gamma \vdash \PEnd \triangleright \emptyset $. With Rule~(\textsf{Par}) and because $ \Delta = \Delta \compS \emptyset $, then $ \Gamma \vdash \PPar{\PT}{\PEnd} \triangleright \Delta $.
		\item[Case $ \PPar{\PT_1}{\PT_2} \equiv \PPar{\PT_2}{\PT_1} $] Assume $ \Gamma \vdash \PPar{\PT_1}{\PT_2} \triangleright \Delta $.
			By Rule~(\textsf{Par}), then there are $ \Delta_{\PT_1}, \Delta_{\PT_2} $ such that $ \Delta = \Delta_{\PT_1} \compS \Delta_{\PT_2} $ and $ \Gamma \vdash \PT_i \triangleright \Delta_{\PT_i} $.
			By Rule~(\textsf{Par}) and since $ \Delta = \Delta_{\PT_1} \compS \Delta_{\PT_2} $ implies $ \Delta = \Delta_{\PT_2} \compS \Delta_{\PT_1} $, then $ \Gamma \vdash \PPar{\PT_2}{\PT_1} \triangleright \Delta $.\\
			The opposite direction is similar.
		\item[Case $ \PPar{\PT_1}{\left( \PPar{\PT_2}{\PT_3} \right)} \equiv \PPar{\left( \PPar{\PT_1}{\PT_2} \right)}{\PT_3} $] Assume $ \Gamma \vdash \PPar{\PT_1}{\left( \PPar{\PT_2}{\PT_3} \right)} \triangleright \Delta $.\\
			By Rule~(\textsf{Par}), then there are $ \Delta_{\PT_1}, \Delta_{\PT_2}, \Delta_{\PT_3} $ such that $ \Delta = \Delta_{\PT_1} \compS \left( \Delta_{\PT_2} \compS \Delta_{\PT_3} \right) $ and $ \Gamma \vdash \PT_i \triangleright \Delta_{\PT_i} $. By Rule~(\textsf{Par}) and because $ \Delta = \left( \Delta_{\PT_1} \compS \Delta_{\PT_2} \right) \compS \Delta_{\PT_3} $, then $ \Gamma \vdash \PPar{\left( \PPar{\PT_1}{\PT_2} \right)}{\PT_3} \triangleright \Delta $.\\
			The opposite direction is similar.
		\item[Case $ \PRep{\ProcV}{\PEnd} \equiv \PEnd $] Assume $ \Gamma \vdash \PRep{\ProcV}{\PEnd} \triangleright \Delta $.
			By Rule~(\textsf{Rec}), then $ \Delta = \Delta' \compS \Typed{\AT{\Chan[s]}{\Role}}{\LRep{\TypeV}{\LT}} $, $ \Gamma \compS \Typed{\ProcV}{\TypeV} \vdash \PEnd \triangleright \Delta' \compS \Typed{\AT{\Chan}{\Role}}{\LT} $, and, by Rule~(\textsf{End}), then $ \Delta' = \emptyset $ and $ \LT = \LEnd $.
			Since $ \LRep{\TypeV}{\LEnd} = \LEnd $ and $ \Delta \compS \Typed{\AT{\Chan[s]}{\Role}}{\LEnd} = \Delta $, then $ \Delta = \emptyset $.
			By Rule~(\textsf{End}), then $ \Gamma \vdash \PEnd \triangleright \Delta $.\\
			For the opposite direction assume $ \Gamma \vdash \PEnd \triangleright \Delta $.
			By Rule~(\textsf{End}), then $ \Delta = \emptyset $.
			By Rule~(\textsf{End}), then also $ \Gamma \compS \Typed{\ProcV}{\TypeV} \vdash \PEnd \triangleright \Delta $.
			By Rule~(\textsf{Rec}) and since $ \LRep{\TypeV}{\LEnd} = \LEnd $ and $ \Delta \compS \Typed{\AT{\Chan[s]}{\Role}}{\LEnd} = \Delta $, then $ \Gamma \vdash \PRep{\ProcV}{\PEnd} \triangleright \Delta $.
		\item[Case $ \PRes{\Args}{\PEnd} \equiv \PEnd $] Assume $ \Gamma \vdash \PRes{\Args}{\PEnd} \triangleright \Delta $.
			By one of the Rules~(\textsf{Res1}) or (\textsf{Res2}), then there are $ \Gamma', \Delta' $ such that $ \Gamma \compS \Gamma' \vdash \PEnd \triangleright \Delta $ or $ \Gamma \compS \Gamma' \vdash \PEnd \triangleright \Delta \compS \Delta' $.
			In both cases we can conclude with Rule~(\textsf{End}) that the session environment is empty, \ie $ \Delta = \emptyset $ and $ \Delta \compS \Delta' = \emptyset $.
			By Rule~(\textsf{End}), then $ \Gamma \vdash \PEnd \triangleright \Delta $.\\
			For the opposite direction assume $ \Gamma \vdash \PEnd \triangleright \Delta $.
			By Rule~(\textsf{End}), then $ \Delta = \emptyset $.
			By Rule~(\textsf{End}), then $ \Gamma \compS \Typed{\Args}{\Sort} \vdash \PEnd \triangleright \Delta $ for some sort $ \Sort $---regardless of whether $ x $ is a value or a session channel.
			By Rule~(\textsf{Res1}), then $ \Gamma \vdash \PRes{\Args}{\PEnd} \triangleright \Delta $.
		\item[Case $ \PRes{\Args}{\PRes{\Args[y]}{\PT}} \equiv \PRes{\Args[y]}{\PRes{\Args}{\PT}} $] Assume $ \Gamma \vdash \PRes{\Args}{\PRes{\Args[y]}{\PT}} \triangleright \Delta $.
			By al\-pha-con\-ver\-sion and the Rules~(\textsf{Res1}) and (\textsf{Res2}), then there is $ \Gamma' $ and some (possibly empty) $ \Delta' $ such that---for all combinations of the Rules~(\textsf{Res1}) and (\textsf{Res2}) for the restrictions of $ \Args $ and $ \Args[y] $---we have $ \Gamma \compS \Gamma' \vdash \PT \triangleright \Delta \compS \Delta' $.
			By the commutativity and associativity of $ \compS $ and two corresponding applications of the Rules~(\textsf{Res1}) and (\textsf{Res2}), then also $ \Gamma \vdash \PRes{\Args[y]}{\PRes{\Args}{\PT}} \triangleright \Delta $.\\
			The opposite direction is similar.
		\item[Case $ \PRes{\Args}{\left( \PPar{\PT_1}{\PT_2} \right)} \equiv \PPar{\PT_1}{\PRes{\Args}{\PT_2}} $ if $ \Args \notin \FreeNames{\PT_1} $] Assume $ \Gamma \vdash \PRes{\Args}{\left( \PPar{\PT_1}{\PT_2} \right)} \triangleright \Delta $.
			By one of the Rules~(\textsf{Res1}) or (\textsf{Res2}), then $ \Args \sharp \left( \Gamma, \Delta \right) $ and there are $ \Gamma', \Delta' $ such that $ \Gamma \compS \Gamma' \vdash \PPar{\PT_1}{\PT_2} \triangleright \Delta \compS \Delta' $, where $ \Gamma' $ assigns to $ \Args $ either a sort or a global type and $ \Delta' $ is either empty or contains only actors and message queues.
			Since $ \Args \notin \FreeNames{\PT_1} $ and by Rule~(\textsf{Par}), then there are $ \Delta_{\PT_1}, \Delta_{\PT_2} $ such that $ \Delta = \Delta_{\PT_1} \compS \Delta_{\PT_2} $, $ \Gamma \vdash \PT_1 \triangleright \Delta_{\PT_1} $, and $ \Gamma \compS \Gamma' \vdash \PT_2 \triangleright \Delta_{\PT_2} \compS \Delta' $.
			With one of the Rules~(\textsf{Res1}) or (\textsf{Res2}), then $ \Gamma \vdash \PRes{\Args}{\PT_2} \triangleright \Delta_{\PT_2} $.
			With Rule~(\textsf{Par}), then $ \Gamma \vdash \PPar{\PT_1}{\PRes{\Args}{\PT_2}} \triangleright \Delta $.\\
			The opposite direction is similar.
			\qedhere
	\end{description}
\end{proof}

Moreover, types are preserved modulo substitution of names by values of the same sort.
The proof is by induction on the typing rules.

\begin{lem}[Substitution]
	\label{lem:substitution}
	If $ \Gamma \compS \Typed{\Args[c]}{\Sort_c} \vdash \PT[Q] \triangleright \Delta $ and $ \Gamma \Vdash \Typed{\Expr[d]}{\Sort_c} $, then $ \Gamma \vdash \PT[Q]\Subst{\Expr[d]}{\Args[c]} \triangleright \Delta $.
\end{lem}

\begin{proof}
	The proof is by induction on the derivation of $ \Gamma \compS \Typed{\Args[c]}{\Sort_c} \vdash \PT[Q] \triangleright \Delta $.
	\begin{description}
		\item[(\textsf{Req})] Then, $ \PT[Q] = \PReq{\Chan[a]}{\Role[n]}{\Chan}{\PT} $, $ \Typed{\Chan[a]}{\GT} \in \Gamma $, $ \Length{\Roles{\GT}} = \Role[n] $, and $ \Gamma \compS \Typed{\Args[c]}{\Sort_c} \vdash \PT \triangleright \Delta \compS \Typed{\AT{\Chan}{\Role[n]}}{\Proj{\GT}{\Role[n]}} $.
			Without loss of generality, assume $ \Chan \neq \Args[c] $. Because of linearity, $ \Typed{\Chan[a]}{\GT} \in \Gamma $ implies $ \Chan[a] \neq \Args[c] $.
			By the induction hypothesis, $ \Gamma \compS \Typed{\Args[c]}{\Sort_c} \vdash \PT \triangleright \Delta \compS \Typed{\AT{\Chan}{\Role[n]}}{\Proj{\GT}{\Role[n]}} $ and $ \Gamma \Vdash \Typed{\Expr[d]}{\Sort_c} $ imply $ \Gamma \vdash \PT\Subst{\Expr[d]}{\Args[c]} \triangleright \Delta \compS \Typed{\AT{\Chan}{\Role[n]}}{\Proj{\GT}{\Role[n]}} $.
			By (\textsf{RReq}), then $ \Gamma \vdash \PT[Q]\Subst{\Expr[d]}{\Args[c]} \triangleright \Delta $.
		\item[(\textsf{Acc})] Then, $ \PT[Q] = \PAcc{\Chan[a]}{\Role}{\Chan}{\PT} $, $ \Typed{\Chan[a]}{\GT} \in \Gamma $, $ 0 < \Role < \Length{\Roles{\GT}} $, and $ \Gamma \compS \Typed{\Args[c]}{\Sort_c} \vdash \PT \triangleright \Delta \compS \Typed{\AT{\Chan}{\Role}}{\Proj{\GT}{\Role}} $.
			Without loss of generality, assume $ \Chan \neq \Args[c] $.
			Because of linearity, $ \Typed{\Chan[a]}{\GT} \in \Gamma $ implies $ \Chan[a] \neq \Args[c] $.
			By the induction hypothesis, $ \Gamma \vdash \PT\Subst{\Expr[d]}{\Args[c]} \triangleright \Delta \compS \Typed{\AT{\Chan}{\Role}}{\Proj{\GT}{\Role}} $.
			By (\textsf{Acc}), then $ \Gamma \vdash \PT[Q]\Subst{\Expr[d]}{\Args[c]} \triangleright \Delta $.
		\item[(\textsf{RSend})] Then, $ \PT[Q] = \PSendR{\Chan}{\Role_1}{\Role_2}{\Expr[y]}{\PT} $, $ \Delta = \Delta' \compS \Typed{\AT{\Chan}{\Role_1}}{\LSendR{\Role_2}{\Sort}{\LT}} $, $ \Gamma \compS \Typed{\Args[c]}{\Sort_c} \Vdash \Typed{\Args[y]}{\Sort} $, and $ \Gamma \compS \Typed{\Args[c]}{\Sort_c} \vdash \PT \triangleright \Delta' \compS \Typed{\AT{\Chan}{\Role_1}}{\LT} $.
			Because $ \Gamma \compS \Typed{\Args[c]}{\Sort_c} \Vdash \Typed{\Args[y]}{\Sort} $, $ \Chan \neq \Args[c] $.
			With $ \Gamma \Vdash \Typed{\Expr[d]}{\Sort_c} $, then $ \Gamma \Vdash \Typed{\Expr[y]\Subst{\Expr[d]}{\Args[c]}}{\Sort} $.
			By the induction hypothesis, $ \Gamma \vdash \PT\Subst{\Expr[d]}{\Args[c]} \triangleright \Delta' \compS \Typed{\AT{\Chan}{\Role_1}}{\LT} $.
			By (\textsf{RSend}), then $ \Gamma \vdash \PT[Q]\Subst{\Expr[d]}{\Args[c]} \triangleright \Delta $.
		\item[(\textsf{RGet})] Then, $ \PT[Q] = \PGetR{\Chan}{\Role_1}{\Role_2}{\Args}{\PT} $, $ \Delta = \Delta' \compS \Typed{\AT{\Chan}{\Role_1}}{\LGetR{\Role_2}{\Sort}{\LT}} $, $ \Args \sharp \left( \Gamma, \Args[c], \Delta', \Chan \right) $, and $ \Gamma \compS \Typed{\Args[c]}{\Sort_c} \compS \Typed{\Args}{\Sort} \vdash \PT_i \triangleright \Delta' \compS \Typed{\AT{\Chan}{\Role_1}}{\LT} $.
			Because $ \Gamma \compS \Typed{\Args[c]}{\Sort_c} \compS \Typed{\Args}{\Sort} \vdash \PT_i \triangleright \Delta' \compS \Typed{\AT{\Chan}{\Role_1}}{\LT} $, $ \Chan \neq \Args[c] $.
			By the induction hypothesis, $ \Gamma \compS \Typed{\Args}{\Sort} \vdash \PT\Subst{\Expr[d]}{\Args[c]} \triangleright \Delta' \compS \Typed{\AT{\Chan}{\Role_1}}{\LT} $.
			By (\textsf{RGet}), then $ \Gamma \vdash \PT[Q]\Subst{\Expr[d]}{\Args[c]} \triangleright \Delta $.
		\textcolor{blue}{\item[(\textsf{USend})] Then, $ \PT[Q] = \PSendU{\Chan}{\Role_1}{\Role_2}{\Label}{\Expr[y]}{\PT} $, $ \Delta = \Delta' \compS \Typed{\AT{\Chan}{\Role_1}}{\LSendU{\Role_2}{\Label'}{\Sort}{\LT}} $, $ \Label \compL \Label' $, $ \Gamma \compS \Typed{\Args[c]}{\Sort_c} \Vdash \Typed{\Expr[y]}{\Sort} $, and $ \Gamma \compS \Typed{\Args[c]}{\Sort_c} \vdash \PT \triangleright \Delta' \compS \Typed{\AT{\Chan}{\Role_1}}{\LT} $.
			Because of $ \Gamma \compS \Typed{\Args[c]}{\Sort_c} \Vdash \Typed{\Expr[y]}{\Sort} $, $ \Chan \neq \Args[c] $.
			With $ \Gamma \Vdash \Typed{\Expr[d]}{\Sort_c} $, then $ \Gamma \Vdash \Typed{\Expr[y]\Subst{\Expr[d]}{\Args[c]}}{\Sort} $.
			By the induction hypothesis, $ \Gamma \vdash \PT\Subst{\Expr[d]}{\Args[c]} \triangleright \Delta' \compS \Typed{\AT{\Chan}{\Role_1}}{\LT} $.
			By (\textsf{USend}), then $ \Gamma \vdash \PT[Q]\Subst{\Expr[d]}{\Args[c]} \triangleright \Delta $.}
		\textcolor{blue}{\item[(\textsf{UGet})] Then, $ \PT[Q] = \PGetU{\Chan}{\Role_1}{\Role_2}{\Label}{\Expr[v]}{\Args}{\PT} $, $ \Delta = \Delta' \compS \Typed{\AT{\Chan}{\Role_1}}{\LGetU{\Role_2}{\Label'}{\Sort}{\LT}} $, $ \Label \compL \Label' $, $ \Gamma \compS \Typed{\Args[c]}{\Sort_c} \Vdash \Typed{\Expr[v]}{\Sort} $, $ \Args \sharp \left( \Gamma, \Args[c], \Delta', \Chan \right) $, and $ \Gamma \compS \Typed{\Args[c]}{\Sort_c} \compS \Typed{\Args}{\Sort} \vdash \PT \triangleright \Delta' \compS \Typed{\AT{\Chan}{\Role_1}}{\LT} $.
			Because of $ \Gamma \compS \Typed{\Args[c]}{\Sort_c} \Vdash \Typed{\Expr[v]}{\Sort} $, $ \Chan[s] \neq \Args[c] $.
			With $ \Gamma \Vdash \Typed{\Expr[d]}{\Sort_c} \in \Gamma $, then $ \Gamma \Vdash \Typed{\Expr[v]\Subst{\Expr[d]}{\Args[c]}}{\Sort} $.
			By the induction hypothesis, $ \Gamma \vdash \PT\Subst{\Expr[d]}{\Args[c]} \triangleright \Delta' \compS \Typed{\AT{\Chan}{\Role_1}}{\LT} $.
			By (\textsf{UGet}), then $ \Gamma \vdash \PT[Q]\Subst{\Expr[d]}{\Args[c]} \triangleright \Delta $.}
		\item[(\textsf{RSel})] Then, $ \PT[Q] = \PSelR{\Chan}{\Role_1}{\Role_2}{\Label}{\PT} $, $ \Delta = \Delta' \compS \Typed{\AT{\Chan}{\Role_1}}{\LSelR{\Role_2}{\Set{ \Label_i. \LT_i }_{i \in \indexSet}}} $, $ j \in \indexSet $, $ \Label \compL \Label_j $, and $ \Gamma \compS \Typed{\Args[c]}{\Sort_c} \vdash \PT \triangleright \Delta' \compS \Typed{\AT{\Chan}{\Role_1}}{\LT_j} $.
			Because $ \Gamma \compS \Typed{\Args[c]}{\Sort_c} \vdash \PT \triangleright \Delta' \compS \Typed{\AT{\Chan}{\Role_1}}{\LT_j} $, $ \Chan \neq \Args[c] $.
			By the induction hypothesis, $ \Gamma \vdash \PT\Subst{\Expr[d]}{\Args[c]} \triangleright \Delta' \compS \Typed{\AT{\Chan}{\Role_1}}{\LT_j} $.
			By (\textsf{RSel}), then $ \Gamma \vdash \PT[Q]\Subst{\Expr[d]}{\Args[c]} \triangleright \Delta $.
		\item[(\textsf{RBran})] Then, $ \PT[Q] = \PBranR{\Chan}{\Role_1}{\Role_2}{\Set{ \Label_i.\PT_i }_{i \in \indexSet_1}} $, $ \Delta = \Delta' \compS \Typed{\AT{\Chan}{\Role_1}}{\LBranR{\Role_2}{\Set{ \Label_i. \LT_i }_{i \in \indexSet_2}}} $, and, for all $ j \in \indexSet_2 $ exists some $ i \in \indexSet_1 $ such that $ \Label_i \compL \Label_j $ and $ \Gamma \compS \Typed{\Args[c]}{\Sort_c} \vdash \PT_i \triangleright \Delta' \compS \Typed{\AT{\Chan}{\Role_1}}{\LT_j} $.
			Fix $ j $ and $ i $.
			Because $ \Gamma \compS \Typed{\Args[c]}{\Sort_c} \vdash \PT_i \triangleright \Delta' \compS \Typed{\AT{\Chan}{\Role_1}}{\LT_j} $, $ \Chan \neq \Args[c] $.
			By the induction hypothesis, $ \Gamma \vdash \PT_i\Subst{\Expr[d]}{\Args[c]} \triangleright \Delta' \compS \Typed{\AT{\Chan}{\Role_1}}{\LT_j} $.
			By (\textsf{RBran}), then $ \Gamma \vdash \PT[Q]\Subst{\Expr[d]}{\Args[c]} \triangleright \Delta $.
		\textcolor{blue}{\item[(\textsf{WSel})] Then, $ \PT[Q] = \PSelW{\Chan}{\Role}{\Role[R]}{\Label}{\PT} $, $ \Delta = \Delta' \compS \Typed{\AT{\Chan}{\Role}}{\LSelW{\Role[R]}{\Set{ \Label_i. \LT_i }_{i \in \indexSet, \LabelD}}} $, $ j \in \indexSet $, $ \Label \compL \Label_j $, and $ \Gamma \compS \Typed{\Args[c]}{\Sort_c} \vdash \PT \triangleright \Delta' \compS \Typed{\AT{\Chan}{\Role}}{\LT_j} $.
			Because $ \Gamma \compS \Typed{\Args[c]}{\Sort_c} \vdash \PT \triangleright \Delta' \compS \Typed{\AT{\Chan}{\Role}}{\LT_j} $, $ \Chan \neq \Args[c] $.
			By the induction hypothesis, $ \Gamma \vdash \PT\Subst{\Expr[d]}{\Args[c]} \triangleright \Delta' \compS \Typed{\AT{\Chan}{\Role}}{\LT_j} $.
			By (\textsf{WSel}), then $ \Gamma \vdash \PT[Q]\Subst{\Expr[d]}{\Args[c]} \triangleright \Delta $.}
		\textcolor{blue}{\item[(\textsf{WBran})] Then, $ \PT[Q] = \PBranW{\Chan}{\Role_1}{\Role_2}{\Set{ \Label_i.\PT_i }_{i \in \indexSet_1, \LabelD}} $, $ \Delta = \Delta' \compS \Typed{\AT{\Chan}{\Role_1}}{\LBranW{\Role_2}{\Set{ \Label_i. \LT_i }_{i \in \indexSet_2, \LabelD'}}} $, $ \LabelD \compL \LabelD' $, and, for all $ j \in \indexSet_2 $ exists some $ i \in \indexSet_1 $ such that $ \Label_i \compL \Label_j $ and $ \Gamma \compS \Typed{\Args[c]}{\Sort_c} \vdash \PT_i \triangleright \Delta' \compS \Typed{\AT{\Chan}{\Role_1}}{\LT_j} $.
			Fix $ j $ and $ i $.
			Because $ \Gamma \compS \Typed{\Args[c]}{\Sort_c} \vdash \PT_i \triangleright \Delta' \compS \Typed{\AT{\Chan}{\Role_1}}{\LT_j} $, $ \Chan \neq \Args[c] $.
			By the induction hypothesis, $ \Gamma \vdash \PT_i\Subst{\Expr[d]}{\Args[c]} \triangleright \Delta' \compS \Typed{\AT{\Chan}{\Role_1}}{\LT_j} $.
			By (\textsf{WBran}), then $ \Gamma \vdash \PT[Q]\Subst{\Expr[d]}{\Args[c]} \triangleright \Delta $.}
		\item[(\textsf{If})] Then, $ \PT[Q] = \PITE{\Expr}{\PT}{\PT'} $, $ \Gamma \compS \Typed{\Args[c]}{\Sort_c} \Vdash \Typed{\Expr}{\bool} $, $ \Gamma \compS \Typed{\Args[c]}{\Sort_c} \vdash \PT \triangleright \Delta $, and $ \Gamma \compS \Typed{\Args[c]}{\Sort_c} \vdash \PT' \triangleright \Delta $.
			With $ \Gamma \Vdash \Typed{\Expr[d]}{\Sort_c} \in \Gamma $, then $ \Gamma \Vdash \Typed{\Expr\Subst{\Expr[d]}{\Args[c]}}{\bool} $.
			By the induction hypothesis, $ \Gamma \vdash \PT\Subst{\Expr[d]}{\Args[c]} \triangleright \Delta $ and $ \Gamma \vdash \PT'\Subst{\Expr[d]}{\Args[c]} \triangleright \Delta $.
			By (\textsf{If}), then $ \Gamma \vdash \PT[Q]\Subst{\Expr[d]}{\Args[c]} \triangleright \Delta $.
		\item[(\textsf{Deleg})] Then $ \PT[Q] = \PDelA{\Chan}{\Role_1}{\Role_2}{\AT{\Chan'}{\Role}}{\PT} $, $ \Delta = \Delta' \compS \Typed{\AT{\Chan}{\Role_1}}{\LDelA{\Role_2}{\Chan'}{\Role}{\LT'}{\LT}} \compS \Typed{\AT{\Chan'}{\Role}}{\LT'} $, and $ \Gamma \compS \Typed{\Args[c]}{\Sort_c} \vdash \PT \triangleright \Delta' \compS \Typed{\AT{\Chan}{\Role_1}}{\LT} $.
			Because $ \Gamma \compS \Typed{\Args[c]}{\Sort_c} \vdash \PT[Q] \triangleright \Delta $, $ \Args[c] \neq \Chan $ and $ \Args[c] \neq \Chan' $.
			By the induction hypothesis, $ \Gamma \vdash \PT\Subst{\Expr[d]}{\Args[c]} \triangleright \Delta' \compS \Typed{\AT{\Chan}{\Role_1}}{\LT} $.
			By (\textsf{Deleg}), then $ \Gamma \vdash \PT[Q]\Subst{\Expr[d]}{\Args[c]} \triangleright \Delta $.
		\item[(\textsf{SRecv})] Then $ \PT[Q] = \PDelB{\Chan}{\Role_1}{\Role_2}{\AT{\Chan'}{\Role}}{\PT} $, $ \Delta = \Delta' \compS \Typed{\AT{\Chan}{\Role_1}}{\LDelB{\Role_2}{\Chan'}{\Role}{\LT'}{\LT}} $, and $ \Gamma \compS \Typed{\Args[c]}{\Sort_c} \vdash \PT \triangleright \Delta' \compS \Typed{\AT{\Chan}{\Role_1}}{\LT} \compS \Typed{\AT{\Chan'}{\Role}}{\LT'} $.
			Because $ \Gamma \compS \Typed{\Args[c]}{\Sort_c} \vdash \PT \triangleright \Delta' \compS \Typed{\AT{\Chan}{\Role_1}}{\LT} \compS \Typed{\AT{\Chan'}{\Role}}{\LT'} $, $ \Args[c] \neq \Chan $ and $ \Args[c] \neq \Chan' $.
			By the induction hypothesis, $ \Gamma \vdash \PT\Subst{\Expr[d]}{\Args[c]} \triangleright \Delta' \compS \Typed{\AT{\Chan}{\Role_1}}{\LT} \compS \Typed{\AT{\Chan'}{\Role}}{\LT'} $.
			By (\textsf{SRecv}), then $ \Gamma \vdash \PT[Q]\Subst{\Expr[d]}{\Args[c]} \triangleright \Delta $.
		\item[(\textsf{Par})] Then, $ \PT[Q] = \PPar{\PT}{\PT'} $, $ \Gamma \compS \Typed{\Args[c]}{\Sort_c} \vdash \PT \triangleright \Delta $, and $ \Gamma \compS \Typed{\Args[c]}{\Sort_c} \vdash \PT' \triangleright \Delta $.
			By the induction hypothesis, $ \Gamma \vdash \PT\Subst{\Expr[d]}{\Args[c]} \triangleright \Delta $ and $ \Gamma \vdash \PT'\Subst{\Expr[d]}{\Args[c]} \triangleright \Delta $.
			By (\textsf{Par}), then $ \Gamma \vdash \PT[Q]\Subst{\Expr[d]}{\Args[c]} \triangleright \Delta $.
		\item[(\textsf{Res1})] Then, $ \PT[Q] = \PRes{\Args}{\PT} $, $ \Args \sharp \left( \Gamma, \Args[c], \Delta \right) $, and $ \Gamma \compS \Typed{\Args[c]}{\Sort_c} \compS \Typed{\Args}{\Sort} \vdash \PT \triangleright \Delta $.
			By the induction hypothesis, $ \Gamma \compS \Typed{\Args}{\Sort} \vdash \PT\Subst{\Expr[d]}{\Args[c]} \triangleright \Delta $.
			By (\textsf{Res1}), then $ \Gamma \vdash \PT[Q]\Subst{\Expr[d]}{\Args[c]} \triangleright \Delta $.
		\item[(\textsf{Rec})] Then, $ \PT[Q] = \PRep{\ProcV}{\PT} $, $ \Delta = \Delta' \compS \Typed{\AT{\Chan}{\Role}}{\LT} $, and $ \Gamma \compS \Typed{\Args[c]}{\Sort_c} \compS \Typed{\ProcV}{\TypeV} \vdash \PT \triangleright \Delta' \compS \Typed{\AT{\Chan}{\Role}}{\LT} $.
			By the induction hypothesis, $ \Gamma \compS \Typed{\ProcV}{\TypeV} \vdash \PT\Subst{\Expr[d]}{\Args[c]} \triangleright \Delta' \compS \Typed{\AT{\Chan}{\Role}}{\LT} $.
			By (\textsf{Rec}), then $ \Gamma \vdash \PT[Q]\Subst{\Expr[d]}{\Args[c]} \triangleright \Delta $.
		\item[(\textsf{Var})] Then, $ \PT[Q] = \ProcV $, $ \Gamma = \Gamma' \compS \Typed{\ProcV}{\Typed{\AT{\Chan}{\Role}}{\TypeV}} $, and $ \Delta = \Typed{\AT{\Chan}{\Role}}{\TypeV} $.
			By (\textsf{Var}), then $ \Gamma \vdash \PT[Q]\Subst{\Expr[d]}{\Args[c]} \triangleright \Delta $.
		\item[(\textsf{End})] Then, $ \PT[Q] = \PEnd $ and $ \Delta = \emptyset $.
			By (\textsf{End}), then $ \Gamma \vdash \PT[Q]\Subst{\Args[d]}{\Args[c]} \triangleright \Delta $.
		\textcolor{blue}{\item[(\textsf{Crash})] Then, $ \PT[Q] = \PCrash $ and $ \Unreliable{\Delta} $.
			By (\textsf{UCrash}), then $ \Gamma \vdash \PT[Q]\Subst{\Expr[d]}{\Args[c]} \triangleright \Delta $.}
		\item[(\textsf{Res2})] Then, $ \Set{ \Typed{\AT{\Chan}{\Role}}{\Proj{\GT}{\Role}} \mid \Role \in \Roles{\GT} } \compS \Set{ \MQ{\Chan}{\Role}{\Role'}{\emptyList} \mid \Role, \Role' \in \Roles{\GT'} \wedge \Role \neq \Role' } \stackrel{\Chan}{\Mapsto} \Delta' $, $ \PT[Q] = \PRes{\Chan[s]}{\PT} $, $ \Chan[s] \sharp \left( \Gamma, \Args[c], \Delta \right) $, $ \Typed{\Chan[a]}{G} \in \Gamma $, and $ \Gamma \compS \Typed{\Args[c]}{\Sort_c} \vdash \PT \triangleright \Delta \compS \Delta' $.
			By the induction hypothesis, $ \Gamma \vdash \PT\Subst{\Expr[d]}{\Args[c]} \triangleright \Delta \compS \Delta' $.
			By (\textsf{ResS}), then $ \Gamma \vdash \PT[Q]\Subst{\Expr[d]}{\Args[c]} \triangleright \Delta $.
		\item[(\textsf{MQComR})] Then, $ \PT[Q] = \MQ{\Chan}{\Role_1}{\Role_2}{\MessR{\Expr[v]}\#\Queue} $, $ \Delta = \MQ{\Chan}{\Role_1}{\Role_2}{\MessR{\Sort}\#\MT} $, $ \Gamma \Vdash \Typed{\Expr[v]}{\Sort} $, and $ \Gamma \compS \Typed{\Args[c]}{\Sort_c} \vdash \MQ{\Chan}{\Role_1}{\Role_2}{\Queue} \triangleright \MQ{\Chan}{\Role_1}{\Role_2}{\MT} $.
			Because $ \Gamma \compS \Typed{\Args[c]}{\Sort_c} \Vdash \Typed{\Expr[v]}{\Sort} $, $ \Chan \neq \Args[c] $.
			With $ \Gamma \Vdash \Typed{\Expr[d]}{\Sort_c} $, then $ \Gamma \Vdash \Typed{\Expr[v]\Subst{\Expr[d]}{\Args[c]}}{\Sort} $.
			By the induction hypothesis, $ \Gamma \vdash \MQ{\Chan}{\Role_1}{\Role_2}{\Queue}\Subst{\Expr[d]}{\Args{c}} \triangleright \MQ{\Chan}{\Role_1}{\Role_2}{\MT} $.
			By (\textsf{MQComR}), then $ \Gamma \vdash \PT[Q]\Subst{\Expr[d]}{\Args[c]} \triangleright \Delta $.
		\textcolor{blue}{\item[(\textsf{MQComU})] Then, $ \PT[Q] = \MQ{\Chan}{\Role_1}{\Role_2}{\MessU{\Label}{\Expr[v]}\#\Queue} $, $ \Delta = \MQ{\Chan}{\Role_1}{\Role_2}{\MessU{\Label'}{\Sort}\#\MT} $, $ \Label \compL \Label' $, $ \Gamma \Vdash \Typed{\Expr[v]}{\Sort} $, and $ \Gamma \compS \Typed{\Args[c]}{\Sort_c} \vdash \MQ{\Chan}{\Role_1}{\Role_2}{\Queue} \triangleright \MQ{\Chan}{\Role_1}{\Role_2}{\MT} $.
			Because $ \Gamma \compS \Typed{\Args[c]}{\Sort_c} \Vdash \Typed{\Expr[v]}{\Sort} $, $ \Chan \neq \Args[c] $.
			With $ \Gamma \Vdash \Typed{\Expr[d]}{\Sort_c} $, then $ \Gamma \Vdash \Typed{\Expr[v]\Subst{\Expr[d]}{\Args[c]}}{\Sort} $.
			By the induction hypothesis, $ \Gamma \vdash \MQ{\Chan}{\Role_1}{\Role_2}{\Queue}\Subst{\Expr[d]}{\Args{c}} \triangleright \MQ{\Chan}{\Role_1}{\Role_2}{\MT} $.
			By (\textsf{MQComU}), then $ \Gamma \vdash \PT[Q]\Subst{\Expr[d]}{\Args[c]} \triangleright \Delta $.}
		\item[(\textsf{MQBranR})] Then, $ \PT[Q] = \MQ{\Chan}{\Role_1}{\Role_2}{\MessBR{\Label}\#\Queue} $, $ \Delta = \MQ{\Chan}{\Role_1}{\Role_2}{\MessBR{\Label'}\#\MT} $, $ \Label \compL \Label' $, and $ \Gamma \compS \Typed{\Args[c]}{\Sort_c} \vdash \MQ{\Chan}{\Role_1}{\Role_2}{\Queue} \triangleright \MQ{\Chan}{\Role_1}{\Role_2}{\MT} $.
			Because $ \Gamma \compS \Typed{\Args[c]}{\Sort_c} \vdash \MQ{\Chan}{\Role_1}{\Role_2}{\Queue} \triangleright \MQ{\Chan}{\Role_1}{\Role_2}{\MT} $, $ \Chan \neq \Args[c] $.
			By the induction hypothesis, $ \Gamma \vdash \MQ{\Chan}{\Role_1}{\Role_2}{\Queue}\Subst{\Expr[d]}{\Args{c}} \triangleright \MQ{\Chan}{\Role_1}{\Role_2}{\MT} $.
			By (\textsf{MQBranR}), then $ \Gamma \vdash \PT[Q]\Subst{\Expr[d]}{\Args[c]} \triangleright \Delta $.
		\textcolor{blue}{\item[(\textsf{MQBranW})] Then, $ \PT[Q] = \MQ{\Chan}{\Role_1}{\Role_2}{\MessBW{\Label}\#\Queue} $, $ \Delta = \MQ{\Chan}{\Role_1}{\Role_2}{\MessBW{\Label'}\#\MT} $, $ \Label \compL \Label' $, and $ \Gamma \compS \Typed{\Args[c]}{\Sort_c} \vdash \MQ{\Chan}{\Role_1}{\Role_2}{\Queue} \triangleright \MQ{\Chan}{\Role_1}{\Role_2}{\MT} $.
			Because $ \Gamma \compS \Typed{\Args[c]}{\Sort_c} \vdash \MQ{\Chan}{\Role_1}{\Role_2}{\Queue} \triangleright \MQ{\Chan}{\Role_1}{\Role_2}{\MT} $, $ \Chan \neq \Args[c] $.
			By the induction hypothesis, $ \Gamma \vdash \MQ{\Chan}{\Role_1}{\Role_2}{\Queue}\Subst{\Expr[d]}{\Args{c}} \triangleright \MQ{\Chan}{\Role_1}{\Role_2}{\MT} $.
			By (\textsf{MQBranW}), then $ \Gamma \vdash \PT[Q]\Subst{\Expr[d]}{\Args[c]} \triangleright \Delta $.}
		\item[(\textsf{MQDeleg})] Then $ \PT[Q] = \MQ{\Chan}{\Role_1}{\Role_2}{\AT{\Chan'}{\Role}\#\Queue} $, $ \Delta = \MQ{\Chan}{\Role_1}{\Role_2}{\AT{\Chan'}{\Role}\#\MT} $, and $ \Gamma \compS \Typed{\Args[c]}{\Sort_c} \vdash \MQ{\Chan}{\Role_1}{\Role_2}{\Queue} \triangleright \MQ{\Chan}{\Role_1}{\Role_2}{\MT} $.
			Because $ \Gamma \compS \Typed{\Args[c]}{\Sort_c} \vdash \MQ{\Chan}{\Role_1}{\Role_2}{\Queue} \triangleright \MQ{\Chan}{\Role_1}{\Role_2}{\MT} $, $ \Chan \neq \Args[c] $.
			By the induction hypothesis, $ \Gamma \vdash \MQ{\Chan}{\Role_1}{\Role_2}{\Queue}\Subst{\Expr[d]}{\Args{c}} \triangleright \MQ{\Chan}{\Role_1}{\Role_2}{\MT} $.
			By (\textsf{MQDeleg}), then $ \Gamma \vdash \PT[Q]\Subst{\Expr[d]}{\Args[c]} \triangleright \Delta $.
		\item[(\textsf{MQNil})] Then, $ \PT[Q] = \MQ{\Chan}{\Role_1}{\Role_2}{\emptyList} $ and $ \Delta = \MQ{\Chan}{\Role_1}{\Role_2}{\emptyList} $.
			By (\textsf{MQNil}), then we have $ \Gamma \vdash \PT[Q]\Subst{\Expr[d]}{\Args[c]} \triangleright \Delta $.
			\qedhere
	\end{description}
\end{proof}

\emph{Subject reduction} tells us that derivatives of well-typed systems are again well-typed.
This ensures that our formalism can be used to analyse processes by static type checking.
We extend subject reduction such that it provides some information on how the session environment evolves alongside reductions of the system using $ \stackrel{\Chan}{\mapsto} $.
In Figure~\ref{fig:reductionSessionEnv} we define the relation $ \stackrel{\Chan}{\mapsto} $ between session environments that emulates the reduction semantics.

\emph{Coherence} intuitively describes that a session environment captures all local endpoints of a collection of global types.
Since we capture all relevant global types in the global environment, we define
coherence on pairs of global and session environments.

\begin{defi}[Coherence]
	\label{def:coherence}
	The type environments $ \Gamma, \Delta $ are \emph{coherent} if, for all session channels $ \Chan $ in $ \Delta $, there exists a global type $ G $ in $ \Gamma $ such that the restriction of $ \Delta $ on assignments with $ \Chan $ is the set $ \Delta' $ such that:
	\begin{align*}
		\Set{ \Typed{\AT{\Chan}{\Role}}{\Proj{\GT}{\Role}} \mid \Role \in \Roles{\GT} } \compS \Set{ \MQ{\Chan}{\Role}{\Role'}{\emptyList} \mid \Role, \Role' \in \Roles{\GT} } \stackrel{\Chan}{\Mapsto} \Delta'
	\end{align*}
\end{defi}

We use $ \stackrel{\Chan}{\Mapsto} $ in the above definition to define coherence for systems that already performed some steps.
We can now prove subject reduction.

\begin{thm}[Subject Reduction]
	\label{thm:subjectReduction}
	If $ \Gamma \vdash \PT \triangleright \Delta $, $ \Gamma, \Delta $ are coherent, and $ \PT \step \PT' $, then there is some $ \Delta' $ such that $ \Gamma \vdash \PT' \triangleright \Delta' $.
\end{thm}

The proof is by induction on the derivation of $ \PT \step \PT' $.
In every case, we use the information about the structure of the processes to generate partial proof trees for the respective typing judgement.
Additionally, we use Condition~\ref{cond:all}.\ref{cond:crash} to ensure that the type environment of a crashed process cannot contain the types of reliable communication prefixes.

For the proof of subject reduction we further strengthen its goal and show additionally that there is some $ \Chan $ such that $ \Gamma, \Delta' $ is coherent and $ \Delta \stackrel{\Chan}{\Mapsto} \Delta' $, \ie that the session environment evolves by mimicking the respective reduction step and that this emulation reduces the session environment modulo $ \stackrel{\Chan}{\Mapsto} $ \wrt a single session $ \Chan $.
Moreover we use an additional goal---with weak coherence instead of coherence---to obtain a stronger induction hypothesis for the case of Rule~(\textsf{Par}).

\begin{defi}[Weak Coherence]
	The type environments $ \Gamma, \Delta $ are \emph{weakly coherent} if there exists some $ \Delta' $ such that $ \Gamma, \Delta \compS \Delta' $ are coherent.
\end{defi}

Ultimately, we are however interested into coherence.
Note that obviously the coherent case implies the respective weakly coherent case.
Our strengthened goal for subject reduction thus becomes:

\[ \begin{array}{c}
	\Gamma \vdash \PT \triangleright \Delta \wedge \Gamma, \Delta \text{ are coherent} \wedge \PT \step \PT' \longrightarrow\\
	\exists \Delta'\logdot \Gamma \vdash \PT' \triangleright \Delta' \wedge \Gamma, \Delta' \text{ are coherent} \wedge \Delta \stackrel{\Chan}{\Mapsto} \Delta'\\
	\text{and}\\
	\Gamma \vdash \PT \triangleright \Delta \wedge \Gamma, \Delta \text{ are weakly coherent} \wedge \PT \step \PT' \longrightarrow\\
	\exists \Delta'\logdot \Gamma \vdash \PT' \triangleright \Delta' \wedge \Gamma, \Delta' \text{ are weakly coherent} \wedge \Delta \stackrel{\Chan}{\Mapsto} \Delta'
\end{array} \]

\begin{proof}[Proof of Theorem~\ref{thm:subjectReduction}]
	The proof is by induction on the reduction $ \PT \step \PT' $ that is derived from the rules of the Figures~\ref{fig:semanticsPartA} and \ref{fig:semanticsPartB}.
	\begin{description}
		%%%%%%%%%%
		%  Init  %
		%%%%%%%%%%
		\item[Case of Rule~(\textsf{Init})] In this case
			\begin{align*}
				\PT &= \PPar{\PReq{\Chan[a]}{\Role[n]}{\Chan}{\PT[Q]_{\Role[n]}}}{\prod_{1 \leq \Role[i] \leq \Role[n] - 1} \PAcc{\Chan[a]}{\Role[i]}{\Chan}{\PT[Q]_{\Role[i]}}}\\
				\PT' &= \PRes{\Chan[s]}{\PT''}\\
				\PT'' &= \PPar{\prod_{1 \leq \Role[i] \leq \Role[n]} \PT[Q]_{\Role[i]}}{\prod_{1 \leq \Role[i], \Role[j] \leq \Role[n], \Role[i] \neq \Role[j]} \MQ{\Chan[s]}{\Role[i]}{\Role[j]}{[]}}
			\end{align*}
			$ \Chan[a] \neq \Chan $, and we use alpha conversion to ensure that $ \Chan \sharp \left( \Gamma, \Delta \right) $.
			By the typing Rules~(\textsf{Par}), (\textsf{Req}), and (\textsf{Acc}), $ \Gamma \vdash \PT \triangleright \Delta $ implies that there are $ \GT $, $ \Delta_{\PT[Q]_1}, \ldots, \Delta_{\PT[Q]_{\Role[n]}} $ such that $ \Delta = \Delta_{\PT[Q]_1} \compS \ldots \compS \Delta_{\PT[Q]_{\Role[n]}} $, $ \Typed{\Chan[a]}{\GT} \in \Gamma $, $ \Gamma \vdash \PT[Q]_{\Role[i]} \triangleright \Delta_{\PT[Q]_{\Role[i]}} \compS \Typed{\AT{\Chan[s]}{\Role[i]}}{\Proj{\GT}{\Role[i]}} $ for all $ 1 \leq \Role[i] \leq \Role[n] $, and $ \Length{\Roles{\GT}} = \Role[n] $.
			By Rule~(\textsf{MQNil}), $ \Gamma \vdash \MQ{\Chan[s]}{\Role[i]}{\Role[j]}{\emptyList} \triangleright \MQ{\Chan[s]}{\Role[i]}{\Role[j]}{\emptyList} $ for all $ \Role[i], \Role[j] \in \Roles{\GT} $ with $ \Role[i] \neq \Role[j] $.
			Since $ \Chan \sharp \Delta $, the composition $ \Delta \compS \Delta_{\Chan} $ for $ \Delta_{\Chan} = \Set{ \Typed{\AT{\Chan[s]}{\Role[i]}}{\Proj{\GT}{\Role[i]}}, \MQ{\Chan[s]}{\Role[i]}{\Role[j]}{\emptyList} \mid \Role[i], \Role[j] \in \Roles{\GT} \wedge \Role[i] \neq \Role[j] } $ is defined.
			By Rule~(\textsf{Par}), then $ \Gamma \vdash \PT'' \triangleright \Delta \compS \Delta_{\Chan} $.
			By Rule~(\textsf{Res2}), where we use reflexivity to obtain $ \Delta_{\Chan} \stackrel{\Chan}{\Mapsto} \Delta_{\Chan} $, then $ \Gamma \vdash \PT' \triangleright \Delta $.
			\\
			Since $ \Delta' = \Delta $, $ \Gamma, \Delta' $ is coherent and, by reflexivity, $ \Delta \stackrel{\Chan}{\Mapsto} \Delta' $.
		%%%%%%%%%%%
		%  RSend  %
		%%%%%%%%%%%
		\item[Case of Rule~(\textsf{RSend})] In this case $ \PT = \PPar{\PSendR{\Chan}{\Role_1}{\Role_2}{\Expr[y]}{\PT[Q]}}{\MQ{\Chan}{\Role_1}{\Role_2}{\Queue}} $, $ \Eval{\Expr[y]} = \Expr[v] $, and $ \PT' = \PPar{\PT[Q]}{\MQ{\Chan}{\Role_1}{\Role_2}{\Queue\#\MessR{\Expr[v]}}} $.
			By the Rules~(\textsf{Par}), (\textsf{RSend}), and the typing rules for message queues, $ \Gamma \vdash \PT \triangleright \Delta $ implies that there are $ \Delta_{\PT[Q]}, \Sort, \LT, \MT $ such that $ \Delta = \Delta_{\PT[Q]} \compS \Typed{\AT{\Chan}{\Role_1}}{\LSendR{\Role_2}{\Sort}{\LT}} \compS \MQ{\Chan}{\Role_1}{\Role_2}{\MT} $, $ \Gamma \Vdash \Typed{\Expr[y]}{\Sort} $, $ \Gamma \vdash \PT[Q] \triangleright \Delta_{\PT[Q]} \compS \Typed{\AT{\Chan}{\Role_1}}{\LT} $, and $ \Gamma \vdash \MQ{\Chan}{\Role_1}{\Role_2}{\Queue} \triangleright \MQ{\Chan[s]}{\Role_1}{\Role_2}{\MT} $.
			By Rule~(\textsf{MQComR}), then $ \Gamma \vdash \MQ{\Chan}{\Role_1}{\Role_2}{\Queue\#\MessR{\Expr[y]}} \triangleright \MQ{\Chan}{\Role_1}{\Role_2}{\MT\#\MessR{\Sort}} $.
			Since $ \Delta_{\PT[Q]} \compS \Typed{\AT{\Chan}{\Role_1}}{\LSendR{\Role_2}{\Sort}{\LT}} \compS \MQ{\Chan}{\Role_1}{\Role_2}{\MT} $ is defined, so is $ \Delta' = \Delta_{\PT[Q]} \compS \Typed{\AT{\Chan}{\Role_1}}{\LT} \compS \MQ{\Chan}{\Role_1}{\Role_2}{\MT\#\MessR{\Sort}} $.
			By Rule~(\textsf{Par}), then $ \Gamma \vdash \PT' \triangleright \Delta' $.
			\\
			By Rule~(\textsf{RSend}) of Figure~\ref{fig:reductionSessionEnv}, then $ \Delta \stackrel{\Chan}{\Mapsto} \Delta' $.
			Since $ \Gamma, \Delta $ are coherent and by transitivity of $ \stackrel{\Chan}{\Mapsto} $, then $ \Gamma, \Delta' $ are coherent.
		%%%%%%%%%%
		%  RGet  %
		%%%%%%%%%%
		\item[Case of Rule~(\textsf{RGet})] In this case $ \PT = \PPar{\PGetR{\Chan}{\Role_1}{\Role_2}{\Args}{\PT}}{\MQ{\Chan}{\Role_2}{\Role_1}{\MessR{\Expr[v]}\#\Queue}} $, $ \PT' = \PPar{\PT[Q]\Subst{\Expr[v]}{\Args}}{\MQ{\Chan}{\Role_2}{\Role_1}{\Queue}} $, and we use alpha conversion to ensure that $ \Args \sharp \left( \Gamma, \Delta, \Chan \right) $.
			By (\textsf{Par}), (\textsf{RGet}), and the typing rules for message queues, $ \Gamma \vdash \PT \triangleright \Delta $ implies that there are $ \Delta_{\PT[Q]}, \Sort_1, \Sort_2, \LT, \MT $ such that $ \Delta = \Delta_{\PT[Q]} \compS \Typed{\AT{\Chan}{\Role_1}}{\LGetR{\Role_2}{\Sort_1}{\LT}} \compS \MQ{\Chan}{\Role_2}{\Role_1}{\MessR{\Sort_2}\#\MT} $, $ \Gamma \compS \Typed{\Args}{\Sort_1} \vdash \PT[Q] \triangleright \Delta_{\PT[Q]} \compS \Typed{\AT{\Chan}{\Role_1}}{\LT} $, $ \Gamma \Vdash \Typed{\Expr[v]}{\Sort_2} $, and $ \Gamma \vdash \MQ{\Chan}{\Role_2}{\Role_1}{\Queue} \triangleright \MQ{\Chan}{\Role_2}{\Role_1}{\MT} $.
			Since $ \Gamma, \Delta $ are coherent, $ \Sort_1 = \Sort_2 $.
			By Lemma \ref{lem:substitution}, then $ \Gamma \vdash \PT[Q]\Subst{\Expr[v]}{\Args} \triangleright \Delta_{\PT[Q]} \compS \Typed{\AT{\Chan}{\Role_1}}{\LT} $.
			Since $ \Delta_{\PT[Q]_j} \compS \Typed{\AT{\Chan}{\Role_1}}{\LGetR{\Role_2}{\Sort}{\LT}} \compS \MQ{\Chan}{\Role_2}{\Role_1}{\MessR{\Sort_2}\#\MT} $ is defined, so is $ \Delta' = \Delta_{\PT[Q]} \compS \Typed{\AT{\Chan}{\Role_1}}{\LT} \compS \MQ{\Chan}{\Role_2}{\Role_1}{\MT} $.
			By Rule~(\textsf{Par}), then $ \Gamma \vdash \PT' \triangleright \Delta' $.
			\\
			By Rule (\textsf{RGet}) of Figure~\ref{fig:reductionSessionEnv}, then $ \Delta \stackrel{\Chan}{\Mapsto} \Delta' $.
			Since $ \Gamma, \Delta $ are coherent and by transitivity of $ \stackrel{\Chan}{\Mapsto} $, then $ \Gamma, \Delta' $ are coherent.
		%%%%%%%%%%%
		%  USend  %
		%%%%%%%%%%%
		\textcolor{blue}{\item[Case of Rule~(\textsf{RSend})] In this case $ \PT = \PPar{\PSendU{\Chan}{\Role_1}{\Role_2}{\Label}{\Expr[y]}{\PT[Q]}}{\MQ{\Chan}{\Role_1}{\Role_2}{\Queue}} $, $ \Eval{\Expr[y]} = \Expr[v] $, and $ \PT' = \PPar{\PT[Q]}{\MQ{\Chan}{\Role_1}{\Role_2}{\Queue\#\MessU{\Label}{\Expr[v]}}} $.
			By the Rules~(\textsf{Par}), (\textsf{USend}), and the typing rules for message queues, $ \Gamma \vdash \PT \triangleright \Delta $ implies that there are $ \Delta_{\PT[Q]}, \Label', \Sort, \LT, \MT $ such that $ \Delta = \Delta_{\PT[Q]} \compS \Typed{\AT{\Chan}{\Role_1}}{\LSendU{\Role_2}{\Label'}{\Sort}{\LT}} \compS \MQ{\Chan}{\Role_1}{\Role_2}{\MT} $, $ \Label \compL \Label' $, $ \Typed{\Label'}{\Sort} \in \Gamma $, $ \Gamma \Vdash \Typed{\Expr[y]}{\Sort} $, $ \Gamma \vdash \PT[Q] \triangleright \Delta_{\PT[Q]} \compS \Typed{\AT{\Chan}{\Role_1}}{\LT} $, and $ \Gamma \vdash \MQ{\Chan}{\Role_1}{\Role_2}{\Queue} \triangleright \MQ{\Chan[s]}{\Role_1}{\Role_2}{\MT} $.
			By Rule~(\textsf{MQComU}), then $ \Gamma \vdash \MQ{\Chan}{\Role_1}{\Role_2}{\Queue\#\MessU{\Label}{\Expr[y]}} \triangleright \MQ{\Chan}{\Role_1}{\Role_2}{\MT\#\MessU{\Label'}{\Sort}} $.
			Since $ \Delta_{\PT[Q]} \compS \Typed{\AT{\Chan}{\Role_1}}{\LSendU{\Role_2}{\Label'}{\Sort}{\LT}} \compS \MQ{\Chan}{\Role_1}{\Role_2}{\MT} $ is defined, so is $ \Delta' = \Delta_{\PT[Q]} \compS \Typed{\AT{\Chan}{\Role_1}}{\LT} \compS \MQ{\Chan}{\Role_1}{\Role_2}{\MT\#\MessU{\Label'}{\Sort}} $.
			By Rule~(\textsf{Par}), then $ \Gamma \vdash \PT' \triangleright \Delta' $.
			\\
			By Rule~(\textsf{USend}) of Figure~\ref{fig:reductionSessionEnv}, then $ \Delta \stackrel{\Chan}{\Mapsto} \Delta' $.
			Since $ \Gamma, \Delta $ are coherent and by transitivity of $ \stackrel{\Chan}{\Mapsto} $, then $ \Gamma, \Delta' $ are coherent.}
		%%%%%%%%%%
		%  UGet  %
		%%%%%%%%%%
		\textcolor{blue}{\item[Case of Rule~(\textsf{UGet})] Here $ \PT = \PPar{\PGetU{\Chan}{\Role_1}{\Role_2}{\Label}{\Expr[dv]}{\Args}{\PT}}{\MQ{\Chan}{\Role_2}{\Role_1}{\MessU{\Label'}{\Expr[v]}\#\Queue}} $, $ \PT' = \PPar{\PT[Q]\Subst{\Expr[v]}{\Args}}{\MQ{\Chan}{\Role_2}{\Role_1}{\Queue}} $, $ \Label \compL \Label' $, and (using alpha conversion) $ \Args \sharp \left( \Gamma, \Delta, \Chan \right) $.
			By (\textsf{Par}), (\textsf{UGet}), and the typing rules for message queues, $ \Gamma \vdash \PT \triangleright \Delta $ implies that there are $ \Delta_{\PT[Q]}, \Sort_1, \Sort_2, \LT, \MT $ such that $ \Delta = \Delta_{\PT[Q]} \compS \Typed{\AT{\Chan}{\Role_1}}{\LGetU{\Role_2}{\Label''}{\Sort_1}{\LT}} \compS \MQ{\Chan}{\Role_2}{\Role_1}{\MessU{\Label'''}{\Sort_2}\#\MT} $, $ \Label \compL \Label'' $, $ \Label' \compL \Label'' $, $ \Typed{\Label''}{\Sort_1} \in \Gamma $, $ \Typed{\Label'''}{\Sort_2} \in \Gamma $, $ \Gamma \compS \Typed{\Args}{\Sort_1} \vdash \PT[Q] \triangleright \Delta_{\PT[Q]} \compS \Typed{\AT{\Chan}{\Role_1}}{\LT} $, $ \Gamma \Vdash \Typed{\Expr[v]}{\Sort_2} $, and $ \Gamma \vdash \MQ{\Chan}{\Role_2}{\Role_1}{\Queue} \triangleright \MQ{\Chan}{\Role_2}{\Role_1}{\MT} $.
			Since $ \Typed{\Label''}{\Sort_1} \in \Gamma $, $ \Typed{\Label'''}{\Sort_2} \in \Gamma $, and $ \Label'' \compL \Label \compL \Label' \compL \Label''' $, we have $ \Label'' = \Label''' $ and $ \Sort_1 = \Sort_2 $.
			By Rule~(\textsf{Par}), then $ \Gamma \vdash \PT' \triangleright \Delta' $.
			\\
			By Rule (\textsf{UGet}) of Figure~\ref{fig:reductionSessionEnv}, then $ \Delta \stackrel{\Chan}{\Mapsto} \Delta' $.
			Since $ \Gamma, \Delta $ are coherent and by transitivity of $ \stackrel{\Chan}{\Mapsto} $, then $ \Gamma, \Delta' $ are coherent.}
		%%%%%%%%%%%
		%  USkip  %
		%%%%%%%%%%%
		\textcolor{blue}{\item[Case of Rule~(\textsf{USkip})] In this case $ \PT = \PGetU{\Chan}{\Role_1}{\Role_2}{\Label}{\Expr[dv]}{\Args}{\PT} $, $ \PT' = \PT[Q]\Subst{\Expr[dv]}{\Args} $, and we use alpha conversion to ensure that $ \Args \sharp \left( \Gamma, \Delta, \Chan \right) $.
			By (\textsf{UGet}), $ \Gamma \vdash \PT \triangleright \Delta $ implies that there are $ \Delta_{\PT[Q]}, \Sort, \LT $ such that $ \Delta = \Delta_{\PT[Q]} \compS \Typed{\AT{\Chan}{\Role_1}}{\LGetU{\Role_2}{\Label'}{\Sort}{\LT}} $, $ \Label \compL \Label' $, $ \Typed{\Label'}{\Sort} \in \Gamma $, $ \Gamma \Vdash \Typed{\Expr[dv]}{\Sort} $, and $ \Gamma \compS \Typed{\Args}{\Sort} \vdash \PT[Q] \triangleright \Delta_{\PT[Q]} \compS \Typed{\AT{\Chan}{\Role_1}}{\LT} $.
			By Lemma \ref{lem:substitution}, then $ \Gamma \vdash \PT' \triangleright \Delta' $ with $ \Delta' = \Delta_{\PT[Q]} \compS \Typed{\AT{\Chan}{\Role_1}}{\LT} $.
			\\
			By Rule (\textsf{USkip}) of Figure~\ref{fig:reductionSessionEnv}, then $ \Delta \stackrel{\Chan}{\Mapsto} \Delta' $.
			Since $ \Gamma, \Delta $ are coherent and by transitivity of $ \stackrel{\Chan}{\Mapsto} $, then $ \Gamma, \Delta' $ are coherent.}
		%%%%%%%%
		%  ML  %
		%%%%%%%%
		\textcolor{blue}{\item[Case of Rule~(\textsf{ML})] In this case $ \PT = \MQ{\Chan}{\Role_1}{\Role_2}{\MessU{\Label}{\Expr[v]}\#\Queue} $, $ \PT' = \MQ{\Chan}{\Role_1}{\Role_2}{\Queue} $.
			By the typing rules for message queues, $ \Gamma \vdash \PT \triangleright \Delta $ implies that there are $ \Sort, \MT $ such that $ \Delta = \MQ{\Chan}{\Role_1}{\Role_2}{\MessU{\Label'}{\Sort}\#\MT} $, $ \Label \compL \Label' $, $ \Typed{\Label'}{\Sort} \in \Gamma $, and $ \Gamma \vdash \PT' \triangleright \Delta' $ with $ \Delta' = \MQ{\Chan}{\Role_1}{\Role_2}{\MT} $.
			\\
			By Rule (\textsf{ML}) of Figure~\ref{fig:reductionSessionEnv}, then $ \Delta \stackrel{\Chan}{\Mapsto} \Delta' $.
			Since $ \Gamma, \Delta $ are coherent and by transitivity of $ \stackrel{\Chan}{\Mapsto} $, then $ \Gamma, \Delta' $ are coherent.}
		%%%%%%%%%%
		%  RSel  %
		%%%%%%%%%%
		\item[Case of Rule~(\textsf{RSel})] In this case $ \PT = \PPar{\PSelR{\Chan}{\Role_1}{\Role_2}{\Label}{\PT[Q]}}{\MQ{\Chan}{\Role_1}{\Role_2}{\Queue}} $, and $ \PT' = \PPar{\PT[Q]}{\MQ{\Chan}{\Role_1}{\Role_2}{\Queue\#\MessBR{\Label}}} $.
			By the Rules~(\textsf{Par}), (\textsf{RSel}), and the typing rules for message queues, $ \Gamma \vdash \PT \triangleright \Delta $ implies that there are $ \Delta_{\PT[Q]}, \indexSet, j, \MT $ and for all $ i \in \indexSet $ there are $ \Label_i, \LT_i $ such that $ \Delta = \Delta_{\PT[Q]} \compS \Typed{\AT{\Chan}{\Role_1}}{\LSelR{\Role_2}{\Set{ \Label_i, \LT_i }_{i \in \indexSet}}} \compS \MQ{\Chan}{\Role_1}{\Role_2}{\MT} $, $ j \in \indexSet $, $ \Label \compL \Label_j $, $ \Gamma \vdash \PT[Q] \triangleright \Delta_{\PT[Q]} \compS \Typed{\AT{\Chan}{\Role_1}}{\LT_j} $, and $ \Gamma \vdash \MQ{\Chan}{\Role_1}{\Role_2}{\Queue} \triangleright \MQ{\Chan[s]}{\Role_1}{\Role_2}{\MT} $.
			By Rule~(\textsf{MQBranR}), then $ \Gamma \vdash \MQ{\Chan}{\Role_1}{\Role_2}{\Queue\#\MessBR{\Label}} \triangleright \MQ{\Chan}{\Role_1}{\Role_2}{\MT\#\MessBR{\Label_j}} $.
			Since $ \Delta_{\PT[Q]} \compS \Typed{\AT{\Chan}{\Role_1}}{\LSelR{\Role_2}{\Set{ \Label_i, \LT_i }_{i \in \indexSet}}} \compS \MQ{\Chan}{\Role_1}{\Role_2}{\MT} $ is defined, so is $ \Delta' = \Delta_{\PT[Q]} \compS \Typed{\AT{\Chan}{\Role_1}}{\LT_j} \compS \MQ{\Chan}{\Role_1}{\Role_2}{\MT\#\MessBR{\Label_j}} $.
			By Rule~(\textsf{Par}), then $ \Gamma \vdash \PT' \triangleright \Delta' $.
			\\
			By Rule~(\textsf{RSel}) of Figure~\ref{fig:reductionSessionEnv}, then $ \Delta \stackrel{\Chan}{\Mapsto} \Delta' $.
			Since $ \Gamma, \Delta $ are coherent and by transitivity of $ \stackrel{\Chan}{\Mapsto} $, then $ \Gamma, \Delta' $ are coherent.
		%%%%%%%%%%%
		%  RBran  %
		%%%%%%%%%%%
		\item[Case of Rule~(\textsf{RBran})] In this case $ \PT = \PPar{\PBranR{\Chan}{\Role_1}{\Role_2}{\Set{ \Label_i.\PT[Q]_i }_{i \in \indexSet_1}}}{\MQ{\Chan}{\Role_2}{\Role_1}{\MessBR{\Label}\#\Queue}} $, $ j \in \indexSet_1 $, $ \Label \compL \Label_j $, and $ \PT' = \PPar{\PT[Q]_j}{\MQ{\Chan}{\Role_2}{\Role_1}{\Queue}} $.
			By the Rules~(\textsf{Par}), (\textsf{RBran}), and the typing rules for message queues, $ \Gamma \vdash \PT \triangleright \Delta $ implies that there are $ \Delta_{\PT[Q]}, \indexSet_2, \MT, \Label'' $ and for all $ i \in \indexSet_2 $ there are $ \Label'_i, \LT_i $ such that $ \Delta = \Delta_{\PT[Q]} \compS \Typed{\AT{\Chan}{\Role_1}}{\LBranR{\Role_2}{\Set{ \Label'_i, \LT_i }_{i \in \indexSet_2}}} \compS \MQ{\Chan}{\Role_2}{\Role_1}{\MessBR{\Label''}\#\MT} $, $ \Gamma \vdash \MQ{\Chan}{\Role_2}{\Role_1}{\Queue} \triangleright \MQ{\Chan[s]}{\Role_2}{\Role_1}{\MT} $, $ \Label_j \compL \Label'' $, and for all $ k \in \indexSet_2 $ exists some $ m \in \indexSet_1 $ such that $ \Label_m \compL \Label'_k $, $ \Gamma \vdash \PT[Q]_m \triangleright \Delta_{\PT[Q]} \compS \Typed{\AT{\Chan}{\Role_1}}{\LT_k} $.
			Since $ \Label \compL \Label_j \compL \Label'' $ and because $ \Gamma, \Delta $ are coherent, there is some $ n \in \indexSet_2 $ such that $ \Label_j \compL \Label'' = \Label_n $ and $ \Gamma \vdash \PT[Q]_j \triangleright \Delta_{\PT[Q]} \compS \Typed{\AT{\Chan}{\Role_1}}{\LT_n} $.
			Since $ \Delta_{\PT[Q]} \compS \Typed{\AT{\Chan}{\Role_1}}{\LBranR{\Role_2}{\Set{ \Label_i, \LT_i }_{i \in \indexSet}}} \compS \MQ{\Chan}{\Role_2}{\Role_1}{\MessBR{\Label''}\#\MT} $ is defined, so is $ \Delta' = \Delta_{\PT[Q]} \compS \Typed{\AT{\Chan}{\Role_1}}{\LT_n} \compS \MQ{\Chan}{\Role_2}{\Role_1}{\MT} $.
			By Rule~(\textsf{Par}), then $ \Gamma \vdash \PT' \triangleright \Delta' $.
			\\
			By Rule~(\textsf{RBran}) of Figure~\ref{fig:reductionSessionEnv}, then $ \Delta \stackrel{\Chan}{\Mapsto} \Delta' $.
			Since $ \Gamma, \Delta $ are coherent and by transitivity of $ \stackrel{\Chan}{\Mapsto} $, then $ \Gamma, \Delta' $ are coherent.
		%%%%%%%%%%
		%  WSel  %
		%%%%%%%%%%
		\textcolor{blue}{\item[Case of Rule~(\textsf{WSel})] In this case $ \PT = \PPar{\PSelW{\Chan}{\Role}{\Role[R]}{\Label}{\PT[Q]}}{\prod_{\Role_i \in \Role[R]} \MQ{\Chan}{\Role}{\Role_i}{\Queue}} $ and we have $ \PT' = \PPar{\PT[Q]}{\prod_{\Role_i \in \Role[R]} \MQ{\Chan}{\Role}{\Role_i}{\Queue_i\#\MessBW{\Label}}} $.
			Let $ \Role[R] = \Set{ \Role_1, \ldots, \Role[n] } $.
			By the Rules~(\textsf{Par}), (\textsf{WSel}), and the typing rules for message queues, $ \Gamma \vdash \PT \triangleright \Delta $ implies that there are $ \Delta_{\PT[Q]}, \indexSet, j, \MT_1, \ldots, \MT_n $ and for all $ i \in \indexSet $ there are $ \Label_i, \LT_i $ such that $ \Delta = \Delta_{\PT[Q]} \compS \Typed{\AT{\Chan}{\Role}}{\LSelW{\Role[R]}{\Set{ \Label_i, \LT_i }_{i \in \indexSet}}} \compS \MQ{\Chan}{\Role}{\Role_1}{\MT} \compS \ldots \compS \MQ{\Chan}{\Role}{\Role_n}{\MT} $, $ j \in \indexSet $, $ \Label \compL \Label_j $, $ \Gamma \vdash \PT[Q] \triangleright \Delta_{\PT[Q]} \compS \Typed{\AT{\Chan}{\Role}}{\LT_j} $, and $ \Gamma \vdash \MQ{\Chan}{\Role}{\Role_1}{\Queue_1} \triangleright \MQ{\Chan[s]}{\Role}{\Role_1}{\MT_1} $, \ldots, $ \Gamma \vdash \MQ{\Chan}{\Role}{\Role_n}{\Queue_n} \triangleright \MQ{\Chan[s]}{\Role}{\Role_n}{\MT_n} $.
			By Rule~(\textsf{MQBranW}), then $ \Gamma \vdash \MQ{\Chan}{\Role}{\Role_1}{\Queue_1\#\MessBW{\Label}} \triangleright \MQ{\Chan}{\Role}{\Role_1}{\MT_1\#\MessBW{\Label_j}} $, \ldots, $ \Gamma \vdash \MQ{\Chan}{\Role}{\Role_n}{\Queue_n\#\MessBW{\Label}} \triangleright \MQ{\Chan}{\Role}{\Role_n}{\MT_n\#\MessBW{\Label_j}} $.
			Since $ \Delta_{\PT[Q]} \compS \Typed{\AT{\Chan}{\Role}}{\LSelW{\Role[R]}{\Set{ \Label_i, \LT_i }_{i \in \indexSet}}} \compS \MQ{\Chan}{\Role}{\Role_1}{\MT_1} \compS \ldots \compS \MQ{\Chan}{\Role}{\Role_n}{\MT_n} $ is defined, so is $ \Delta' = \Delta_{\PT[Q]} \compS \Typed{\AT{\Chan}{\Role}}{\LT_j} \compS \MQ{\Chan}{\Role}{\Role_1}{\MT_1\#\MessBW{\Label_j}} \compS \ldots \compS \MQ{\Chan}{\Role}{\Role_n}{\MT_n\#\MessBW{\Label_j}} $.
			By Rule~(\textsf{Par}), then $ \Gamma \vdash \PT' \triangleright \Delta' $.
			\\
			By Rule~(\textsf{WSel}) of Figure~\ref{fig:reductionSessionEnv}, then $ \Delta \stackrel{\Chan}{\Mapsto} \Delta' $.
			Since $ \Gamma, \Delta $ are coherent and by transitivity of $ \stackrel{\Chan}{\Mapsto} $, then $ \Gamma, \Delta' $ are coherent.}
		%%%%%%%%%%%
		%  WBran  %
		%%%%%%%%%%%
		\textcolor{blue}{\item[Case of Rule~(\textsf{WBran})] Here $ \PT = \PPar{\PBranW{\Chan}{\Role_1}{\Role_2}{\Set{ \Label_i.\PT[Q]_i }_{i \in \indexSet_1, \LabelD}}}{\MQ{\Chan}{\Role_2}{\Role_1}{\MessBW{\Label}\#\Queue}} $, $ j \in \indexSet_1 $, $ \Label \compL \Label_j $, and $ \PT' = \PPar{\PT[Q]_j}{\MQ{\Chan}{\Role_2}{\Role_1}{\Queue}} $.
			By the Rules~(\textsf{Par}), (\textsf{WBran}), and the typing rules for message queues, $ \Gamma \vdash \PT \triangleright \Delta $ implies that there are $ \Delta_{\PT[Q]}, \LabelD', \indexSet_2, \MT, \Label'' $ and for all $ i \in \indexSet_2 $ there are $ \Label'_i, \LT_i $ such that $ \Delta = \Delta_{\PT[Q]} \compS \Typed{\AT{\Chan}{\Role_1}}{\LBranW{\Role_2}{\Set{ \Label'_i, \LT_i }_{i \in \indexSet_2, \LabelD'}}} \compS \MQ{\Chan}{\Role_2}{\Role_1}{\MessBW{\Label''}\#\MT} $, $ \LabelD \compL \LabelD' $, $ \Gamma \vdash \MQ{\Chan}{\Role_2}{\Role_1}{\Queue} \triangleright \MQ{\Chan[s]}{\Role_2}{\Role_1}{\MT} $, $ \Label_j \compL \Label'' $, and for all $ k \in \indexSet_2 $ exists some $ m \in \indexSet_1 $ such that $ \Label_m \compL \Label'_k $, $ \Gamma \vdash \PT[Q]_m \triangleright \Delta_{\PT[Q]} \compS \Typed{\AT{\Chan}{\Role_1}}{\LT_k} $.
			Since $ \Label \compL \Label_j \compL \Label'' $ and because $ \Gamma, \Delta $ are coherent, there is some $ n \in \indexSet_2 $ such that $ \Label_j \compL \Label'' = \Label_n $ and $ \Gamma \vdash \PT[Q]_j \triangleright \Delta_{\PT[Q]} \compS \Typed{\AT{\Chan}{\Role_1}}{\LT_n} $.
			Since $ \Delta_{\PT[Q]} \compS \Typed{\AT{\Chan}{\Role_1}}{\LBranW{\Role_2}{\Set{ \Label_i, \LT_i }_{i \in \indexSet, \LabelD'}}} \compS \MQ{\Chan}{\Role_2}{\Role_1}{\MessBW{\Label''}\#\MT} $ is defined, so is $ \Delta' = \Delta_{\PT[Q]} \compS \Typed{\AT{\Chan}{\Role_1}}{\LT_n} \compS \MQ{\Chan}{\Role_2}{\Role_1}{\MT} $.
			By Rule~(\textsf{Par}), then $ \Gamma \vdash \PT' \triangleright \Delta' $.
			\\
			By Rule~(\textsf{WBran}) of Figure~\ref{fig:reductionSessionEnv}, then $ \Delta \stackrel{\Chan}{\Mapsto} \Delta' $.
			Since $ \Gamma, \Delta $ are coherent and by transitivity of $ \stackrel{\Chan}{\Mapsto} $, then $ \Gamma, \Delta' $ are coherent.}
		%%%%%%%%%%%
		%  WSkip  %
		%%%%%%%%%%%
		\textcolor{blue}{\item[Case of Rule~(\textsf{WSkip})] In this case $ \PT = \PBranW{\Chan}{\Role_1}{\Role_2}{\Set{ \Label_i.\PT[Q]_i }_{i \in \indexSet_1, \LabelD}} $ and $ \PT' = \PT[Q]_{\default} $.
			By the Rule(\textsf{WBran}), $ \Gamma \vdash \PT \triangleright \Delta $ implies that there are $ \Delta_{\PT[Q]}, \LabelD', \indexSet_2 $ and for all $ i \in \indexSet_2 $ there are $ \Label'_i, \LT_i $ such that $ \Delta = \Delta_{\PT[Q]} \compS \Typed{\AT{\Chan}{\Role_1}}{\LBranW{\Role_2}{\Set{ \Label'_i, \LT_i }_{i \in \indexSet_2, \LabelD'}}} $, $ \LabelD \compL \LabelD' $, and for all $ k \in \indexSet_2 $ exists some $ m \in \indexSet_1 $ such that $ \Label_m \compL \Label'_k $, $ \Gamma \vdash \PT[Q]_m \triangleright \Delta_{\PT[Q]} \compS \Typed{\AT{\Chan}{\Role_1}}{\LT_k} $.
			Since $ \LabelD \compL \LabelD' $, there is some $ n \in \indexSet_2 $ such that $ \LabelD \compL \LabelD' = \Label_n $ and $ \Gamma \vdash \PT' \triangleright \Delta' $ with $ \Delta' = \Delta_{\PT[Q]} \compS \Typed{\AT{\Chan}{\Role_1}}{\LT_{\default}} $.
			\\
			By Rule~(\textsf{WSkip}) of Figure~\ref{fig:reductionSessionEnv}, then $ \Delta \stackrel{\Chan}{\Mapsto} \Delta' $.
			Since $ \Gamma, \Delta $ are coherent and by transitivity of $ \stackrel{\Chan}{\Mapsto} $, then $ \Gamma, \Delta' $ are coherent.}
		%%%%%%%%%%%
		%  Crash  %
		%%%%%%%%%%%
		\textcolor{blue}{\item[Case of Rule~(\textsf{Crash})] In this case $ \fpCrash $ and $ \PT' = \PCrash $.
			By the typing rules and Condition~\ref{cond:all}.\ref{cond:crash}, $ \Gamma \vdash \PT \triangleright \Delta $ and $ \fpCrash $ imply $ \Unreliable{\Delta} $.
			By Rule~(\textsf{Crash}), then $ \Gamma \vdash \PCrash \triangleright \Delta $.
			\\
			Since $ \Delta' = \Delta $, $ \Gamma, \Delta' $ is coherent and, by reflexivity, $ \Delta \stackrel{\Chan}{\Mapsto} \Delta' $.}
		%%%%%%%%%%
		%  If-T  %
		%%%%%%%%%%
		\item[Case of Rule~(\textsf{If-T})] In this case $ \PT = \PITE{\Expr}{\PT[Q]}{\PT[Q]'} $ and $ \PT' = \PT[Q] $.
			By Rule~(\textsf{If}), $ \Gamma \vdash \PT \triangleright \Delta $ implies $ \Gamma \vdash \PT' \triangleright \Delta $.
			\\
			Since $ \Delta' = \Delta $, $ \Gamma, \Delta' $ is coherent and, by reflexivity, $ \Delta \stackrel{\Chan}{\Mapsto} \Delta' $.
		%%%%%%%%%%
		%  If-F  %
		%%%%%%%%%%
		\item[Case of Rule~(\textsf{If-F})] In this case $ \PT = \PITE{\Expr}{\PT[Q]}{\PT[Q]'} $ and $ \PT' = \PT[Q]' $.
			By Rule~(\textsf{If}), $ \Gamma \vdash \PT \triangleright \Delta $ implies $ \Gamma \vdash \PT' \triangleright \Delta $.
			\\
			Since $ \Delta' = \Delta $, $ \Gamma, \Delta' $ is coherent and, by reflexivity, $ \Delta \stackrel{\Chan}{\Mapsto} \Delta' $.
		%%%%%%%%%%%
		%  Deleg  %
		%%%%%%%%%%%
		\item[Case of Rule~(\textsf{Deleg})] In this case $ \PT = \PDelA{\Chan}{\Role_1}{\Role_2}{\AT{\Chan'}{\Role}}{\PT[Q]} \mid \MQ{\Chan}{\Role_1}{\Role_2}{\Queue} $ and $ \PT' = \PT[Q] \mid \MQ{\Chan}{\Role_1}{\Role_2}{\Queue\#\AT{\Chan'}{\Role}} $.
			By the Rules~(\textsf{Par}), (\textsf{Deleg}), and the typing rules for message queues, $ \Gamma \vdash \PT \triangleright \Delta $ implies that there are $ \Delta_{\PT[Q]}, \LT, \LT', \MT $ such that $ \Delta = \Delta_{\PT[Q]} \compS \Typed{\AT{\Chan}{\Role_1}}{\LDelA{\Role2}{\Chan'}{\Role}{\LT'}{\LT}} \compS \Typed{\AT{\Chan'}{\Role}}{\LT'} \compS \MQ{\Chan}{\Role_1}{\Role_2}{\MT} $, $ \Gamma \vdash \PT[Q] \triangleright \Delta_{\PT[Q]} \compS \Typed{\AT{\Chan}{\Role_1}}{\LT} $, and $ \Gamma \vdash \MQ{\Chan}{\Role_1}{\Role_s}{\MT} \triangleright \MQ{\Chan}{\Role_1}{\Role_2}{\MT} $.
			By Rule~(\textsf{MQDeleg}), then $ \Gamma \vdash \MQ{\Chan}{\Role_1}{\Role_s}{\MT\#\AT{\Chan'}{\Role}} \triangleright \MQ{\Chan}{\Role_1}{\Role_2}{\MT\#\AT{\Chan'}{\Role}} $.
			Since $ \Delta = \Delta_{\PT[Q]} \compS \Typed{\AT{\Chan}{\Role_1}}{\LDelA{\Role2}{\Chan'}{\Role}{\LT'}{\LT}} \compS \Typed{\AT{\Chan'}{\Role}}{\LT'} \compS \MQ{\Chan}{\Role_1}{\Role_2}{\MT} $ is defined, so is $ \Delta' = \Delta_{\PT[Q]} \compS \Typed{\AT{\Chan}{\Role_1}}{\LT} \compS \MQ{\Chan}{\Role_1}{\Role_2}{\MT\#\AT{\Chan'}{\Role}} $.
			By Rule~(\textsf{Par}), then $ \Gamma \vdash \PT' \triangleright \Delta' $.
			\\
			By Rule~(\textsf{Deleg}) of Figure~\ref{fig:reductionSessionEnv}, then $ \Delta \stackrel{\Chan}{\Mapsto} \Delta' $.
			Since $ \Gamma, \Delta $ are coherent and by transitivity of $ \stackrel{\Chan}{\Mapsto} $, then $ \Gamma, \Delta' $ are coherent.
		%%%%%%%%%%%
		%  SRecv  %
		%%%%%%%%%%%
		\item[Case of Rule~(\textsf{SRecv})] In this case $ \PT = \PDelB{\Chan}{\Role_1}{\Role_2}{\AT{\Chan'}{\Role}}{\PT[Q]} \mid \MQ{\Chan}{\Role_2}{\Role_1}{\AT{\Chan''}{\Role'}\#\Queue} $ and $ \PT' = \PT[Q]\Subst{\Chan''}{\Chan'}\Subst{\Role'}{\Role} \mid \Queue $. We use alpha conversion to ensure that $ \Chan' = \Chan'' $ and $ \Role = \Role'' $.
			By the Rules~(\textsf{Par}), (\textsf{SRecv}), and the typing rules for message queues, $ \Gamma \vdash \PT \triangleright \Delta $ implies that there are $ \Delta_{\PT[Q]}, \LT, \LT', \MT $ such that $ \Delta = \Delta_{\PT[Q]} \compS \Typed{\AT{\Chan}{\Role_1}}{\LDelB{\Role_2}{\Chan'}{\Role}{\LT'}{\LT}} \compS \MQ{\Chan}{\Role_2}{\Role_1}{\AT{\Chan'}{\Role}\#\MT} $, $ \Gamma \vdash \PT[Q] \triangleright \Delta_{\PT[Q]} \compS \Typed{\AT{\Chan}{\Role_1}}{\LT} \compS \Typed{\AT{\Chan'}{\Role}}{\LT'} $, and $ \Gamma \vdash \MQ{\Chan}{\Role_2}{\Role_1}{\Queue} \triangleright \MQ{\Chan}{\Role_2}{\Role_1}{\MT} $.
			Since $ \Delta = \Delta_{\PT[Q]} \compS \Typed{\AT{\Chan}{\Role_1}}{\LDelB{\Role_2}{\Chan'}{\Role}{\LT'}{\LT}} \compS \MQ{\Chan}{\Role_2}{\Role_1}{\AT{\Chan''}{\Role'}\#\MT} $ is defined, so is $ \Delta' = \Delta_{\PT[Q]} \compS \Typed{\AT{\Chan}{\Role_1}}{\LT} \compS \Typed{\AT{\Chan'}{\Role}}{\LT'} \compS \MQ{\Chan}{\Role_2}{\Role_1}{\MT} $.
			By Rule~(\textsf{Par}), then $ \Gamma \vdash \PT' \triangleright \Delta' $.
			\\
			By Rule~(\textsf{SRecv}) of Figure~\ref{fig:reductionSessionEnv}, then $ \Delta \stackrel{\Chan}{\Mapsto} \Delta' $.
			Since $ \Gamma, \Delta $ are coherent and by transitivity of $ \stackrel{\Chan}{\Mapsto} $, then $ \Gamma, \Delta' $ are coherent.
		%%%%%%%%%
		%  Par  %
		%%%%%%%%%
		\item[Case of Rule~(\textsf{Par})] In this case $ \PT = \PPar{\PT[Q]_1}{\PT[Q]_2} $, $ \PT[Q]_1 \step \PT[Q]_1' $, and $ \PT' = \PPar{\PT[Q]_1'}{\PT[Q]_2} $.
			By Rule~(\textsf{Par}), $ \Gamma \vdash \PT \triangleright \Delta $ implies that there are $ \Delta_{\PT[Q]_1}, \Delta_{\PT[Q]_2} $ such that $ \Delta = \Delta_{\PT[Q]_1} \compS \Delta_{\PT[Q]_2} $, $ \Gamma \vdash \PT[Q]_1 \triangleright \Delta_{\PT[Q]_1} $, and $ \Gamma \vdash \PT[Q]_2 \triangleright \Delta_{\PT[Q]_2} $.
			Since $ \Gamma, \Delta $ are coherent, $ \Gamma, \Delta_{\PT[Q]_1} $ is weakly coherent.
			By the induction hypothesis (for the weakly coherent case), $ \Gamma \vdash \PT[Q]_1' \triangleright \Delta_{\PT[Q]_1}' $ with $ \Delta_{\PT[Q]_1} \stackrel{s}{\Mapsto} \Delta_{\PT[Q]_1}' $.
			Since $ \Delta_{\PT[Q]_1} \compS \Delta_{\PT[Q]_2} $ is defined, so is $ \Delta' = \Delta_{\PT[Q]_1} \compS \Delta_{\PT[Q]_2} $.
			By Rule~(\textsf{Par}), then $ \Gamma \vdash \PT' \triangleright \Delta' $.
			\\
			By $ \Delta_{\PT[Q]_1} \stackrel{s}{\Mapsto} \Delta_{\PT[Q]_1}' $, then $ \Delta \stackrel{\Chan}{\Mapsto} \Delta' $.
			Since $ \Gamma, \Delta $ are coherent and by transitivity of $ \stackrel{\Chan}{\Mapsto} $, then $ \Gamma, \Delta' $ are coherent.
		%%%%%%%%%
		%  Res  %
		%%%%%%%%%
		\item[Case of Rule~(\textsf{Res})] In this case $ \PT = \PRes{\Args}{\PT[Q]} $, $ \PT[Q] \step \PT[Q'] $, and $ \PT' = \PRes{\Args}{\PT[Q]'} $.
			Then, one of the Rules~(\textsf{Res1}) or (\textsf{Res2}) was used to type the restriction on $ \Args $ in $ \PT $.
			\begin{description}
				%%%%%%%%%%
				%  Res1  %
				%%%%%%%%%%
				\item[Case of (\textsf{Res1})] Then there is some $ \Sort $ such that $ \Args \sharp \left( \Gamma, \Delta \right) $ and $ \Gamma \compS \Typed{\Args}{\Sort} \vdash \PT[Q] \triangleright \Delta $.
					Since $ \Gamma, \Delta $ are coherent, so are $ \Gamma  \compS \Typed{\Args}{\Sort}, \Delta $.
					By the induction hypothesis, $ \Gamma \compS \Typed{\Args}{\Sort} \vdash \PT[Q]' \triangleright \Delta' $ for some $ \Delta' $ such that $ \Gamma, \Delta' $ is coherent and $ \Delta \stackrel{\Chan}{\Mapsto} \Delta' $.
					With Rule~(\textsf{Res1}), then $ \Gamma \vdash \PT' \triangleright \Delta' $.
				%%%%%%%%%%
				%  Res2  %
				%%%%%%%%%%
				\item[Case of (\textsf{Res2})] In this case $ \Args = \Chan $ and there are $ \GT, \Chan[a], \Delta'' $ such that we have $ \Chan \sharp \left(\Gamma, \Delta \right) $, $ \Set{ \Typed{\AT{\Chan}{\Role}}{\Proj{\GT}{\Role}} \mid \Role \in \Roles{\GT} } \compS \Set{ \MQ{\Chan}{\Role}{\Role'}{\emptyList} \mid \Role, \Role' \in \Roles{\GT'} \wedge \Role \neq \Role' } \stackrel{\Chan}{\Mapsto} \Delta'' $, $ \Typed{\Chan[a]}{G} \in \Gamma $, and $ \Gamma \vdash \PT[Q] \triangleright \Delta \compS \Delta'' $.
					Then $ \Gamma, \Delta \compS \Delta'' $ are coherent.
					By the induction hypothesis, $ \Gamma \vdash \PT[Q]' \triangleright \Delta''' $ for some $ \Delta''' $ such that $ \Delta \compS \Delta'' \stackrel{\Chan}{\Mapsto} \Delta''' $.
					By Rule~(\textsf{Res2}), then $ \Gamma \vdash \PT' \triangleright \Delta $.
					Since $ \Delta' = \Delta $, $ \Gamma, \Delta' $ is coherent and, by reflexivity, $ \Delta \stackrel{\Chan}{\Mapsto} \Delta' $.
			\end{description}
		%%%%%%%%%
		%  Rec  %
		%%%%%%%%%
		\item[Case of (\textsf{Rec})] In this case $ \PT = \PRep{\ProcV}{\PT[Q]} $ and $ \PT' = \PT[Q]\Subst{\PRep{\ProcV}{\PT[Q]}}{\ProcV} $.
			By Rule~(\textsf{Rec}), then there are $ \Delta_{\PT[Q]}, \Chan, \Role, \TypeV, \LT $ such that $ \Delta = \Delta_{\PT[Q]} \compS \Typed{\AT{\Chan}{\Role}}{\LRep{\TypeV}{\LT}} $ and $ \Gamma \compS \Typed{\ProcV}{\Typed{\AT{\Chan}{\Role}}{\TypeV}} \vdash \PT[Q] \triangleright \Delta_{\PT[Q]} \compS \Typed{\AT{\Chan}{\Role}}{\LT} $.
			By replacing in the proof tree of $ \Gamma \compS \Typed{\ProcV}{\Typed{\AT{\Chan}{\Role}}{\TypeV}} \vdash \PT[Q] \triangleright \Delta_{\PT[Q]} \compS \Typed{\AT{\Chan}{\Role}}{\LT} $ all occurrences of Rule~(\textsf{Var}) by the proof tree, we obtain $ \Gamma \vdash \PT' \triangleright \Delta' $ with $ \Delta' = \Delta_{\PT[Q]} \compS \Typed{\AT{\Chan}{\Role}}{\LT\Subst{\LRep{\TypeV}{\LT}}{\TypeV}} $.
			\\
			By Rule~(\textsf{Rec}) of Figure~\ref{fig:reductionSessionEnv}, then $ \Delta \stackrel{\Chan}{\Mapsto} \Delta' $.
			Since $ \Gamma, \Delta $ are coherent and by transitivity of $ \stackrel{\Chan}{\Mapsto} $, then $ \Gamma, \Delta' $ are coherent.
		%%%%%%%%%%%
		%  Struc  %
		%%%%%%%%%%%
		\item[Case of Rule~(\textsf{Struc})] In this case $ \PT \equiv \PT[Q] $, $ \PT[Q] \step \PT[Q]' $, $ \PT[Q]' \equiv \PT' $.
			By Lemma~\ref{lem:structuralCongruencePreservesJudgement}, $ \Gamma \vdash \PT \triangleright \Delta $ and $ \PT \equiv \PT[Q] $ imply $ \Gamma \vdash \PT[Q] \triangleright \Delta $.
			By the induction hypothesis, $ \Gamma \vdash \PT[Q]' \triangleright \Delta' $ for some $ \Delta' $ such that $ \Delta \stackrel{\Chan}{\Mapsto} \Delta' $ and $ \Gamma, \Delta' $ are coherent.
			By Lemma~\ref{lem:structuralCongruencePreservesJudgement}, then $ \Gamma \vdash \PT' \triangleright \Delta' $.
			\qedhere
	\end{description}
\end{proof}

Since we restrict our attention to linear environments, type judgements ensure linearity of session channels.
With subject reduction, this holds for all derivatives of well-typed processes.

\begin{lem}[Linearity]
	\label{lem:linearity}
	Let $ \Gamma \vdash \PT \triangleright \Delta $, $ \Gamma, \Delta $ be coherent, and there are no name clashes on session channels.
	Then all session channels of $ \PT $ are linear, \ie for all $ \PT \steps \PT' $ and all $ \Chan[s], \Role, \Role_1, \Role_2 $ there is at most one unguarded actor $ \AT{\Chan}{\Role} $ and at most one queue $ \MQS{\Chan}{\Role_1}{\Role_2} $ in $ \PT' $.
\end{lem}

\begin{proof}
	By Theorem~\ref{thm:subjectReduction}, there is some $ \Delta' $ such that $ \Gamma \vdash \PT' \triangleright \Delta' $ and $ \Gamma, \Delta' $ are coherent.
	By the Definition~\ref{def:coherence} of coherence and projection, $ \Delta' $ contains at most one actor $ \AT{\Chan}{\Role} $ and at most one queue $ \MQS{\Chan}{\Role_1}{\Role_2} $ for each $ \Typed{\Chan[a]}{\GT} \in \Gamma $ and $ \Role \in \Roles{\GT} $.
	By the Figures~\ref{fig:typingRules} and \ref{fig:runtimeTypingRules}, only the Rules~(\textsf{Req}), (\textsf{Acc}), and (\textsf{Res2}) can introduce new actors or queues.
	The linearity of global environments ensures, that all new actors and queues introduced by the rules are on fresh channel names and are pairwise distinct.
	The Rules~(\textsf{Req}) and (\textsf{Acc}) introduce exactly one actor each on a fresh session channel $ \Chan $ that is bound by a prefix for session initialisation.
	Rule~(\textsf{Res2}) introduces assignments for actors and queues for pairwise different roles on a fresh session channel $ \Chan $ that is bound by restriction.
	Since there are no name clashes, the session channels in binders are pairwise different and distinct from free session channels.
	By the typing rules and because $ \Gamma \vdash \PT' \triangleright \Delta' $, all actors and queues in $ \PT' $ have to satisfy their specification as described by an assignment of this actor or queue towards a local type.
	By the linearity of session environments and since new assignments for actors result from bound session channels, all unguarded actors and queues in $ \PT' $ are pairwise different.
\end{proof}

For \strongR systems coherence ensures that for each actor there is a matching communication partner.
In the case of asynchronous communication, this means that for each sender (or message on a queue) there is a receiver and for each receiver there is a sender or a message on a queue, where the receiver as well as the sender or the message queue appear under the same binder of the session channel or both are free.
In the case of \unrel communication, messages get lost, senders can crash, and receivers can crash themselves or suspect the sender.
In the case of \weakR branching for each sender (or message on a queue) there are all specified receivers that are not crashed and vice versa.

We summarise these properties of \strongR and \weakR interactions in \emph{error-freedom}: for each \strongR sender or message there is a matching receiver and vice versa, for each \weakR sender or message there is a possibly crashed receiver and vice versa.
We obtain similar requirements for session delegation.

\begin{lem}[Error-Freedom]
	\label{lem:errorFredom}
	If $ \Gamma \vdash \PT \triangleright \Delta $ and $ \Gamma, \Delta $ is coherent then:
	\begin{compactitem}
		\item for each unguarded $ \PSendR{\Chan}{\Role_1}{\Role_2}{\Expr[y]}{\PT[Q]_1} $ and each message $ \MessR{\Expr[y]} $ on a message queue $ \MQS{\Chan}{\Role_1}{\Role_2} $ in $ \PT $ there is some $ \PGetR{\Chan}{\Role_2}{\Role_1}{\Args}{\PT[Q]_2} $ in $ \PT $,
		\item for each unguarded $ \PGetR{\Chan}{\Role_2}{\Role_1}{\Args}{\PT[Q]_2} $ in $ \PT $ there is some $ \PSendR{\Chan}{\Role_1}{\Role_2}{\Expr[y]}{\PT[Q]_1} $ or a message $ \MessR{\Expr[y]} $ on a message queue $ \MQS{\Chan}{\Role_1}{\Role_2} $ in $ \PT $,
		\item for each unguarded $ \PSelR{\Chan}{\Role_1}{\Role_2}{\Label}{\PT[Q]} $ and each message $ \MessBR{\Label} $ on a message queue $ \MQS{\Chan}{\Role_1}{\Role_2} $ in $ \PT $ there is some $ j \in \indexSet $ and $ \PBranR{\Chan}{\Role_2}{\Role_1}{\Set{ \Label_i.\PT[Q]_i }_{i \in \indexSet}} $ in $ \PT $ with $ \Label_j \compL \Label $,
		\item for each unguarded $ \PBranR{\Chan}{\Role_2}{\Role_1}{\Set{ \Label_i.\PT[Q]_i }_{i \in \indexSet}} $ in $ \PT $ there is $ j \in \indexSet $ and $ \PSelR{\Chan}{\Role_1}{\Role_2}{\Label}{\PT[Q]} $ or a message $ \MessBR{\Label} $ on a message queue $ \MQS{\Chan}{\Role_1}{\Role_2} $ in $ \PT $ with $ \Label_j \compL \Label $,
		\item for each unguarded $ \PSelW{\Chan}{\Role}{\Role[R]}{\Label}{\PT[Q]} $ and each message $ \MessBW{\Label} $ on a message queue $ \MQS{\Chan}{\Role}{\Role'} $ in $ \PT $ and each $ \Role' \in \Role[R] $ there is some $ \PBranW{\Chan}{\Role'}{\Role}{\Set{ \Label_i.\PT_i }_{i \in \indexSet, \LabelD}} $ and $ j \in \indexSet $ in $ \PT $ with $ \Label_j \compL \Label $ or $ \PT $ does not contain an actor $ \AT{\Chan}{\Role'} $,
		\item for each unguarded $ \PBranW{\Chan[s]}{\Role'}{\Role}{\Set{ \Label_i.\PT_i }_{i \in \indexSet, \LabelD}} $ in $ \PT $ there is $ j \in \indexSet $ and $ \PSelW{\Chan}{\Role}{\Role[R]}{\Label}{\PT[Q]} $ or a message $ \MessBW{\Label} $ on a message queue $ \MQS{\Chan}{\Role}{\Role'} $ in $ \PT $ with $ \Label_j \compL \Label $ and $ \Role' \in \Role[R] $ or $ \PT $ does not contain an actor $ \AT{\Chan}{\Role} $,
		\item for each unguarded $ \PDelA{\Chan}{\Role_1}{\Role_2}{\AT{\Chan'}{\Role}}{\PT[Q]_1} $ and each message $ \AT{\Chan'}{\Role} $ on a message queue $ \MQS{\Chan}{\Role_1}{\Role_2} $ in $ \PT $ there is some $ \PDelB{\Chan}{\Role_2}{\Role_1}{\AT{\Chan''}{\Role'}}{\PT[Q]_2} $ in $ \PT $, and
		\item for each unguarded $ \PDelB{\Chan}{\Role_2}{\Role_1}{\AT{\Chan''}{\Role'}}{\PT[Q]_2} $ in $ \PT $ there is some $ \PDelA{\Chan}{\Role_1}{\Role_2}{\AT{\Chan'}{\Role}}{\PT[Q]_1} $ or a message $ \AT{\Chan'}{\Role} $ on a message queue $ \MQS{\Chan}{\Role_1}{\Role_2} $ in $ \PT $.
	\end{compactitem}
\end{lem}

\begin{proof}
	By coherence and projection for each \strongR and each \weakR sender there is initially a matching receiver for each free session channel in the session environment.
	By the typing rules and Rule~(\textsf{Res2}) in particular, this holds also for restricted session channels.
	Session environments may evolve using $ \stackrel{\Chan}{\mapsto} $ but all such steps preserve the above defined requirements, \ie \strongR or \weakR send prefixes can be mapped onto the type of the respective message in a queue but no such message can be dropped.
	The typing rules ensure that the processes follow their specification in the local types of session environments.
	Then, the first four and the last two cases follow from the typing rules and coherence, and the fact that only unreliable processes can crash.
	The remaining two cases follow from the typing rules and coherence.
\end{proof}

\emph{Session fidelity} claims that the interactions of a well-typed process follow exactly the specification described by its global types, \ie if a system is well-typed \wrt to coherent type environments then the system follows its specification in the global type.
One direction of this property already follows from the above variant of subject reduction.
The steps of well-typed systems are reflected by corresponding steps of the session environment and, thus, respect their specification in global types.
What remains to show is that the specified interactions can indeed be performed.
The above formulation of error-freedom alone is not strong enough to show this, because it ensures only the existence of matching communication partners and not that they can be unguarded.

To obtain session fidelity we prove \emph{progress}.
\emph{Progress} states that no part of a well-typed and coherent system can block other parts, that eventually all matching communication partners as described by error-freedom are unguarded, that interactions specified by the global type can happen, and that there are no communication mismatches.
Subject reduction and progress together then imply \emph{session fidelity}, \ie that processes behave as specified in their global types.

To ensure that the interleaving of sessions and session delegation cannot introduce deadlocks, we assume an interaction type system as introduced in \cite{BettiniAtall08,hondaYoshidaCarbone16}.
For this type system it does not matter whether the considered actions are \strongR, \weakR, or \unrel.
More precisely, we can adapt the interaction type system of \cite{BettiniAtall08} in a straightforward way to the above session calculus, where \unrel communication and \weakR branching is treated in exactly the same way as \strongR communication/branching.
We say that \emph{$ \PT $ is free of cyclic dependencies between sessions} if this interaction type system does not detect any cyclic dependencies.
In this sense fault-tolerance is more flexible than explicit failure handling, which oftens requires a more substantial revision of the interaction type system to cover the additional dependencies that are introduced \eg by the propagation of faults.

In the literature there are different formulations of progress.
We are interested in a rather strict definition of progress that ensures that well-typed systems cannot block.
Therefore, we need an additional assumption on session requests and acceptances.
Coherence ensures the existence of communication partners within sessions only.
If we want to avoid blocking, we need to be sure, that no participant of a session is missing during its initialisation.
Note that without action prefixes all participants either terminated or crashed.

\begin{thm}[Progress/Session Fidelity]
	\label{thm:progress}
	Let $ \Gamma \vdash \PT \triangleright \Delta $, $ \Gamma, \Delta $ be coherent, and let $ \PT $ be free of cyclic dependencies between sessions.
	Assume that in the derivation of $ \Gamma \vdash \PT \triangleright \Delta $, whenever $ \PReq{\Chan[a]}{\Role[n]}{\Chan}{\PT[Q]} $ or $ \PAcc{\Chan[a]}{\Role}{\Chan}{\PT[Q]} $ in $ \PT $, then $ \Typed{\Chan[a]}{\GT} \in \Gamma $, $ \Length{\Roles{G}} = \Role[n] $, and there are $ \PReq{\Chan[a]}{\Role[n]}{\Chan}{\PT[Q]_n} $ as well as $ \PAcc{\Chan[a]}{\Role_i}{\Chan}{\PT[Q]_i} $ in $ \PT $ for all $ 1 \leq \Role_i < \Role[n] $.
	\begin{compactenum}
		\item Then either $ \PT $ does not contain any action prefixes or $ \PT \step \PT' $.
		\item If $ \PT $ does not contain recursion, then there exists $ P' $ such that $ \PT \steps \PT' $ and $ \PT' $ does not contain any action prefixes.
	\end{compactenum}
\end{thm}

\begin{proof}
	1.) If $ \PT $ contains an unguarded conditional, then it can perform a step $ \PT \step \PT' $ such that $ \Gamma \vdash \PT' \triangleright \Delta $ using one of the Rules~(\textsf{If-T}) or (\textsf{If-F}) as described in the corresponding cases of the proof of Theorem~\ref{thm:subjectReduction}.
	Similarly, if $ \PT $ contains an unguarded recursion, then it can perform a step $ \PT \step \PT' $ such that $ \Gamma \vdash \PT' \triangleright \Delta' $ with $ \Delta \stackrel{\Chan}{\mapsto} \Delta' $ using Rule~(\textsf{Rec}) as described in the corresponding case of the proof of Theorem~\ref{thm:subjectReduction}.
	Assume that $ \PT $ is not structural congruent to $ \PEnd $ and does not contain unguarded conditionals or recursions.

	Since $ \PT $ is not structural congruent to $ \PEnd $ it contains session channels.
	All session channels of $ \PT $ that are not contained in $ \Delta $ are bound in $ \PT $, \ie the names of $ \Delta $ are exactly the free session channels of $ \PT $.
	Since there are no cyclic dependencies between sessions, we can pick a minimal session in $ \PT $, \ie a session such that its next action is not blocked by any action of another session or session delegation.
	Let $ \Chan $ denote this session channel.
	By the typing rules in the Figures~\ref{fig:typingRules} and \ref{fig:runtimeTypingRules}, there are some $ \GT, \Chan[b] $ such that $ \Typed{\Chan[b]}{\GT} \in \Gamma $ and $ \GT $ specifies the session $ \Chan $.
	Since the next action of this session is not blocked, $ \PT $ contains at least one unguarded prefix or at least one unguarded not empty queue on $ \Chan $.
	Among all such unguarded prefixes and messages that are head of a queue we pick a minimal, \ie one that is typed by the projection of a part of $ \GT $ such that no part of $ \GT $ that is guarding it is used to type another unguarded prefix or message in $ \PT $:
	\begin{description}
		\item[$ \PReq{\Chan[a]}{\Role[n]}{\Chan}{\PT[Q]} $] Then $ \Chan \sharp \Delta $ and $ \Typed{\Chan[a]}{\GT} \in \Gamma $.
			By the assumption on session initialisation, then for all $ 1 \leq \Role_i < \Role[n] $ we have $ \PAcc{\Chan[a]}{\Role_i}{\Chan} $ in $ \PT $.
			Since $ \Chan $ is minimal, all these session acceptances are unguarded.
			By Rule~(\textsf{Init}), then $ \PT \step \PT' $.
		\item[$ \PAcc{\Chan[a]}{\Role}{\Chan}{\PT[Q]} $] Then $ \Chan \sharp \Delta $ and $ \Typed{\Chan[a]}{\GT} \in \Gamma $.
			By the assumption on session initialisation, then $ \PReq{\Chan[a]}{\Role[n]}{\Chan}{\PT[Q]} $ and for all $ 1 \leq \Role_i < \Role[n] $ with $ \Role_i \neq \Role $ we have $ \PAcc{\Chan[a]}{\Role_i}{\Chan} $ in $ \PT $.
			Since $ \Chan $ is minimal, all this session request and these session acceptances are unguarded.
			By Rule~(\textsf{Init}), then $ \PT \step \PT' $.
		\item[$ \PSendR{\Chan}{\Role_1}{\Role_2}{\Expr[y]}{\PT[Q]} $] If $ \Chan \sharp \Delta $ then the typing rules in the Figures~\ref{fig:typingRules} and \ref{fig:runtimeTypingRules} and (\textsf{Req}), (\textsf{Acc}), and (\textsf{Res2}) in particular ensure that $ \MQ{\Chan}{\Role_1}{\Role_2}{\MT} $ in $ \PT $.
			If $ \Chan $ is free in $ \Delta $, then $ \MQ{\Chan}{\Role_1}{\Role_2}{\MT} $ in $ \PT $ because of coherence.
			Since $ \Chan $ and the action are minimal, $ \MQ{\Chan}{\Role_1}{\Role_2}{\MT} $ is unguarded.
			By Rule~(\textsf{RSend}), then $ \PT \step \PT' $.
		\item[$ \MQ{\Chan}{\Role_1}{\Role_2}{\MessR{\Expr[v]}\#\Queue} $] By Lemma~\ref{lem:errorFredom}, then $ \PGetR{\Chan}{\Role_2}{\Role_1}{\Args}{\PT[Q]} $ in $ \PT $.
			Since $ \Chan $ and the action are minimal, $ \PGetR{\Chan}{\Role_2}{\Role_1}{\Args}{\PT[Q]} $ is unguarded.
			By Rule~(\textsf{RGet}), then $ \PT \step \PT' $.
		\item[$ \PGetR{\Chan}{\Role_2}{\Role_1}{\Args}{\PT[Q]} $] By Lemma~\ref{lem:errorFredom}, then $ \MQ{\Chan}{\Role_1}{\Role_2}{\MessR{\Expr[v]}\#\Queue} $ in $ \PT $.
			Since $ \Chan $ and the action are minimal, $ \MQ{\Chan}{\Role_1}{\Role_2}{\MessR{\Expr[v]}\#\Queue} $ is unguarded.
			By Rule~(\textsf{RGet}), then $ \PT \step \PT' $.
		\textcolor{blue}{\item[$ \PSendU{\Chan}{\Role_1}{\Role_2}{\Label}{\Expr[y]}{\PT[Q]} $] If $ \Chan \sharp \Delta $ then the typing rules in the Figures~\ref{fig:typingRules} and \ref{fig:runtimeTypingRules} and (\textsf{Req}), (\textsf{Acc}), and (\textsf{Res2}) in particular ensure that $ \MQ{\Chan}{\Role_1}{\Role_2}{\MT} $ in $ \PT $.
			If $ \Chan $ is free in $ \Delta $, then $ \MQ{\Chan}{\Role_1}{\Role_2}{\MT} $ in $ \PT $, because of coherence.
			Since $ \Chan $ and the action are minimal, $ \MQ{\Chan}{\Role_1}{\Role_2}{\MT} $ is unguarded.
			By Rule~(\textsf{USend}), then $ \PT \step \PT' $.}
		\textcolor{blue}{\item[$ \MQ{\Chan}{\Role_1}{\Role_2}{\MessU{\Label}{\Expr[v]}\#\Queue} $] Then the typing rules and coherence ensure that either there is no actor $ \AT{\Chan}{\Role_2} $, the actor $ \AT{\Chan}{\Role_2} $ skipped, or $ \PGetU{\Chan}{\Role_2}{\Role_1}{\Label}{\Expr[dv]}{\Args}{\PT[Q]} $ in $ \PT $.
			If there is no actor $ \AT{\Chan}{\Role_2} $, then it has crashed.
			Then $ \fpCrash(\PT[Q], \ldots) $ was satisfied with $ \AT{\Chan}{\Role_2} \in \Actors{\PT[Q]} $.
			By Condition~\ref{cond:all}.\ref{cond:fpCrashImpliesML}, then eventually $ \fpML(\Chan, \Role_1, \Role_2, \Label, \ldots) $.
			By Rule~(\textsf{ML}), then $ \PT \step \PT' $.
			If the actor $ \AT{\Chan}{\Role_2} $ skipped this reception, then $ \fpUSkip(\Chan, \Role_2, \Role_1, \Label, \ldots) $.
			By Condition~\ref{cond:all}.\ref{cond:fpMLifffpUSkip}, then $ \fpML(\Chan, \Role_1, \Role_2, \Label, \ldots) $.
			By Rule~(\textsf{ML}), then $ \PT \step \PT' $.
			If $ \PGetU{\Chan}{\Role_2}{\Role_1}{\Label}{\Expr[dv]}{\Args}{\PT[Q]} $ in $ \PT $, since $ \Chan $ and the action are minimal, then this term is unguarded.
			By Condition~\ref{cond:all}.\ref{cond:ugetValid} and Rule~(\textsf{UGet}), then $ \PT \step \PT' $.}
		\textcolor{blue}{\item[$ \PGetU{\Chan}{\Role_2}{\Role_1}{\Label}{\Expr[dv]}{\Args}{\PT[Q]} $] Then the typing rules and coherence ensure that there is the queue $ \MQS{\Chan}{\Role_1}{\Role_2} $ and either $ \MessU{\Label}{\Expr[v]} $ is on top of it, or this message was dropped, or the sender did not yet send this message, or the sender crashed before transmitting this message.
			If $ \MessU{\Label}{\Expr[v]} $ is on top of the queue, then $ \PT \step \PT' $ by Rule~(\textsf{UGet}) and Condition~\ref{cond:all}.\ref{cond:ugetValid}.
			If the message was dropped, then $ \fpML(\Chan, \Role_1, \Role_2, \Label, \ldots) $.
			By Condition~\ref{cond:all}.\ref{cond:fpMLifffpUSkip}, then $ \fpUSkip(\Chan, \Role_2, \Role_1, \Label, \ldots) $.
			By Rule~(\textsf{USkip}), then $ \PT \step \PT' $.
			If the sender did not yet send the message, we proceed as in the Case~$ \PSendU{\Chan}{\Role_1}{\Role_2}{\Label}{\Expr[y]}{\PT[Q]} $ above.
			If the sender crashed, then $ \fpCrash(\PT[Q], \ldots) $ with $ \AT{\Chan}{\Role_1} \in \Actors{\PT[Q]} $.
			By Condition~\ref{cond:all}.\ref{cond:fpCrashImpliesSkip}, then eventually $ \fpUSkip(\Chan, \Role_2, \Role_1, \Label, \ldots) $.
			By Rule~(\textsf{USkip}), then $ \PT \step \PT' $.}
		\item[$ \PSelR{\Chan}{\Role_1}{\Role_2}{\Label}{\PT[Q]} $] If $ \Chan \sharp \Delta $ then the typing rules in the Figures~\ref{fig:typingRules} and \ref{fig:runtimeTypingRules} and (\textsf{Req}), (\textsf{Acc}), and (\textsf{Res2}) in particular ensure that $ \MQ{\Chan}{\Role_1}{\Role_2}{\MT} $ in $ \PT $.
			If $ \Chan $ is free in $ \Delta $, then $ \MQ{\Chan}{\Role_1}{\Role_2}{\MT} $ in $ \PT $, because of coherence.
			Since $ \Chan $ and the action are minimal, $ \MQ{\Chan}{\Role_1}{\Role_2}{\MT} $ is unguarded.
			By Rule~(\textsf{RSel}), then $ \PT \step \PT' $.
		\item[$ \MQ{\Chan}{\Role_1}{\Role_2}{\MessBR{\Label}\#\Queue} $] By Lemma~\ref{lem:errorFredom}, then $ \PBranR{\Chan}{\Role_2}{\Role_1}{\Set{ \Label_i.\PT[Q]_i }_{i \in \indexSet}} $ in $ \PT $.
			Since $ \Chan $ and the action are minimal, $ \PBranR{\Chan}{\Role_2}{\Role_1}{\Set{ \Label_i.\PT[Q]_i }_{i \in \indexSet}} $ is unguarded.
			By (\textsf{RBran}), then $ \PT \step \PT' $.
		\item[$ \PBranR{\Chan}{\Role_2}{\Role_1}{\Set{ \Label_i.\PT[Q]_i }_{i \in \indexSet}} $] By Lemma~\ref{lem:errorFredom}, then $ \MQ{\Chan}{\Role_1}{\Role_2}{\MessBR{\Label}\#\Queue} $ in $ \PT $ with $ j \in \indexSet $ and $ \Label \compL \Label_j $.
			Since $ \Chan $ and the action are minimal, $ \MQ{\Chan}{\Role_1}{\Role_2}{\MessBR{\Label}\#\Queue} $ is unguarded.
			By (\textsf{RBran}), then $ \PT \step \PT' $.
		\textcolor{blue}{\item[$ \PSelW{\Chan}{\Role}{\Role[R]}{\Label}{\PT[Q]} $] If $ \Chan \sharp \Delta $ then the typing rules and (\textsf{Req}), (\textsf{Acc}), and (\textsf{Res2}) in particular ensure that $ \MQ{\Chan}{\Role}{\Role_i}{\MT_i} $ in $ \PT $ for all $ \Role_i \in \Role[R] $.
			If $ \Chan $ is free in $ \Delta $, then $ \MQ{\Chan}{\Role}{\Role_i}{\MT_i} $ in $ \PT $ for all $ \Role_i \in \Role[R] $, because of coherence.
			Since $ \Chan $ and the action are minimal, all $ \MQ{\Chan}{\Role}{\Role_i}{\MT_i} $ are unguarded.
			By Rule~(\textsf{WSel}), then $ \PT \step \PT' $.}
		\textcolor{blue}{\item[$ \MQ{\Chan}{\Role_1}{\Role_2}{\MessBW{\Label}\#\Queue} $] Then the typing rules, coherence, and Condition~\ref{cond:all}.\ref{cond:fpWskip} ensure that either there is no actor $ \AT{\Chan}{\Role_2} $ or $ \PBranW{\Chan}{\Role_2}{\Role_1}{\Set{ \Label_i.\PT[Q]_i }_{i \in \indexSet, \LabelD}} $ in $ \PT $.
			In the former case the actor $ \AT{\Chan}{\Role_2} $ has crashed.
			Then, since the message queue will not be needed any more, proceed with the next session and action that are minimal if you ignore the queue $ \MQS{\Chan}{\Role_1}{\Role_2} $.}
		\textcolor{blue}{\item[$ \PBranW{\Chan}{\Role_2}{\Role_1}{\Set{ \Label_i.\PT[Q]_i }_{i \in \indexSet, \LabelD}} $] Then the typing rules and coherence ensure that there is the queue $ \MQS{\Chan}{\Role_1}{\Role_2} $ and either $ \MessBW{\Label} $ is on top of it, or the sender did not yet send this message, or the sender crashed before transmitting this message.
			If $ \MessBW{\Label} $ is on top of the queue, then $ \PT \step \PT' $ by Rule~(\textsf{WBran}).
			If the sender did not yet send the message, we proceed as in the Case~$ \PSelW{\Chan}{\Role}{\Role[R]}{\Label}{\PT[Q]} $ above.
			If the sender crashed, then $ \fpCrash(\PT[Q]', \ldots) $ with $ \AT{\Chan}{\Role_1} \in \PT[Q]' $.
			By Condition~\ref{cond:all}.\ref{cond:fpCrashImpliesSkip}, then eventually $ \fpWSkip(\Chan, \Role_2, \Role_1, \Label, \ldots) $ for $ \AT{\Chan}{\Role_2} $.
			By Rule~(\textsf{WSkip}), then $ \PT \step \PT' $.}
		\item[$ \PDelA{\Chan}{\Role_1}{\Role_2}{\AT{\Chan'}{\Role}}{\PT[Q]} $] If $ \Chan \sharp \Delta $ then the typing rules in the Figures~\ref{fig:typingRules} and \ref{fig:runtimeTypingRules} and (\textsf{Req}), (\textsf{Acc}), and (\textsf{Res2}) in particular ensure that $ \MQ{\Chan}{\Role_1}{\Role_2}{\MT} $ in $ \PT $.
			If $ \Chan $ is free in $ \Delta $, then $ \MQ{\Chan}{\Role_1}{\Role_2}{\MT} $ in $ \PT $ because of coherence.
			Since $ \Chan $ and the action are minimal, $ \MQ{\Chan}{\Role_1}{\Role_2}{\MT} $ is unguarded.
			By Rule~(\textsf{Deleg}), then $ \PT \step \PT' $.
		\item[$ \MQ{\Chan}{\Role_1}{\Role_2}{\AT{\Chan'}{\Role}\#\Queue} $] By Lemma~\ref{lem:errorFredom}, then $ \PDelB{\Chan}{\Role_2}{\Role_1}{\AT{\Chan''}{\Role'}}{\PT[Q]} $ in $ \PT $.
			Since $ \Chan $ and the action are minimal, $ \PDelB{\Chan}{\Role_2}{\Role_1}{\AT{\Chan''}{\Role'}}{\PT[Q]} $ is unguarded.
			By Rule~(\textsf{SRecv}), then $ \PT \step \PT' $.
		\item[$ \PDelB{\Chan}{\Role_2}{\Role_1}{\AT{\Chan''}{\Role'}}{\PT[Q]} $] By Lemma~\ref{lem:errorFredom}, then $ \MQ{\Chan}{\Role_1}{\Role_2}{\AT{\Chan'}{\Role}\#\Queue} $ in $ \PT $.
			Since $ \Chan $ and the action are minimal, $ \MQ{\Chan}{\Role_1}{\Role_2}{\AT{\Chan'}{\Role}\#\Queue} $ is unguarded.
			By Rule~(\textsf{SRecv}), then $ \PT \step \PT' $.
	\end{description}
	2.) By Theorem~\ref{thm:subjectReduction}, $ \PT' $ is well-typed \wrt coherent $ \Gamma, \Delta' $ with $ \Delta \stackrel{\Chan}{\Mapsto} \Delta' $.
	If $ \PT $ does not contain recursion, then the session environment strictly reduces with every reduction of the process that does not reduce a conditional, though it may grow in cases of session initialisation.
	Since $ \PT $ is finite and loop-free there are only finitely many possible session initialisations.
	We can repeat the above proof for 1.) to show that $ \PT $ cannot get stuck as long as it contains action prefixes.
\end{proof}

The proof of progress relies on the Condition~\ref{cond:all}.\ref{cond:ugetValid}--\ref{cond:all}.\ref{cond:fpWskip} to ensure that failures cannot block the system: in the failure-free case \unrel messages are eventually received (\ref{cond:all}.\ref{cond:ugetValid}), the receiver of a lost message can skip (\ref{cond:all}.\ref{cond:fpMLifffpUSkip}), no receiver is blocked by a crashed sender (\ref{cond:all}.\ref{cond:fpCrashImpliesSkip}), messages towards receivers that crashed or skipped can be dropped (\ref{cond:all}.\ref{cond:fpCrashImpliesML} + \ref{cond:all}.\ref{cond:fpMLifffpUSkip}), and branching requests cannot be ignored (\ref{cond:all}.\ref{cond:fpWskip}).

%%%%%%%%%%%%%%%%%%%%%%%%%%%%%%%%%%%%%%
%  Verifying Distributed Algorithms  %
%%%%%%%%%%%%%%%%%%%%%%%%%%%%%%%%%%%%%%

\section{The Rotating Coordinator Algorithm}
\label{sec:example}

To illustrate our approach we study a Consensus algorithm by Chandra and Toueg (\cf \cite{ChandraToueg96,nestmann07}).
This algorithm is small but not trivial.
It was designed for systems with crash failures, but the majority of the algorithm can be implemented with \unrel communication.

As this algorithm models consensus, the goal is that every agent $i$ eventually decides on a proposed belief value, where no two agents decide on different values.
It is a round based algorithm, where each round consists of four phases.
In each round, one process acts as a coordinator decided by round robin, denoted by $ \Role[c] $.
\begin{compactdesc}
	\item[In Phase~1] every agent $ i $ sends its current belief to the coordinator $ \Role[c] $.
	\item[In Phase~2] the coordinator waits until it has received at least half of the messages of the current round and then sends the best belief to all other agents.
	\item[In Phase~3] the agents either receive the message of the coordinator or suspect the coordinator to have crashed and reply with ack or nack accordingly. Suspicion can yield false positives.
	\item[In Phase~4] the coordinator waits, as in Phase~2, until it has received at least half of the messages of the current round and then sends a \weakR broadcast if at least half of the messages contained ack.
\end{compactdesc}

It is possible for agents to skip rounds by suspecting the coordinator of the current round and by proceeding to the next round.
There are also no synchronisation fences thus it is possible for the agents to be in different rounds and have messages of different rounds in the system.
Having agents in different rounds makes proving correctness much more difficult.

\subsection{Specification and Implementation}

Let $ {\left( \bigodot_{1 \leq i \leq n} \pi_i \right)}.\GT $ abbreviate $ \pi_1.\ldots.\pi_n.\GT $ to simplify the presentation, where $ \GT \in \globalTypes $ is a global type and $ \pi_1, \ldots, \pi_n $ are sequences of prefixes. More precisely, each $ \pi_i $ is of the form $ \pi_{i, 1}.\ldots.\pi_{i, m} $ and each $ \pi_{i, j} $ is a type prefix of the form $ \GComUS{\Role_1}{\Role_2}{\Label}{\Sort} $ or $ \GBranW{\Role}{\Role[R]}{\Label_1.\LT_1 \oplus \ldots \oplus \Label_n.\LT_n \oplus \LabelD} $, where the latter case represents a \weakR branching prefix with the branches $ \Label_1, \ldots, \Label_n, \LabelD $, the default branch $ \LabelD $, and where the next global type provides the missing specification for the default case.

We assume the sorts $\SortBel = \Set{0,1}$ and $ \SortAck = \Set{\true, \false} $.
Let $ \Role[n] $ be the number of agents.
We start with the specification of the algorithm as a global type.
\begin{align*}
	\GT_{rc}{\left( \Role[n] \right)} \deff{} & \GRep{\TypeV}{}\bigodot_{1 \leq \Role[c] \leq \Role[n]}\Big(
		{\Big( \bigodot_{1 \leq \Role[i] \leq \Role[n], \Role[i] \neq \Role[c]} \GComUS{\Role[i]}{\Role[c]}{\Label[p]_1}{\SortBel} \Big)}.
		{\Big( \bigodot_{1 \leq \Role[i] \leq \Role[n], \Role[i] \neq \Role[c]} \GComUS{\Role[c]}{\Role[i]}{\Label[p]_2}{\SortBel} \Big)}.\\
		& \hspace{5em} {\Big( \bigodot_{1 \leq \Role[i] \leq \Role[n], \Role[i] \neq \Role[c]} \GComUS{\Role[i]}{\Role[c]}{\Label[p]_3}{\SortAck} \Big)}.\\
		& \hspace{5.5em} \GBranW{\Role[c]}{\Set{\Role[i] | 1 \leq \Role[i] \leq \Role[n], \Role[i] \neq \Role[c]}}{\Label[Zero].\GEnd \oplus \Label[One].\GEnd \oplus \LabelD}
	\Big).\TypeV
\end{align*}
It specifies a loop containing a collection of $ \Role[n] $ rounds, where each process functions as a coordinator once.
This collection of $ \Role[n] $ rounds is specified with the first $ \bigodot $, \ie the continuation of $ \LabelD $ in the end of the description is the specification of round $ \Role + \Role[1] $ for all rounds $ \Role < \Role[n] $ whereas in the last round $ \Role[n] $ we have $ \LabelD.\TypeV $.
By unfolding the recursion on $ \TypeV $, $ \GT_{rc}{\left( \Role[n] \right)} $ starts the next $ \Role[n] $ rounds.
The following three $ \bigodot $ specify the Phases~1--3 of the algorithm within one round.
Phase~4 is specified by a \weakR branching that does not need a $ \bigodot $, since it is a broadcast.

In Phase~1 all processes except the coordinator $ \Role[c] $ transmit a belief to $ \Role[c] $ using label $ \Label[p]_1 $.
In Phase~2 $ \Role[c] $ transmits a belief to all other processes using label $ \Label[p]_2 $.
Then all processes transmit a value of type $ \SortAck $ to the coordinator using label $ \Label[p]_3 $ in Phase~3.
Finally, in Phase~4 the coordinator broadcasts one of the labels $ \Label[Zero] $, $ \Label[One] $, or $ \LabelD $, where the first two labels represent a decision and terminate the protocol, whereas the default label $ \LabelD $ specifies the need for another round.
All interactions in the specification are \unrel or \weakR.

Let $ \left( \bigodot_{1 \leq i \leq n} \pi_i \right).\PT $ abbreviate the sequence $ \pi_1.\ldots.\pi_n.\PT $, where $ \PT \in \processes $ is a process and $ \pi_1, \ldots, \pi_n $ are sequences of prefixes.
\begin{align*}
	& \RCSys \deff \PReq{\Chan[a]}{\Role[n]}{\Chan}{\RCP[n]{\InitKnowledge[n]}} \mid \prod_{1 \leq \Role[i] < \Role[n]} \PAcc{\Chan[a]}{\Role[i]}{\Chan}{\RCP[i]{\InitKnowledge[i]}}\\
	& \RCP{\Knowledge} \! \deff \PRep{\ProcV}{{\Big( \bigodot_{1 \leq \Role[c] \leq \Role[n]} \PITE{\Role[i] = \Role[c]}{\RCPioS{i}}{\RCPio[NC]{i}}\!\!\Big)}. \ProcV}
\end{align*}
$\RCSys$ describes the session initialisation of a system with $\Role[n]$ participants and the initial knowledge \InitKnowledge, where \InitKnowledge[i] is a vector that contains only the initial knowledge of role $\Role[i]$.
Let $ \Length{\Knowledge} \deff \Length{\Set{i \mid v_i \neq \bot }} $ return the number of non-empty entries.
$\RCP{\Knowledge}$ describes a process $\Role[i]$ in a set of $\Role[n]$ processes.
Each process then runs for at most $\Role[n]$ rounds and then loops.
\begin{align*}
	\RCPioS{c} \deff & \Big( \bigodot_{1 \leq \Role[i] \leq \Role[n], \Role[i] \neq \Role[c]} \PGetUS{\Chan}{\Role[c]}{\Role[i]}{\Label[p]_1}{\Args[\bot]}{\Args[v]_i} \Big).\\
	& \myif \; \Length{\Knowledge} \geq \ceil{\frac{\Role[n]-1}{2}}
	\; \mythen \; \RCPiioS{c}{ \FuncBest(\Knowledge)} \; \myelse \; \RCPiioS{c}{\InitKnowledge[c]}\\
	\RCPio[NC]{i} \deff & \PSendU{\Chan}{\Role[i]}{\Role[c]}{\Label[p]_1}{\Args[v]_i}{\RCPiio[NC]{i}{\Knowledge}}
\end{align*}
Every non-coordinator $ \RCPio[NC]{i} $ sends its own belief via unreliable communication to the coordinator and proceeds to Phase~2.
The coordinator receives (some of) these messages and proceeds to Phase~2.
If the reception of at least half of the messages was successful, it is updating its belief using the function $ \FuncBest() $ that replaces all belief values with the best one.
Otherwise, it discards all beliefs except its own.
We are using $\big\lceil\frac{\Role[n]-1}{2}\big\rceil$ to check for a majority, since in our implementation processes do not transmit to themselves.
\begin{align*}
	\RCPiioS{c}{\Knowledge} \deff & \Big( \bigodot_{1 \leq \Role[i] \leq \Role[n], \Role[i] \neq \Role[c]} \PSendUS{\Chan}{\Role[c]}{\Role[i]}{\Label[p]_2}{\Args[v]_i} \Big).\RCPiiioS{c}\\
	\RCPiio[NC]{i}{\Knowledge} \deff & \PGetUS{\Chan}{\Role[i]}{\Role[c]}{\Label[p]_2}{\Args[\bot]}{\Args[x]}.\\
	& \myif \; \Args[x] = \Args[\bot] \; \mythen \; \RCPiiin{i}{\Knowledge}{\false}
	\; \myelse \; \RCPiiin{i}{\FuncUpdate(\Knowledge, \Role[i], \Args[x])}{\true}
\end{align*}
In Phase~2, the coordinator sends its updated belief to all other processes via unreliable communication and proceeds.
Note that, $ \Args[v]_i $ is either $ \bot $ for all $ \Role[i] \neq \Role[c] $ or the best belief identified in Phase~1.
If a non-coordinator process successfully receives a belief other than $ \bot $, it updates its own belief with the received value and proceeds to Phase~3, where we use the Boolean value $ \true $ for the acknowledgement.
If the coordinator is suspected to have crashed or $ \bot $ was received, the process proceeds to Phase~3 with the Boolean value $ \false $, signalling nack.
\begin{align*}
	\RCPiiioS{c} \deff & \Big( \bigodot_{1 \leq \Role[i] \leq \Role[n], \Role[i] \neq \Role[c]} \PGetUS{\Chan}{\Role[c]}{\Role[i]}{\Label[p]_3}{\Args[\bot]}{\Args[v]_i} \Big).\RCPivoS{c}\\
	\RCPiiin{i}{\Knowledge}{\Args[b]} \deff & \PSendUS{\Chan}{\Role[i]}{\Role[c]}{\Label[p]_3}{\Args[b]}.\RCPivo[NC]{i}
\end{align*}
In Phase~3, every non-coordinator sends either ack or nack to the coordinator.
If the coordinator successfully receives the message, it writes the Boolean value at the index of the sender into its knowledge vector.
In case of failure, $\bot$ is used as default.
After that the processes continue with Phase 4.
Let $ \mathcal{I} = \Set{\Role[i] \mid 1 \leq \Role[i] \leq \Role[n], \Role[i] \neq \Role[c]} $.
\begin{align*}
	\RCPivoS{c} \deff & \myif\; \FuncCountAck(\Knowledge) \geq \ceil{\frac{\Role[n]-1}{2}} \; \mythen \; (\myif \; \Args[v]_{\Role[c]} = 0 \; \mythen \; \PSelW{\Chan}{\Role[c]}{\mathcal{I}}{\Label[Zero]}{\PEnd}\\
	& \myelse \; \PSelW{\Chan}{\Role[c]}{\mathcal{I}}{\Label[One]}{\PEnd}) \; \myelse \; \PSelWS{\Chan}{\Role[c]}{\mathcal{I}}{\LabelD}\\
	\RCPivo[NC]{i} \deff & \PBranW{\Chan}{\Role[i]}{\Role[n]}{\Label[Zero].\PEnd \oplus \Label[One].\PEnd \oplus \LabelD}
\end{align*}
In Phase 4 the coordinator checks if at least half of the non-coordinator roles signalled acknowledgement, utilising the function \FuncCountAck{} to count.
If it received enough acknowledgements it transmits the decision via $ \Label[Zero] $ or $ \Label[One] $ and causes all participants to terminate.
Otherwise, the coordinator sends the default label and continues with the next round.
Remember that the missing continuation after the default label $ \LabelD $ for coordinators and non-coordinators is implemented by the next round.

We use a \weakR branching mechanism in conjunction with \unrel communication.
The algorithm was modelled for systems with crash failures but without message loss.
However, as long as the branching mechanism (\ie the specified broadcast of decisions) is \weakR, we can relax the system requirements for the remainder of the algorithm.
To ensure termination, however, we have to further restrict the number of lost messages.

\subsection{Failure Patterns}

Chandra and Toueg introduce in \cite{ChandraToueg96} also the failure detector \FDS.
The failure detector \FDS is called \emph{eventually strong}, meaning that
(1) eventually every process that crashes is permanently suspected by every correct process and
(2) there is a time after which some correct process is never suspected by any other process.
We observe that the suspicion of senders is only possible in Phase~3, where processes may suspect the coordinator of the round.
Accordingly, the failure pattern $ \fpUSkip $ implements this failure detector to allow processes to suspect unreliable coordinators in Phase~2, \ie with label~$ \Label[p]_2 $.
In Phase~1 and Phase~3 $ \fpUSkip $ may allow to suspect processes that are not crashed after the coordinator received enough messages. In all other cases these two patterns eventually return true iff the respective sender is crashed.
$ \fpUGet $ can be used to reject outdated messages, since this is not important for this algorithm we implement it with the constant true.
To ensure that messages of wrongly suspected coordinators in Phase~2 do not block the system, $ \fpML $ is eventually true for messages with label~$ \Label[p]_2 $ that were suspected using \FDS or skipped $ \Label[p]_1 $/$ \Label[p]_3 $-messages and otherwise returns false.
By the system requirements in \cite{ChandraToueg96}, no messages get lost, but it is realistic to assume that receivers can drop messages of skipped receptions on their incoming message queues.
As there are at least half of the processes required to be correct for this algorithm, we implement $ \fpCrash $ by false if only half of the processes are alive and true otherwise.
For the \weakR broadcast, $ \fpWSkip $ returns true if and only if the respective coordinator is crashed, \ie not suspected but indeed crashed.
In \cite{ChandraToueg96} this broadcast, which is called just reliable in \cite{ChandraToueg96}, is used to announce the decision.
Since we use it for branching even before a decision was reached, our implementation is less efficient compared to \cite{ChandraToueg96}. We briefly discuss in the conclusions how to regain the original algorithm.
These failure patterns satisfy the Condition~\ref{cond:all}.\ref{cond:crash}--\ref{cond:all}.\ref{cond:fpWskip}.

\subsection{Termination, Agreement, and Validity}

Following \cite{Tel94,Lynch96}, a network of processes solves Consensus if
\begin{compactdesc}
	\item[Termination] all non-failing participants eventually decide,
	\item[Agreement] all decision values are the same, and
	\item[Validity] each decision value is an initial value of some participant.
\end{compactdesc}
By \cite{nestmann07} the Rotating Coordinator algorithm solves Consensus in the presence of crash-failures.

In the proof of termination well-typedness of the implementation and progress in Theorem~\ref{thm:progress} are the main ingredients.
Well-typedness ensures that the implementation follows its specification and progress ensures that it cannot get stuck.
Apart from that, we only need to deduce from the system requirements, \ie the used failure detector, that the implementation will eventually exit the loop.

\begin{lem}
	\label{lem:termination}
	The implementation of the rotating coordinator example satisfies termination.
\end{lem}

\begin{proof}
  The global type ensures progress of the system, \ie it either loops forever or terminates after a broadcast with $ \Label[Zero] $ or $ \Label[One] $.
  Following the assumption of \cite{ChandraToueg96}, there is eventually a round in which the coordinator is not suspected by any correct process and in which the coordinator does not suspect correct processes (or at least trusts a majority of them).
  Then the coordinator receives enough beliefs in Phase~1, transmits the best belief in Phase~2, receives enough acknowledgements in Phase~3, and transmits $ \Label[Zero] $ or $ \Label[One] $ in Phase~4. Finally, all correct processes receive the decision and terminate.
  Well-typedness and progress in Theorem~\ref{thm:progress} ensure that no process deadlocks along this way.
\end{proof}

Validity can be checked by analysing the code of a single process and check whether it uses (to decide and to transmit) only its own initial belief or the belief it received from others. Session fidelity in Theorem~\ref{thm:progress} then ensures that there are no communication mismatches and all steps preserve this property.

\begin{lem}
	\label{lem:validity}
	The implementation of the rotating coordinator example satisfies validity.
\end{lem}

\begin{proof}
	Well-typedness and the global type ensure that the communication structure is as specified.
	By the implementation, in Phase~1 every non-coordinator sends its own belief to the coordinator on label $\Label[p_{1}]$ of which the coordinator selects the best value using the function \FuncBest, thus altering its own belief with a valid value.
	By assumption, \FuncBest{} does not modify any received belief value but only chooses one among them.
	Label $\Label[p_{2}]$ is used to propagate the best belief of the coordinator of this round.
	On successful reception of this message, the non-coordinators update their own belief with the received value, which is a valid value.
	$ \Label[Zero] $ and $ \Label[One] $ are chosen matching to the current belief of the coordinator, \ie the local belief is sent and updated on reception without alternation.
	As there are no more labels left on which belief values are sent, we conclude, that validity holds for each local process.
Because of the global type and the communication protocol it ensures following session fidelity in Theorem~\ref{thm:progress}---there are no communication mismatches and the received messages are of the expected sort---, then validity holds globally.
\end{proof}

Agreement follows from the \weakR broadcast, because the decision is broadcasted.

\begin{lem}
	\label{lem:ageement}
  	The implementation of the rotating coordinator example satisfies agreement.
\end{lem}

\begin{proof}
	The only way to decide is via a broadcast.
	Let us fix a coordinator $ \Role[c] $ of such a deciding round.
	To decide and broadcast a value, the coordinator $ \Role[c] $ needs a majority of processes in its round.
	Messages $ \MessBW{\Label[Zero]} $ and $ \MessBW{\Label[One]} $ cannot be dropped and the failure pattern $ \fpWSkip $ does not allow to suspect a sender $ \Role[c] $ of \weakR branching that is not crashed. Hence, a majority of processes receive the decision.
	Accordingly all decision values originate from the same broadcast.
	This ensures agreement.
\end{proof}

Combining the three above lemmata, we conclude that our implementation of the algorithm solves Consensus.

\begin{thm}
	\label{thm:algo}
	The algorithm solves Consensus, \ie satisfies termination, validity, and agreement.
\end{thm}

\begin{proof}
	By the Lemmata~\ref{lem:termination}, \ref{lem:validity}, and \ref{lem:ageement}.
\end{proof}

%%%%%%%%%%%%%%%%%
%  Conclusions  %
%%%%%%%%%%%%%%%%%

\section{Conclusions}
\label{sec:conclusions}

We present a fault-tolerant variant of \MPST for systems that may suffer from message loss or crash failures.
We implement \unrel communication and \weakR branching.
The failure patterns in the semantics allow to verify algorithms modulo system requirements.
We prove subject reduction and progress and present a small but relevant case study.

An open question is how to conveniently type unreliable recursive parts of protocols.
Distributed algorithms are often recursive and exit this recursion if a result was successfully computed.
%How do we ensure that processes branch---loop or exit---consistently in the case of unreliable communication?
We present a first attempt to solve this problem using a \weakR branching.
In further research we want to analyse, whether or in how far branching can be extended to the case of message loss.

Indeed our implementation of the rotating coordinator algorithm is not ideal.
It implements the decision making procedure correctly and also allows processes to be in different rounds at the same time. So, it represents a non-trivial variant of the rotating coordinator algorithm.
But it does not allow the processes to diverge in their rounds as freely as the original rotating coordinator algorithm, because the \weakR branching implementation implies that the coordinator has always to be the first process to leave a round.
We can solve this problem by wrapping each round in an \unrel sub-session (\eg an \unrel variant of the sub-sessions introduced in \cite{DemangeonHonda12,Demangeon15}). If we allow processes to skip such an \unrel sub-session altogether, we obtain the intended behaviour.

We considered \strongR session delegation.
Next we want to study whether and in how far we can introduce \weakR or \unrel session delegation.
Similarly, we want to study \unrel variants of session initialisation including process crashes and lost messages during session initialisation.
Moreover, \unrel variants of session initialisation open a new perspective on \MPST-frameworks such as \cite{char16} with dynamically changing network topologies and sessions for that the number of roles is determined at run-time.

In Section~\ref{sec:typing} we fix one set of conditions on failure patterns to prove subject reduction, session fidelity, and progress.
We can also think of other sets of conditions.
The failure pattern $ \fpUGet $ can be used to reject the reception of outdated messages.
Therefore, we drop Condition~\ref{cond:all}.\ref{cond:ugetValid} and instead require for each message $ m $ whose reception is refused that $ \fpML $ ensures that $ m $ is eventually dropped from the respective queue and that $ \fpUSkip $ allows to skip the reception of these messages.
An interesting question is to find minimal requirements and minimal sets of conditions that allow to prove correctness in general.

%The type check can be automated easily (at least for a given global type).
It would be nice to also fully automate the remaining proofs for the distributed algorithm in Section~\ref{sec:example}. The approach in \cite{petersWagnerNestmann19} sequentialises well-typed systems and gives the much simpler remaining verification problem to a model checker. Interestingly, the main challenges to adopt this approach are not the \unrel or \weakR prefixes but the failure patterns.

%%%%%%%%%%%%%%%%%%
%  Bibliography  %
%%%%%%%%%%%%%%%%%%

\bibliographystyle{alphaurl}
\bibliography{FTMPST}

\end{document}